\documentclass[11pt,reqno]{amsart}

\usepackage[T1]{fontenc}
\usepackage[utf8]{inputenc}
\usepackage{times}
\usepackage{amsmath,amssymb,amsfonts,amsthm,mathtools}
\usepackage{mathrsfs}
\usepackage{bm}
\usepackage[margin=1.1in,includefoot,headheight=0pt,headsep=0pt,footskip=18pt]{geometry}
\usepackage[colorlinks=true,citecolor=blue,linkcolor=blue,urlcolor=blue]{hyperref}

\renewcommand{\enlargethispage}[1]{}
\date{}
\makeatletter
\def\ps@plain{\ps@empty
  \def\@oddfoot{\normalfont\small\hfil\thepage\hfil}
  \let\@evenfoot\@oddfoot}
\let\ps@firstpage\ps@plain
\makeatother
\AtBeginDocument{\pagestyle{plain}}
\allowdisplaybreaks
\newtheorem{theorem}{Theorem}[section]
\newtheorem{lemma}[theorem]{Lemma}
\newtheorem{proposition}[theorem]{Proposition}
\newtheorem{corollary}[theorem]{Corollary}

\theoremstyle{definition}

\theoremstyle{remark}
\newtheorem{remark}[theorem]{Remark}

\renewcommand{\leq}{\leqslant}

\renewcommand{\geq}{\geqslant}

\newcommand{\F}{\mathbb F}
\newcommand{\eqdef}{\triangleq}
\newcommand{\T}{\intercal}
\newcommand{\bch}{\mathrm{BCH}}
\newcommand{\bfs}{\bm{s}}
\newcommand{\bft}{\bm{t}}
\newcommand{\bfu}{\bm{u}}
\DeclareMathOperator{\Gal}{Gal}
\DeclareMathOperator{\Span}{span}

\title[Generalized covering radii of BCH codes]{Asymptotically Tight Bounds for
  Generalized Covering Radii of Binary Primitive BCH Codes at All Higher Orders}

\author{Zeev Vladimir Belinsky}
\address{Faculty of Civil and Environmental Engineering, Technion --
Israel Institute of Technology, Haifa 32000, Israel}
\email{z.belinsky@campus.technion.ac.il}
\author{Aryeh Lev Zabokritskiy (Yohananov)}
\address{Department of Computer Science, Tel-Hai University of Kiryat
Shmona in the Galilee, Israel}
\address{MIGAL -- Galilee Research Institute, P.O. Box 831,
Kiryat Shmona 1101602, Israel}
\email{yuhanalev@telhai.ac.il}

\hypersetup{
  pdflang={en},
  pdftitle={Asymptotically Tight Bounds for Generalized Covering Radii of Binary Primitive BCH Codes at All Higher Orders},
  pdfauthor={Zeev Vladimir Belinsky and Aryeh Lev Zabokritskiy (Yohananov)},
  pdfsubject={Coding theory and generalized covering radii of primitive BCH codes},
  pdfkeywords={generalized covering radius, BCH codes, algebraic varieties over finite fields, monodromy, generalized Hamming weights}
}

\begin{document}

\begin{abstract}
We study the generalized covering radii of binary primitive BCH codes,
which measure how many parity-check columns suffice to span several prescribed
syndromes.  For the four-error-correcting family of length \(2^m-1\), the
second radius is exactly \(11\) for \(m\geq55\), and it is either \(11\) or
\(12\) for \(m\geq16\).  For every fixed error parameter at least two, we
give an explicit stable bound with two adjacent possible values for the second
radius and an arithmetic criterion for exactness.  At every higher order,
we obtain explicit lower and upper bounds whose additive gap is bounded
independently of the order for each fixed error parameter, once explicit
field-size conditions hold.  The bounds are therefore asymptotically tight
as the order grows.  The higher-order upper bound retains the common-core
support count established by Xiong, Yip, and Zullo; our refinement reduces
its sufficient field-size threshold at large order.
The refinement combines componentwise mixed degrees with the ordinary
degree of a multicone.  The upper bounds come from constructing syndrome
representations with a common set of parity-check columns.
\end{abstract}

\maketitle

\section{Introduction}
\label{sec:introduction}

The generalized covering radii of a linear code ask how economically several
syndromes can share parity-check columns.  More precisely, \(R_t(C)\) is the
worst-case minimum size of one column set whose binary span contains \(t\)
prescribed syndromes.  These radii were introduced
in~\cite{ElimelechFirerSchwartz2021}, where equivalent coding, geometric, and
scalar-extension formulations were established.  The ordinary covering
radius is the first member of the hierarchy.  A complementary
finite-geometric formulation in terms of \((\rho,t)\)-saturating sets appears
in~\cite{AlfaranoMarinoNeriTrombetti2026}.

For integers \(e\geq1\) and \(m\geq2\), let
\[
 C_{e,m}\eqdef\bch(e,m)
\]
be the binary primitive narrow-sense BCH code of length \(2^m-1\) and designed
distance \(2e+1\).  Thus \(e\) is the error-correction parameter, whereas the
subscript in \(R_t\) is the generalized-covering order.  Over
\(F=\F_{2^m}\), the parity-check column indexed by \(x\in F^*\) is
\[
 h_e(x)\eqdef(x,x^3,\ldots,x^{2e-1})^\T.
\]
Consequently, the problem is to span a binary syndrome space by as few of
these columns as possible.

Our earliest effective second-order bound concerns four errors:
\[
 11\leq R_2(C_{4,m})\leq12\quad(m\geq16),
 \qquad
 R_2(C_{4,m})=11\quad(m\geq55).
\]
For every fixed \(e\geq2\), the same common-core method gives an explicit
stable bracket \(3e-1\leq R_2(C_{e,m})\leq3e\) and a finite-field criterion
for the lower value.  The exact second radius was already known for two and
three errors~\cite{YohananovSchwartz2025,EssayagZabokritskiy2026BCH3}; the new
second-order content begins at four errors.

A complementary precursor introduced a weak second generalized radius for
the triple-error-correcting family and proved the upper bound \(9\) for odd
\(m\geq11\) and even \(m\geq20\)
\cite[Theorem~5.2]{OzbudakOzturk2026Second}.  That result does not determine
the ordinary second generalized radius.

The dedicated second-order framework, including the four-error bounds and
the general two-value bracket with its arithmetic exactness criterion, was obtained by A.L.Z. and
submitted separately before preparation of this unified higher-order version.
It is included here, together with its proofs, so that the present
manuscript is self-contained.  The unconditional \(e=5,6\) thresholds
\(m\geq81,108\) are sharper
consequences of the low-error sign calculation developed for the unified
paper.  The \(e=7,8\) root-free exact instances remain part of
the separately submitted second-order analysis.

The higher-order construction seeks representations of \(t\) targets using
\(e-1\) common locators and \(e\) further locators per target, for a total of
\[
 (e-1)+te=(t+1)e-1.
\]
There are two ways to obtain these simultaneous representations.  A selected
component at infinity handles arbitrary targets directly.  A separate
completion-cover construction analyzes the completion fields and may first
regularize a syndrome space using extra shared columns.  We sharpen the
field threshold for the first route and retain the second for its arithmetic
and structural conclusions.

To state the bounds, put
\[
 k_e\eqdef2e-3,
 \qquad
 D_e\eqdef\prod_{j=1}^{e}(2j-1),
 \qquad
 g_{e,m}\eqdef\gcd(k_e,2^m-1),
\]
and, for \(t\geq2\), set
\[
 B_{e,t}\eqdef\min\left\{D_e^t,
 (e!)^t\left(1+\frac{(e-1)D_e}{e!}t\right)^{e-1}\right\}.
\]
The number \(B_{e,t}\) is the earlier completion-cover degree budget;
its two terms come from two degree estimates.  The integer
\(g_{e,m}\) is the size of every nonempty fiber of
\(\gamma\mapsto\gamma^{k_e}\) on \(F^*\).
The second-order arithmetic criterion uses the polynomial
\[
 G_e(x)\eqdef
 \gcd_{\F_2[x]}\bigl(x^{k_e}+1,(x+1)^{k_e}+1\bigr).
\]
In particular, once the explicit field-size condition holds, the absence of
a root of \(G_e\) in \(\F_{2^m}\) forces the exact value \(3e-1\).  Finally,
define an auxiliary integer \(\ell_{e,m,t}\) and the additional support
cost \(c_{e,m,t}\) by
\[
 \ell_{e,m,t}\eqdef
 \left\lceil\log_2\left(1+\frac{2^t-1}{g_{e,m}}\right)\right\rceil,
 \qquad
 c_{e,m,t}\eqdef
 \left\lceil\log_2\bigl(t-\ell_{e,m,t}+1\bigr)\right\rceil.
\]

For the third generalized radius, {\"O}zbudak and {\"O}zt{\"u}rk
\cite{OzbudakOzturk2026Third} determined the even-degree value for the
double-error-correcting family and left two adjacent possibilities in odd
degree.  Xiong and Yip~\cite{XiongYip2026} subsequently obtained a general
supercode lower bound and stable bounds for every generalized order in that
family.  Their upper-bound construction uses one shared locator and two
additional locators per target, a precursor of the common core used here.
Their estimates are sharper when \(e=2\).  The uniform theorem below is
included in that case only for completeness, not as an improvement to their
specialized bound.

After the first version appeared, Xiong, Yip, and
Zullo~\cite{XiongYipZullo2026} proved
\(R_t(C_{e,m})\leq(t+1)e-1\) for every fixed \(e,t\) and sufficiently
large \(m\).  Their argument uses a smooth rational point at infinity in
the same incidence system, and removes the correction for general \(e\).
In particular, it gives \(R_2(C_{e,m})=3e-1\) throughout the stable range
without the arithmetic restriction.  Their sufficient field threshold uses
the ordinary degree budget \(D_e^t\).  Our refinement applies to their
selected component without assuming that all completion covers are connected
or that their sign fields are independent.

Put \(a=e-1\), \(E=e!\), and define
\begin{equation}
\label{eq:intro-componentwise-budget}
\begin{split}
 P_e(w)&\eqdef E+D_e\sum_{u=1}^{a-1}w^u+Ew^a,\\
 \Gamma_{e,t}&\eqdef\sum_{v=0}^{a}[w^v]P_e(w)^t,
 \qquad L_{e,t}\eqdef\min\{D_e^t,\Gamma_{e,t}\}.
\end{split}
\end{equation}
Here \([w^v]\) denotes coefficient extraction, and the interior sum is
empty when \(e=2\).  The budget satisfies \(L_{e,t}\leq B_{e,t}\).

\begin{theorem}[Componentwise common-core bound]
\label{thm:intro-componentwise}
Let \(e,t\geq2\), and put \(q=2^m\).  If
\[
 q>2L_{e,t}^{\,4},
\]
then
\[
 R_t(C_{e,m})\leq(t+1)e-1.
\]
More precisely, every ordered list of \(t\) targets in \(\F_q^e\)
admits a representation with \(e-1\) common locator slots and \(e\)
completing slots per target.  Zero and repeated targets and locator slots
are allowed; no independence or admissibility assumption is imposed.
\end{theorem}

The support count is the one established by Xiong--Yip--Zullo; the improvement is a
smaller sufficient field threshold, not a smaller support count.  For
\(e=8,t=16\), ordinary B\'ezout requires \(m\geq1342\), the earlier
exact mixed-degree sum gives \(m\geq1321\), and
Theorem~\ref{thm:intro-componentwise} requires only \(m\geq1193\).
These are sufficient thresholds, not assertions of minimality.
The earlier finite-field twelve-column bound and the pure ten-column
construction remain separate contributions, as does the structural analysis
of the completion fields.

The equations for one target, with the even moments determined by Frobenius
adjoined, belong to the framework of moment subset sums studied by Lai et
al.~\cite[Theorem~1]{LaiMarinoRobinsonWan2020} and to the diagonal systems
studied by Gottig et al.~\cite[Theorems~1.2 and~1.4]{GottigPerezPrivitelli2024}.
Here several targets must share the same core; repeated locator slots must
also be allowed, since they cancel in binary support.
In the separate completion-cover route, the difficulty is to place an arbitrary syndrome space in a position
where a construction with common locators can treat all of its basis targets.
A reduction using shared columns places the space in that situation and returns
the resulting support to the original space.  In
particular, \(c_{e,m,t}=0\) when \(g_{e,m}=1\).  For every \(m\),
\[
 c_{e,m,t}\leq
 \left\lceil\log_2\!\min\left\{
 t,\,1+\left\lfloor\log_2(2e-3)\right\rfloor
 \right\}\right\rceil,
\]
so the correction term in that route is bounded independently of \(t\) once \(e\) is fixed.  When
\(t=2\), the construction specializes to the universal \(3e\) bound
proved in the self-contained second-order part of this paper.

The new field threshold implies the BCH rank condition and \(t\leq m\),
so the Griesmer lower bound gives
\[
 \sum_{j=0}^{t-1}\left\lceil\frac{2e-1}{2^j}\right\rceil
 \leq R_t(C_{e,m}).
\]
If, in addition,
\(2^m\geq 2^{t(et-1)}/(et-1)!\),
the sphere-covering inequality gives \(R_t(C_{e,m})\geq et\).
The additive gap is then at most \(e-1\), independently of \(t\) for
fixed \(e\); this is the asserted asymptotic tightness.

For comparison, the earlier degree term satisfies
\[
 B_{e,t}\leq
 (e!)^t\left(1+\frac{(e-1)D_e}{e!}t\right)^{e-1}.
\]
Consequently, for fixed \(e\), the field exponent sufficient for the upper
bound from either refined construction grows as
\[
 4t\log_2(e!)+O_e(\log t),
\]
rather than the \(4t\log_2D_e+O(1)\) supplied by ordinary B\'ezout.
For \(e=2,3,4\), this earlier refined term first improves ordinary
B\'ezout at orders \(t=6,9,10\), respectively.  The budget \(L_{e,t}\)
is no larger and can be strictly smaller.  The displayed third-order thresholds
below remain unchanged.

For three through six errors, the correction-free upper bound holds without
any congruence restriction on \(m\), once the field-size condition is met.
At order three it gives the following explicit intervals.

\begin{theorem}[Third-order consequences]
\label{thm:intro-main}
For the triple-error-correcting family,
\[
 10\leq R_3(C_{3,m})\leq11
 \qquad(m\geq48).
\]
For the four-error-correcting family,
\[
 13\leq R_3(C_{4,m})\leq15
 \qquad(m\geq82).
\]
For the five- and six-error-correcting families, respectively,
\[
 17\leq R_3(C_{5,m})\leq19
 \qquad(m\geq120),
\]
and
\[
 20\leq R_3(C_{6,m})\leq23
 \qquad(m\geq162).
\]
\end{theorem}

The lower bounds follow from the Generalized Supercode Lemma and the
binary Griesmer bound.  The upper bounds follow from the componentwise
common-core bound above: the displayed field thresholds imply
\(2^m>2L_{e,3}^{\,4}\).

The bounds in Theorem~\ref{thm:intro-main} do not determine any of these
stable radii exactly.  We do, however, determine the minimum number of BCH
parity-check columns on every three-dimensional pure syndrome space for three
errors; here \emph{pure} means confined to the highest syndrome coordinate.
We obtain exactly 10 for \(m\geq18\). Thus a
possible global value 11 for \(R_3(C_{3,m})\) must come from a nonpure
syndrome space all of whose nonzero syndromes require at least three
parity-check columns.  The exact global values and lower effective
thresholds remain open.

\paragraph{Our contributions.}
\begin{itemize}
\item We determine exact second-order families and give effective bounds
with two adjacent possible values for every fixed \(e\geq2\).
\item We prove explicit bounds at every higher order whose additive gap,
under the stated field-size conditions, is bounded independently of the order
for each fixed error parameter.
\item We sharpen the sufficient field-size threshold for the common-core
bound for arbitrary targets, using componentwise degree estimates and an
unweighted multicone bound.
\item We determine exact support size ten for every three-dimensional pure
syndrome space of the three-error-correcting family when \(m\geq18\).
\end{itemize}

Section~\ref{sec:preliminaries} fixes the syndrome formulation and proves the
general lower bounds.  The self-contained second-order part then gives the
exact four-error result, the universal two-value bound, and the root-free
criterion. Section~\ref{sec:cre-refinement} proves the componentwise
common-core bound with its sharper degree budget.
Section~\ref{sec:uniform-order} retains the separate completion-cover
bound and its regularization argument, and Section~\ref{sec:third-order}
gives the third-order consequences. Completion covers for one target and the
associated sign valuations are developed before Section~\ref{sec:higher-framework}, which proves
the simultaneous higher-order criterion and returns to BCH support.

\section{Preliminaries and lower bounds}
\label{sec:preliminaries}

We express the covering problem as support for binary syndrome spaces and
record two lower bounds against which the constructions will be compared.
Throughout,
\[
 q\eqdef2^m,\qquad F\eqdef\F_q,\qquad F^*\eqdef F\setminus\{0\}.
\]
Recall that \(k_e=2e-3\).

\subsection{BCH columns and common-support distance}

The field-valued parity-check column of \(C_{e,m}\) indexed by \(x\in F^*\)
is
\[
 h_e(x)\eqdef(x,x^3,\ldots,x^{2e-1})^\T\in F^e.
\]
We also put \(h_e(0)=0\).  Let
\[
 \pi_e:F^e\longrightarrow F^{e-1}
\]
delete the highest syndrome coordinate, and write
\[
 \bfs_\gamma^\T\eqdef(0,\ldots,0,\gamma)^\T
\]
for a pure highest-coordinate syndrome.  A nonzero syndrome is called
\emph{one-column} if it equals \(h_e(z)\) for some \(z\in F^*\).
Fix an ordered \(\F_2\)-basis \(\mathcal B\) of \(F\), and let
\[
 \operatorname{coord}_{\mathcal B,e}:F^e\longrightarrow\F_2^{em}
\]
be the componentwise coordinate isomorphism, with its \(e\) blocks in the
order displayed in \(h_e\).  Define
\[
 H_{e,m}\eqdef
 \bigl(\operatorname{coord}_{\mathcal B,e}(h_e(x))\bigr)_{x\in F^*}
 \in\F_2^{em\times(q-1)}.
\]
Thus the columns are indexed by \(F^*\), and the \(x\)-column is the binary
coordinate expansion of \(h_e(x)\).  Changing \(\mathcal B\) applies an
invertible binary row operation to each coordinate block, so all support
quantities below are independent of this choice.  Under the standard condition
\begin{equation}
\label{eq:rank-condition}
 2e-1\leq2^{\lceil m/2\rceil},
\end{equation}
the binary cyclotomic cosets of \(1,3,\ldots,2e-1\) are distinct and have
size \(m\), so \(H_{e,m}\) has rank \(em\) and is a full parity-check matrix
for \(C_{e,m}\)
\cite[Lemmas~8 and~9]{AlyKlappeneckerSarvepalli2007}.  Every field vector in
\(F^e\) is then a binary syndrome.  Condition~\eqref{eq:rank-condition}
holds throughout the ranges of Theorems~\ref{thm:intro-general-order}
and~\ref{thm:intro-main}.

For a binary subspace \(W\subseteq F^e\), put
\[
 \mu_e(W)\eqdef
 \min\left\{|I|:I\subseteq F^*,\quad
 W\subseteq\Span_{\F_2}\{h_e(x):x\in I\}\right\}.
\]
For \(t\) syndromes, take \(W\) to be their binary span.  This formulation
makes clear that changing the ordered generating tuple does not change the
minimum.  Consequently,
\begin{equation}
\label{eq:Rt-subspaces}
 R_t(C_{e,m})=
 \max_{\substack{W\subseteq F^e\\ \dim_{\F_2}W\leq t}}\mu_e(W).
\end{equation}

For an ordered tuple \(\mathbf x=(x_1,\ldots,x_r)\in F^r\), its reduced
support is
\[
 E(\mathbf x)\eqdef
 \{x\in F^*:|\{i:x_i=x\}|\text{ is odd}\}.
\]
Then \(\sum_i h_e(x_i)=\sum_{x\in E(\mathbf x)}h_e(x)\), so deleting zero
entries and even repetitions never increases the supporting set.

\subsection{The lower bound}

The generalized supercode lemma of Xiong and Yip
\cite[Lemma~III.1]{XiongYip2026} relates generalized covering radii to
generalized Hamming weights.

\begin{proposition}[Supercode--Griesmer lower bound]
\label{prop:lower-bound}
Let \(e,t\geq2\).  Suppose that~\eqref{eq:rank-condition} holds and
\(m\geq t\).  Then
\[
 R_t(C_{e,m})\geq
 \sum_{j=0}^{t-1}\left\lceil\frac{2e-1}{2^j}\right\rceil.
\]
In particular, for \(t=3\),
\[
 R_3(C_{3,m})\geq10,
 \qquad
 R_3(C_{4,m})\geq13,
 \qquad
 R_3(C_{5,m})\geq17,
 \qquad
 R_3(C_{6,m})\geq20.
\]
\end{proposition}

\begin{proof}
Under~\eqref{eq:rank-condition}, the added cyclotomic coset has size \(m\),
so \(C_{e,m}\) has codimension \(m\geq t\) in \(C_{e-1,m}\).
The same rank formula gives
\[
 \dim C_{e-1,m}=2^m-1-(e-1)m\geq t.
\]
Indeed, the rank condition gives \(e-1<2^{\lceil m/2\rceil-1}\), from
which the inequality is immediate for \(m\geq4\); the remaining case is
\(m=3,e=2\), when \(\dim C_{1,3}=4\).  Thus the \(t\)-th generalized
Hamming weight is defined.
The cited supercode lemma gives
\[
 R_t(C_{e,m})\geq d_t(C_{e-1,m}),
\]
where the right-hand side is the \(t\)-th generalized Hamming weight.  A
\(t\)-dimensional subcode that realizes this weight has minimum distance at
least \(2e-1\), by the BCH designed-distance bound.  The binary Griesmer
bound~\cite{Griesmer1960} gives the displayed sum.  The four numerical
values follow by substituting \(t=3\) and \(e=3,4,5,6\).
\end{proof}

For \(t=2\), the resulting universal lower bound
\(R_2(C_{e,m})\geq3e-1\) was already obtained by
Yohananov and Schwartz~\cite{YohananovSchwartz2025}.

The sphere-covering inequality gives a second lower bound that becomes
stronger as the generalized order grows.

\begin{proposition}[Sphere-covering lower bound]
\label{prop:sphere-lower-bound}
Let \(e,t\geq2\), and suppose that~\eqref{eq:rank-condition} holds.  Then
\[
 \sum_{i=0}^{R_t(C_{e,m})}
 \binom{q-1}{i}(2^t-1)^i\geq q^{et}.
\]
Equivalently, for every integer \(\rho\geq0\),
\[
 \sum_{i=0}^{\rho}\binom{q-1}{i}(2^t-1)^i<q^{et}
 \quad\Longrightarrow\quad
 R_t(C_{e,m})\geq\rho+1.
\]
Consequently, if \(q-1\geq et-1\) and
\[
 q\geq\frac{2^{t(et-1)}}{(et-1)!},
\]
then \(R_t(C_{e,m})\geq et\).
\end{proposition}

\begin{proof}
The first assertion is the generalized sphere-covering inequality of
Elimelech, Firer, and Schwartz
\cite[Corollary~10 and the proof of Proposition~12]
{ElimelechFirerSchwartz2021}, specialized to a binary code of length
\(q-1\) and redundancy \(em\).

For the explicit consequence, put \(s=et-1\) and suppose that the radius is
at most \(s\).  For a fixed \(s\)-element column set \(I\), a coefficient
array
\[
 (\varepsilon_{j,x})_{1\leq j\leq t,\ x\in I}
 \in\F_2^{t\times I}
\]
represents the \(j\)-th syndrome by
\(\sum_{x\in I}\varepsilon_{j,x}h_e(x)\).  An array supported on fewer than
\(s\) columns can be padded by zero coefficients.  Thus a fixed set \(I\)
produces at most \(2^{ts}\) syndrome \(t\)-tuples, and the total number that
can be covered is at most
\[
 \binom{q-1}{s}2^{ts}
 <\frac{q^s2^{ts}}{s!}
 \leq q^{s+1}=q^{et},
\]
contradicting the first assertion.
\end{proof}

For \(e=2\), this explicit counting threshold is the one established by
Xiong and Yip~\cite[Theorem~IV.1]{XiongYip2026}; the argument above extends
it to arbitrary \(e\).

\begin{remark}
For \(m\geq5\), the equality
\(d_3(C_{2,m})=10\)~\cite{VanDerGeerVanDerVlugt1994} recovers the lower
bound used for \(C_{3,m}\).  Likewise, the known equality
\(d_3(C_{3,m})=13\) in the stable range
\cite{VanDerGeerVanDerVlugt1995} recovers the lower bound used for
\(C_{4,m}\).  Thus the Griesmer argument uniformly recovers the exact values
10 and 13 for the two immediate supercodes.
\end{remark}

\section{The second generalized radius}
\label{r2:sec:second-order}

We now state the second-order results and reduce them to a common geometric
engine before turning to higher generalized-covering order.  The shared engine
is developed in Section~\ref{r2:sec:second-order-engine}.  Recall from the
introduction that \(B_{e,2}=D_e^2\), \(k_e=2e-3\), and
\[
 G_e(x)=\gcd_{\F_2[x]}
 \bigl(x^{k_e}+1,(x+1)^{k_e}+1\bigr).
\]
Thus the refined degree budget agrees at order two with the ordinary
B\'ezout budget.
The polynomial \(G_e\) governs the only syndrome planes not separated by the
lower-coordinate projection \(\pi_e\).

The four-error family has an additional direct construction which begins far
below the general algebraic threshold.

\begin{theorem}[Four errors]
\label{r2:thm:four-errors}
For every \(m\geq16\),
\[
 11\leq R_2(C_{4,m})\leq12.
\]
Moreover, for every \(m\geq55\),
\[
 R_2(C_{4,m})=11.
\]
\end{theorem}

The uniform second-order result is as follows.

\begin{theorem}[Effective second-order common-core bound]
\label{r2:thm:effective-general-errors}
Fix \(e\geq2\), put \(q=2^m\), and suppose that
\begin{equation}
\label{r2:eq:general-field-threshold}
 q>\max\{2eB_{e,2}^2,\,2B_{e,2}^4\}.
\end{equation}
Then
\[
 3e-1\leq R_2(C_{e,m})\leq3e.
\]
If \(G_e\) has no root in \(F=\F_q\), then
\[
 R_2(C_{e,m})=3e-1.
\]
For every fixed \(e\geq2\), the last equality holds for infinitely many
extension degrees \(m\).
\end{theorem}

The root condition is sufficient.  When \(G_e\) has a root in \(F\), this
particular basis and sign-class criterion leaves the choice between
\(3e-1\) and \(3e\) unresolved.

In the following examples, the first two use the stronger low-error clause
of the higher-order theorem, while the last two follow from the root-free
criterion above.

\begin{corollary}[Sample exact families]
\label{r2:cor:sample-exact-families}
One has
\[
 R_2(C_{5,m})=14\qquad(m\geq81),
\]
\[
 R_2(C_{6,m})=17\qquad(m\geq108),
\]
\[
 R_2(C_{7,m})=20\qquad(m\geq138),
\]
and
\[
 R_2(C_{8,m})=23\qquad(m\geq169).
\]
\end{corollary}

\subsection{The one-column reduction}
\label{r2:subsec:plane-reductions}

The direct four-error argument and the later common-core construction both
begin by removing one-column targets.  Recall that a nonzero syndrome is
one-column when it equals \(h_e(x)\) for some \(x\in F^*\).

\begin{lemma}[Avoiding one-column basis elements]
\label{r2:lem:one-column-basis}
Let \(W\subseteq F^e\) be a two-dimensional binary syndrome space.  If
\(W\) contains two distinct nonzero one-column syndromes, those two columns
span \(W\).  If \(W\) contains exactly one nonzero one-column syndrome, the
other two nonzero elements form a basis of \(W\), and neither is one-column.
\end{lemma}

\begin{proof}
The first assertion is immediate.  In the second case, the other two
nonzero elements are distinct and their sum is the unique one-column
syndrome.  They therefore form a basis and cannot themselves be one-column.
\end{proof}

\subsection{Four errors: two common locators and quintic completion}
\label{r2:subsec:four-errors}

The general theorem gives the exact stable value for four errors.  A separate
degree-five argument gives the nearly exact interval much earlier.  Recall
\[
 h_4(x)=(x,x^3,x^5,x^7)^\T.
\]

We shall use Cohen's sufficient condition for the stable ordinary covering
radius \cite[Theorem~1.1]{Cohen1997LengthBCH}:
\begin{equation}
\label{r2:eq:ordinary-radius}
 R_1(C_{e,m})=2e-1
 \quad\text{if}\quad
 2^m>\bigl[(2e-3)(2e-1)!\bigr]^2.
\end{equation}

\begin{proposition}[The direct twelve-column bound]
\label{r2:prop:four-upper-twelve}
If \(m\geq16\), then every binary syndrome space
\(W\subseteq F^4\) of dimension at most two satisfies
\[
 \mu_4(W)\leq12.
\]
\end{proposition}

\begin{proof}
Enlarging \(W\) if necessary, we may assume that \(\dim W=2\), since
\(\mu_4\) is nondecreasing under inclusion of syndrome spaces.
By Lemma~\ref{r2:lem:one-column-basis}, either two columns
already span \(W\), or \(W\) has a basis
\(\bfs^\T,\bft^\T\) in which neither target is one-column.  We treat the
second case.

We seek two common locators \(c_1,c_2\) and five target-specific locators for
each basis element:
\[
 \bfs^\T
 =\sum_{i=1}^{2}h_4(c_i)+\sum_{j=1}^{5}h_4(x_j),
 \qquad
 \bft^\T
 =\sum_{i=1}^{2}h_4(c_i)+\sum_{j=1}^{5}h_4(y_j).
\]
The reduced-support convention then gives a common support of size at most
\[
 2+5+5=12.
\]

The completion step uses the degree-five case of Cohen's polynomial
splitting estimate
\cite[Lemma~3.1 and Eq.~(3.2)]{Cohen1997LengthBCH}.  Let
\(f_0,f_1\in F[y]\) be monic of degrees five and three.  Suppose that they
are coprime, that \(f_0/f_1\) is functionally indecomposable, that every
geometric root of \(f_1\) has odd multiplicity, and that, for every
\(\lambda\in\overline F\), every root of \(f_0+\lambda f_1\) has odd
multiplicity or multiplicity two.  If \(N\) is the number of \(a\in F\)
for which \(f_0+af_1\) has five distinct roots in \(F\), then
\begin{equation}
\label{r2:eq:cohen-five-bound}
 120N\geq q-242\sqrt q-600.
\end{equation}
The right-hand side is positive for \(q\geq2^{16}\).

For either target, write
\[
 \bfu^\T=(\sigma_1,\sigma_3,\sigma_5,\sigma_7)^\T.
\]
For a prospective common pair
\(\mathbf c=(c_1,c_2)\), define the residual moments
\[
 p_d^{[5]}(\bfu,\mathbf c)
 \eqdef \sigma_d+c_1^d+c_2^d,
 \qquad d\in\{1,3,5,7\},
\]
and abbreviate
\[
 (s,u,v,w)\eqdef
 \bigl(p_1^{[5]},p_3^{[5]},p_5^{[5]},p_7^{[5]}\bigr).
\]
Thus a five-locator completion is a choice of five field elements having
these four odd power sums.

Let \(e_i\) denote the elementary symmetric functions of the five completing
locators, and use \(a\eqdef e_2\) as a free parameter.  Put
\[
 \alpha\eqdef u+s^3,\qquad
 \beta\eqdef v+us^2,\qquad
 \gamma\eqdef w+su^2+s^7+vs^2.
\]
The characteristic-two Newton identities give
\[
 e_1=s,\qquad e_3=\alpha+sa,\qquad
 \alpha e_4=\gamma+\beta a,\qquad
 e_5=\beta+\alpha a+se_4.
\]
When \(\alpha\neq0\), the polynomial of the completing locators is the
pencil
\[
 f_a(y)=f_0(y)+af_1(y),
\]
where
\begin{align*}
 f_0(y)
 &=y^5+sy^4+\alpha y^2+\frac{\gamma}{\alpha}y
   +\beta+\frac{s\gamma}{\alpha},\\
 f_1(y)
 &=y^3+sy^2+\frac{\beta}{\alpha}y
   +\alpha+\frac{s\beta}{\alpha}.
\end{align*}
Define the second admissibility scalar
\[
 \kappa\eqdef\alpha\gamma+\beta^2.
\]
A direct calculation gives
\begin{equation}
\label{r2:eq:quintic-coprime-identity}
 \alpha^2f_0+(\alpha^2y^2+\alpha\beta)f_1
 =\kappa(y+s).
\end{equation}
If \(\alpha\kappa\neq0\), the identity and \(f_1(s)=\alpha\) show that
\(f_0\) and \(f_1\) are coprime.  The rational map \(f_0/f_1\) has prime
degree five and is functionally indecomposable.  Moreover,
\(f_1'(y)=y^2+\beta/\alpha\), and substituting its unique geometric root
into \(f_1\) leaves \(\alpha\); hence \(f_1\) is square-free.  A second
direct calculation gives
\[
 f_0'(y)f_1(y)+f_0(y)f_1'(y)=\frac{\kappa}{\alpha}.
\]
The same identity holds with \(f_0\) replaced by
\(f_0+\lambda f_1\), so every member of the pencil is square-free.  Thus
all hypotheses of~\eqref{r2:eq:cohen-five-bound} hold.  Consequently, the
target has a five-locator completion whenever \(\alpha\kappa\neq0\).

It remains to choose one common pair for which this condition holds for both
targets.  Fix a nonzero desired sum \(r\) of the common locators, and denote
their product by \(p\).  We count the products forbidden by each target after
normalizing its first moment to zero.

For an odd-moment vector
\(\boldsymbol\xi=(\xi_1,\xi_3,\xi_5,\xi_7)\), put
\[
 \xi_0\eqdef1,\qquad
 \xi_{2j}\eqdef\xi_j^2\quad(1\leq j\leq3),
\]
recursively, and define the triangular translation map
\[
 (\mathcal T_b\boldsymbol\xi)_d
 \eqdef
 \sum_{\ell=0}^{d}\binom d\ell b^{d-\ell}\xi_\ell,
 \qquad d\in\{1,3,5,7\},
\]
with binomial coefficients reduced modulo two.  Translating all seven
locator slots by \(b\) changes their odd-moment vector by
\(\mathcal T_b\).  In particular,
\[
 \mathcal T_b^2=\operatorname{id},
 \qquad
 \mathcal T_b(h_4(x))=h_4(x+b).
\]
Choosing \(b=\xi_1\) makes the first transformed coordinate zero.  The sum
\(r\) of the two common slots is unchanged, while their product changes by
\[
 p\longmapsto p+br+b^2.
\]
This is a bijection of the possible product values.  The two targets may use
different normalizing translations: after a common original product \(p\)
is chosen, each target receives its image under its own bijection, and the
inverse translations return both pairs to the same original pair.  Thus the
bad products may be counted separately in the two normalized coordinates.

For either target, its normalized syndrome has the form
\[
 (0,\lambda,\nu,\rho)^\T
\]
and is nonzero; otherwise the original target would be one-column.  The power
sums of a pair with sum \(r\) and product \(p\) give
\[
 \alpha=\lambda+rp,\qquad
 \beta=\nu+\lambda r^2+rp^2,\qquad
 \gamma=\rho+\lambda^2r+\nu r^2+rp^3,
\]
and hence
\begin{align}
\label{r2:eq:quintic-bad-product}
 \kappa(r,p)={}&
 \nu^2+\lambda\rho+\lambda^3r+\lambda\nu r^2
 +\lambda^2r^4+\rho rp+\lambda^2r^2p \notag\\
 &{}+\nu r^3p+\lambda rp^3.
\end{align}
For fixed \(r\neq0\), if \(\lambda\neq0\), the condition \(\alpha=0\)
excludes one value of \(p\), and the nonzero cubic
\(\kappa(r,p)\) excludes at most three more.  If
\(\lambda=0\) and \(\nu\neq0\), the two conditions exclude at most two
values.  If \(\lambda=\nu=0\), then \(\rho\neq0\), and only \(p=0\) is
excluded.  Each target therefore forbids at most four products, and the two
targets together forbid at most eight.

The products of split pairs with sum \(r\) are exactly
\[
 \{x(x+r):x\in F\}.
\]
The additive map \(x\mapsto x^2+rx\) has kernel \(\{0,r\}\), so this set
has \(q/2\) elements.  Choose a product outside the bad values, and then a
split pair \(c_1,c_2\) with that sum and product.  Since \(q\geq2^{16}\),
Cohen's estimate supplies five completing locators for each normalized
target.  Applying the inverse translations gives the required two
representations with the same original pair.  Zeros, repetitions, and
further overlaps only reduce the resulting binary support, so
\(\mu_4(W)\leq12\).
\end{proof}

\begin{proof}[Proof of Theorem~\ref{r2:thm:four-errors}]
For \(m\geq16\), Proposition~\ref{r2:prop:four-upper-twelve} gives the upper
bound.  Proposition~\ref{prop:lower-bound} with \(e=4,t=2\) gives the lower
bound \(11\).

For the exact stable range,
\[
 B_{4,2}=D_4^2=(1\cdot3\cdot5\cdot7)^2=11025,
\]
and
\[
 \max\{8B_{4,2}^2,\,2B_{4,2}^4\}=2B_{4,2}^4
 =29549108875781250.
\]
The first power of two larger than this number is \(2^{55}\).  Moreover,
\[
 G_4(x)=
 \gcd_{\F_2[x]}\bigl(x^5+1,(x+1)^5+1\bigr)=1.
\]
The root-free clause of
Theorem~\ref{r2:thm:effective-general-errors} therefore gives
\(R_2(C_{4,m})=11\) for every \(m\geq55\).
\end{proof}

\subsection{The known boundary value}
\label{r2:subsec:elementary-boundary}

The smallest four-error code does not satisfy the full-rank field model, but
it is a repetition code, so its second generalized radius is already covered
by the repetition-code formula.

\begin{proposition}[The case \(m=4\)]
\label{r2:prop:four-errors-m4}
One has
\[
 R_2(C_{4,4})=11.
\]
\end{proposition}

\begin{proof}
The code \(C_{4,4}\) is the binary repetition code of length \(15\).
The frequency argument of
\cite[Proposition~9]{ElimelechWeiSchwartz2022ReedMuller} extends to arbitrary
length \(n\).  Regard the columns of a binary \(t\)-row array as symbols
from an alphabet of size \(2^t\).  Its distance from the nearest repetition
array is \(n\) minus the largest symbol frequency.  This frequency is at
least \(\lceil n/2^t\rceil\), with equality when the symbol counts differ
by at most one.  Hence
\[
 R_t(\operatorname{Rep}_2(n))=n-\left\lceil\frac{n}{2^t}\right\rceil.
\]
Substituting \(n=15\) and \(t=2\) gives the claim.
\end{proof}

The stable theorem below is proved independently of this boundary calculation.

\subsection{Exact examples and the higher-order sharpening}
\label{r2:subsec:exact-examples}

\begin{proof}[Proof of Corollary~\ref{r2:cor:sample-exact-families}]
For \(e=5,6\), apply the low-error clause of
Theorem~\ref{thm:intro-general-order} with \(t=2\).  It gives
\[
 R_2(C_{e,m})\leq3e-1
\]
under its field-size hypothesis, while
Proposition~\ref{prop:lower-bound} gives the reverse inequality.  Direct
substitution in
\[
 q>\max\{2eB_{e,2}^2,\,2B_{e,2}^4\},
 \qquad B_{e,2}=D_e^2,
\]
shows that the first permitted extension degrees are \(81\) and \(108\),
respectively.  This proves the first two assertions as consequences of the
general higher-order theorem, without any root restriction.

For the remaining two families, a direct Euclidean calculation gives
\[
 G_7(x)=G_8(x)=1.
\]
The first powers of two satisfying
\eqref{r2:eq:general-field-threshold} occur at extension degrees \(138\) and
\(169\), respectively.  The root-free clause of
Theorem~\ref{r2:thm:effective-general-errors} gives the last two assertions.
\end{proof}

\begin{remark}
\label{r2:rem:small-error-comparison}
Sharper results are known for two and three errors.  The double-error-correcting
family has second radius
five apart from its known exceptional degree, and the triple-error-correcting
family has second radius eight throughout its full-rank range
\cite{YohananovSchwartz2025,EssayagZabokritskiy2026BCH3}.
\end{remark}

\subsection{Projection reductions and the common-core input}
\label{r2:subsec:common-core-input}

We now isolate the remaining elementary choices that place a binary syndrome
plane in the range of the uniform common-core construction.

\begin{lemma}[A basis adapted to the lower projection]
\label{r2:lem:projection-basis}
Suppose that the restriction of \(\pi_e\) to a two-dimensional syndrome
space \(W\) has positive rank.  Then either two one-column syndromes already
span \(W\), or \(W\) has a basis of non-one-column syndromes with different
lower projections.
\end{lemma}

\begin{proof}
If \(\pi_e|_W\) has rank one, choose the nonzero kernel element and any
element outside the kernel.  If it has rank two, every basis has distinct
images.  When \(W\) has exactly one one-column syndrome, use the other two
nonzero elements as in Lemma~\ref{r2:lem:one-column-basis}; their projections
differ by the nonzero lower projection of that column.
\end{proof}

A plane on which \(\pi_e\) vanishes consists of the zero syndrome and three
pure highest-coordinate syndromes
\[
 \bfs_\gamma^\T=(0,\ldots,0,\gamma)^\T.
\]
The following lemma explains the definition of \(G_e\).

\begin{lemma}[The pure-highest arithmetic obstruction]
\label{r2:lem:pure-highest-obstruction}
Let \(U\subseteq F\) be a two-dimensional binary subspace.  It has no basis
\(\gamma_1,\gamma_2\) satisfying
\[
 \gamma_1^{k_e}\neq\gamma_2^{k_e}
\]
if and only if, after scaling its three nonzero elements to
\[
 1,\qquad r,\qquad 1+r,
\]
one has \(G_e(r)=0\).
\end{lemma}

\begin{proof}
The three possible unordered bases of \(U\) are the three pairs among its
nonzero elements.  Every basis has equal \(k_e\)-th powers precisely when all
three nonzero elements do.  After scaling one element to one, this condition
is
\[
 r^{k_e}=1,\qquad (1+r)^{k_e}=1,
\]
which is equivalent to \(G_e(r)=0\).
\end{proof}

The geometric construction needed below is most naturally stated directly
in the coding language.  Its proof is deferred to
Section~\ref{r2:sec:second-order-engine}.

\begin{proposition}[Second-order common-core plane criterion]
\label{r2:prop:common-core-plane}
Let \(e\geq2\), let \(W\subseteq F^e\) be a two-dimensional binary syndrome
space, and suppose that
\[
 q>\max\{2eB_{e,2}^2,\,2B_{e,2}^4\}.
\]
Assume that either
\begin{enumerate}
\item the restriction of \(\pi_e\) to \(W\) has positive rank; or
\item \(\pi_e(W)=0\) and \(W\) has a basis
      \(\bfs_{\gamma_1}^\T,\bfs_{\gamma_2}^\T\) for which
      \[
       \gamma_1^{2e-3}\neq\gamma_2^{2e-3}.
      \]
\end{enumerate}
Then
\[
 \mu_e(W)\leq3e-1.
\]
\end{proposition}

The proposition constructs one common tuple of length \(e-1\) and two
completing tuples of length \(e\).  Lemma~\ref{lem:common-core-certificate},
proved below, turns those tuples into the displayed support bound.  The
proposition is the only deferred input used in the proofs below.

\subsection{The uniform second-order theorem}
\label{r2:subsec:uniform-second-order}

\begin{proof}[Proof of Theorem~\ref{r2:thm:effective-general-errors}]
Condition~\eqref{r2:eq:general-field-threshold} implies both the rank
condition~\eqref{eq:rank-condition} and the hypothesis of
\eqref{r2:eq:ordinary-radius}.  Indeed, \(B_{e,2}=D_e^2\), and pairing each
even factor with the next odd factor gives \((2e-1)!\leq D_e^2\).
Since \(D_e\geq2e-1\), the field-size condition implies
\[
 q>2D_e^8>\bigl[(2e-3)D_e^2\bigr]^2
 \geq\bigl[(2e-3)(2e-1)!\bigr]^2.
\]
Also \(B_{e,2}\geq(2e-1)^2\), so the field-size condition gives
\[
 2^{m/2}=\sqrt q>\sqrt2\,B_{e,2}^2>2e-1.
\]
Thus every element of \(F^e\) is a syndrome, and
Proposition~\ref{prop:lower-bound} with \(t=2\) gives
\[
 R_2(C_{e,m})\geq
 (2e-1)+\left\lceil\frac{2e-1}{2}\right\rceil
 =3e-1.
\]

Let \(W\subseteq F^e\) have dimension at most two.  If
\(\dim W\leq1\), equation~\eqref{r2:eq:ordinary-radius} gives
\[
 \mu_e(W)\leq2e-1.
\]
Suppose therefore that \(\dim W=2\).

If \(\pi_e|_W\) has positive rank, the first case of
Proposition~\ref{r2:prop:common-core-plane} gives
\(\mu_e(W)\leq3e-1\).  Suppose instead that \(\pi_e(W)=0\).  If \(G_e\)
has no root in \(F\), Lemma~\ref{r2:lem:pure-highest-obstruction} supplies
a basis
\(\bfs_{\gamma_1}^\T,\bfs_{\gamma_2}^\T\) with
\[
 \gamma_1^{2e-3}\neq\gamma_2^{2e-3}.
\]
The second case of the common-core plane criterion again gives
\(\mu_e(W)\leq3e-1\).

Without the root-free hypothesis, choose any basis
\(\bfs^\T,\bft^\T\) of \(W\) and any \(z\in F^*\).  The adjusted targets
\[
 \widetilde\bfs^\T\eqdef\bfs^\T+h_e(z),
 \qquad
 \widetilde\bft^\T\eqdef\bft^\T
\]
span a plane whose lower projection has positive rank.  The first case of
Proposition~\ref{r2:prop:common-core-plane} supports the adjusted plane on at
most \(3e-1\) columns.  Restoring the single column \(h_e(z)\) supports the
original plane on at most \(3e\) columns.  Equation~\eqref{eq:Rt-subspaces}
therefore proves
\[
 3e-1\leq R_2(C_{e,m})\leq3e,
\]
and proves equality in the root-free case.

It remains to obtain infinitely many root-free extensions.  Neither zero nor
one is a root of \(G_e\), so \(G_e\) has no linear factor over \(\F_2\).
Let \(d_1,\ldots,d_s\geq2\) be the degrees of its distinct irreducible
factors.  A root of the factor of degree \(d_i\) belongs to
\(\F_{2^m}\) exactly when \(d_i\mid m\).  If \(G_e=1\), every extension is
root-free.  Otherwise, every prime
\(m>\max_i d_i\) is root-free.  Infinitely many such degrees also exceed the
fixed threshold~\eqref{r2:eq:general-field-threshold}.
\end{proof}

\section{A componentwise refinement of the common-core bound}
\label{sec:cre-refinement}

How small a degree bound can be retained when the simultaneous completion
cover is reducible?  To prove the upper bound in
Theorem~\ref{thm:intro-componentwise}, it suffices to control one absolutely
irreducible component defined over \(F\).  We use the component selected by
the smooth point at infinity of Xiong, Yip, and
Zullo~\cite[Proposition~3.2]{XiongYipZullo2026}.  We first verify its
dominance over the common-core space, then bound its mixed degrees through
the one-target projections.  A comparison with the projectivized
multicone converts their sum into an ordinary degree bound, to which the
finite-field point estimate applies.

Throughout this section, let \(e,t\geq2\), and write
\[
 a=e-1,\qquad E=e!,\qquad D=D_e=\prod_{j=1}^{e}(2j-1).
\]
The polynomial and degree budgets used below are
\begin{equation}
\label{eq:cre-degree-budgets}
 \begin{split}
 P_e(w)&=E+D\sum_{u=1}^{a-1}w^u+Ew^a,\\
 \Gamma_{e,t}&=\sum_{v=0}^{a}[w^v]P_e(w)^t,
 \qquad
 L_{e,t}=\min\{D^t,\Gamma_{e,t}\}.
 \end{split}
\end{equation}
Here \([w^v]\) denotes coefficient extraction, and the interior sum is
empty when \(a=1\).  In particular, \(P_2(w)=2(1+w)\).
We will also obtain the closed expression
\begin{equation}
\label{eq:cre-closed-budget}
 \Gamma_{e,t}
 =\sum_{k=0}^{\min\{a,t\}}
   \binom tk\binom ak E^{t-k}D^k
   -t(D-E)E^{t-1}.
\end{equation}

For targets
\(\bfs_i^\T=(s_{i,1},\ldots,s_{i,e})^\T\in F^e\), let
\(\mathcal Z_{\bfs_1,\ldots,\bfs_t}\subseteq\mathbb A_F^{a+et}\)
be the affine incidence variety of tuples
\((\mathbf z,\mathbf x_1,\ldots,\mathbf x_t)\) satisfying
\begin{equation}
\label{eq:cre-common-core}
 \sum_{u=1}^{a}h_e(z_u)+\sum_{\ell=1}^{e}h_e(x_{i,\ell})
 =\bfs_i^\T
 \qquad(1\leq i\leq t).
\end{equation}
Its ordinary degree means the degree of its closure under the standard
affine embedding into projective space.  No distinctness or nonvanishing
conditions are imposed on its locator slots.

\begin{theorem}[Unrestricted common-core construction]
\label{thm:cre-common-core}
Let \(q=2^m>2L_{e,t}^{\,4}\).  For every ordered list of targets
\(\bfs_1^\T,\ldots,\bfs_t^\T\in F^e\), the system
\eqref{eq:cre-common-core} has an \(F\)-rational solution.
Thus \(e-1\) common locator slots and \(e\) completing slots for each
target suffice, allowing repetitions and zero.  In particular, zero,
one-column, and linearly dependent targets are included.
\end{theorem}

The unrestricted support count is due to
Xiong--Yip--Zullo~\cite{XiongYipZullo2026}.  The estimate here concerns the
degree of their selected component and the resulting sufficient field size.
It does not require the entire simultaneous cover to be connected or its
sign fields to be independent.

\subsection{The selected component and its one-target degrees}

The point at infinity must select a component that reaches a general core,
since the degree calculation will use its one-target projections.

\begin{lemma}[Dominance of the infinity component]
\label{lem:cre-infinity-component}
If \(q\geq e\), then, for arbitrary targets,
\(\mathcal Z_{\bfs_1,\ldots,\bfs_t}\) has an absolutely irreducible
\(F\)-component \(\mathcal C\) of dimension \(a\) that dominates
the common-core space \(\mathbb A_F^a\).
\end{lemma}

\begin{proof}
Homogenize the moment equations using a coordinate \(\tau\):
\[
 g_{i,j}=
 \sum_{u=1}^{a}z_u^{2j-1}
 +\sum_{\ell=1}^{e}x_{i,\ell}^{2j-1}
 -s_{i,j}\tau^{2j-1}
 \qquad(1\leq i\leq t,\ 1\leq j\leq e).
\]
Choose distinct \(a_1,\ldots,a_e\in F\) with \(a_1=1\) and
\(a_e=0\).  The point
\[
 P=[\tau=0;\ z_u=a_u;\ x_{i,\ell}=a_\ell]
\]
lies on these projective equations.  Their completion-coordinate Jacobian
at \(P\) has \(t\) nonsingular Vandermonde blocks, with entries
\(a_\ell^{2j-2}\).  As in
\cite[Proposition~3.2]{XiongYipZullo2026}, \(P\) is therefore a smooth
point of a unique geometric component \(Y\), of dimension \(a\).
Uniqueness and the rationality of \(P\) imply that \(Y\) is defined
over \(F\).

On the chart \(z_1=1\), the same minor is the relative Jacobian for
projection to \((\tau,z_2,\ldots,z_a)\).  The projection is
\'etale near \(P\), so its image contains a nonempty open subset of
this chart of the projective core space.  Hence \(Y\) dominates
\([\tau:\mathbf z]\in\mathbb P_F^a\).  Its affine part
\(\mathcal C=Y\cap\{\tau\ne0\}\) is nonempty and dominates
\(\mathbb A_F^a\).  It is the required absolutely irreducible
component of the affine incidence.
\end{proof}

The next lemma bounds each one-target cover without imposing a monodromy
condition on the target.

\begin{lemma}[One-target completion degree]
\label{lem:cre-one-target-degree}
Let \(k=\overline F\), \(K=k(z_1,\ldots,z_a)\), and let
\(\mathcal C_i\) be a geometric irreducible component of a one-target
incidence that dominates the core space.  Then
\[
 [k(\mathcal C_i):K]\leq E.
\]
After restricting the base to a nonempty open set, this map is surjective
finite \'etale.  For finitely many targets, the same base open set can
be used for all such maps and their fiber product.
\end{lemma}

\begin{proof}
For the fixed target, set
\[
 p_{2j-1}=s_{i,j}+\sum_{u=1}^{a}z_u^{2j-1},
 \qquad p_0=e\bmod2,\qquad p_{2r}=p_r^2,
\]
for the moments up to index \(2e-1\), and form the Hankel matrix
\[
 H=(p_{u+v})_{0\leq u,v<e}.
\]
Its determinant is not the zero polynomial in the core coordinates.
Indeed, in its homogeneous version at
\([\tau:\mathbf z]=[0:a_1:\cdots:a_a]\), with
\(a_1,\ldots,a_a,0\) distinct in \(k\), the moments are those of
these \(e\) distinct elements.  Thus the determinant is their
Vandermonde square.  Every term in the homogeneous determinant has degree
\(e(e-1)\), so this also certifies nonvanishing on the affine base.

Over \(K\), solve the nonsingular linear system
\[
 \sum_{\ell=0}^{e-1}c_\ell p_{r+\ell}=-p_{e+r}
 \quad(0\leq r<e),
 \qquad
 f(y)=y^e+\sum_{\ell=0}^{e-1}c_\ell y^\ell.
\]
On \(\mathcal C_i\), the moment equations and Frobenius give
\(p_r=\sum_{\ell=1}^{e}x_{i,\ell}^{r}\) for
\(0\leq r\leq2e-1\).  Consequently,
\[
 H=VV^\T,
 \qquad
 V=(x_{i,\ell}^{r})_{\substack{0\leq r<e\\1\leq\ell\leq e}},
 \qquad
 \sum_{\ell=1}^{e}x_{i,\ell}^{r}f(x_{i,\ell})=0.
\]
The matrix \(V\) is invertible.  Hence the completing roots are
distinct, each satisfies \(f\), and
\[
 f(y)=\prod_{\ell=1}^{e}(y-x_{i,\ell}).
\]
The function field of \(\mathcal C_i\) is therefore the splitting
field of this separable polynomial over \(K\).  Its Galois group
embeds in \(S_e\), proving the degree bound.

The same calculation over the base ring with \(\det H\) inverted
makes every completing root integral.  The completion Jacobian is a
Vandermonde square and is invertible there.  Thus the one-target incidence
is finite \'etale over this open set.  Restricting to each dominating
component and, if necessary, shrinking the base gives a surjective finite
\'etale map.  Intersecting finitely many such base opens proves the
last assertion, since finite \'etale maps are preserved by fiber products.
\end{proof}

\subsection{Mixed degrees and boundary intersections}

Eliminate the first moment in each completion block by writing
\begin{equation}
\label{eq:cre-elimination}
 x_{i,e}=s_{i,1}+\sum_{u=1}^{a}z_u+\sum_{\ell=1}^{a}x_{i,\ell}.
\end{equation}
This affine-linear identification preserves ordinary degree.  Compactify
the core and each retained completion block separately in copies of
\(\mathbb P^a\).  Write \(H_z\) and \(H_i\) for the hyperplane
classes pulled back from the core and completion factor \(i\).
These compactifications can introduce excess intersections at the
boundary; the following positivity observation lets us bound the isolated
affine slices without assuming that all intersections are proper.

Let \(M\) be a smooth projective variety with globally generated tangent
bundle, and let \(V_1,\ldots,V_s\subseteq M\) have codimensions
summing to \(\dim M\).  The small diagonal in \(M^s\) has normal
bundle \(T_M^{\oplus(s-1)}\), which is globally generated.
The localized intersection product with
\(V_1\times\cdots\times V_s\) is obtained by intersecting its normal
cone with the zero section.  By
Fulton--Lazarsfeld~\cite[Theorem~2(A), and its proof under hypothesis
(A'), p.~101]{FultonLazarsfeld1982}, each irreducible normal-cone
contribution has nonnegative degree; the normal-cone description is
recalled on p.~100 of that paper.  Hence, for any collection \(Q\) of
isolated intersection points,
\begin{equation}
\label{eq:cre-isolated-length}
 \sum_{p\in Q}i_p(V_1,\ldots,V_s;M)
 \leq\deg([V_1]\cdots[V_s]).
\end{equation}
If the local equations form a regular sequence at \(p\), the local
intersection multiplicity is the scheme length.  The argument uses global
generation, not characteristic-zero transversality, and applies in
particular to products of projective spaces in characteristic two.

\begin{lemma}[Componentwise mixed-degree product]
\label{lem:cre-mixed-product}
Let \(C_i\subseteq\mathbb P_z^a\times\mathbb P_i^a\),
\(1\leq i\leq t\), be irreducible varieties whose first projections
are surjective finite \'etale over a common nonempty open set \(U\).
Put
\[
 c_{i,u}=\deg([C_i]H_z^{a-u}H_i^u)
 \qquad(0\leq u\leq a).
\]
For a component of their fiber product over \(U\), let \(C\) be its
closure in \(M=\mathbb P_z^a\times\prod_i\mathbb P_i^a\).
If \(u_i\geq0\) and \(v=\sum_i u_i\leq a\), then
\[
 \deg\left([C]H_z^{a-v}\prod_iH_i^{u_i}\right)
 \leq\prod_i c_{i,u_i}.
\]
\end{lemma}

\begin{proof}
The complement in \(C\) of its chosen finite \'etale affine open has
dimension at most \(a-1\).  Choose \(a-v\) core hyperplanes and
\(u_i\) hyperplanes in completion factor \(i\) generally.  Each
system is basepoint-free: on any remaining component, a general member
either lowers its dimension by one or avoids it.  Thus these \(a\)
hyperplanes give a finite, possibly empty, slice \(Q\) of \(C\)
avoiding the complement.  Its length is the mixed degree on the left.

Let \(W_i\subseteq M\) be the inverse image of \(C_i\), and let
\(A\subseteq M\) be the product linear subspace defined by the slicing
hyperplanes.  Near \(Q\), the intersection of the \(W_i\) is the
finite \'etale fiber product.  It is smooth, and its components are
disjoint, so it agrees locally with the selected component.  The block
equations form a regular sequence there.  On the resulting smooth
\(a\)-dimensional component, the isolated hyperplane cuts also form a
regular sequence.  The local intersection contributions are therefore
exactly the lengths in \(Q\).

Each \(W_i\) has codimension \(a\), as does \(A\), and
\(\dim M=a(t+1)\).  The tangent bundle of \(M\) is globally
generated by the Euler sequences of its factors, so
\eqref{eq:cre-isolated-length} bounds the sum of these local lengths by
\(\deg([W_1]\cdots[W_t][A])\), even when boundary components
intersect in excess.  Finally,
\[
 [W_i]=\sum_{u=0}^{a}c_{i,u}H_z^uH_i^{a-u},
 \qquad
 [A]=H_z^{a-v}\prod_iH_i^{u_i}.
\]
The coefficient of \(H_z^a\prod_iH_i^a\) in their product is
\(\prod_i c_{i,u_i}\), which proves the inequality.
\end{proof}

We next bound the one-target mixed degrees using total degrees of affine
equations.  At the two endpoints, the Hankel argument gives a smaller
bound because the unknown locators can be interchanged.

\begin{lemma}[One-target mixed degrees]
\label{lem:cre-one-target-mixed}
Let \(C_i\subseteq\mathbb P_z^a\times\mathbb P_i^a\) be the
multiprojective closure of the eliminated image of a dominating
one-target component.  Then
\[
 c_{i,0}\leq E,\qquad c_{i,a}\leq E,
 \qquad c_{i,u}\leq D\quad(0<u<a).
\]
\end{lemma}

\begin{proof}
The first endpoint is the degree of the core projection, bounded by
Lemma~\ref{lem:cre-one-target-degree}.  For the other endpoint, if the
projection to the retained completion coordinates has image of dimension
less than \(a\), then \(c_{i,a}=0\).  Otherwise it is dominant.
Regard \((x_{i,1},\ldots,x_{i,a})\) as the core and
\((z_1,\ldots,z_a,x_{i,e})\) as the \(e\) completing locators.
The one-target moment system is unchanged by this permutation of slots.
Lemma~\ref{lem:cre-one-target-degree} therefore bounds this projection
degree by \(E\) as well.  No dominance assertion is needed for the
completion projection in the first case.

For an interior index \(u\), choose a general mixed slice by
\(a-u\) core hyperplanes and \(u\) completion hyperplanes.
As in the preceding proof, it avoids the multiprojective boundary and
the complement of the smooth finite \'etale affine open.  Its length
is \(c_{i,u}\).  After \eqref{eq:cre-elimination}, the remaining
\(a\) moment equations in \(2a\) affine variables have total degrees
at most
\[
 3,5,\ldots,2e-1.
\]
They define a smooth complete intersection near the slice, since the
completion Jacobian before elimination is invertible.  Adding the
\(a\) affine linear slicing equations gives isolated
complete-intersection zeros.  Ordinary affine B\'ezout bounds their
total length by \(3\cdot5\cdots(2e-1)=D\).  Equivalently, homogenize
in a single projective space and apply the isolated-intersection bound
\eqref{eq:cre-isolated-length}; any excess at infinity contributes
nonnegatively.  This proves the interior estimates.
\end{proof}

\subsection{From mixed degrees to ordinary degree}

The Segre embedding weights mixed degrees by multinomial coefficients.
For the ordinary affine degree needed in the point estimate, those
coefficients can be avoided.  The relevant comparison uses the affine
multicone, obtained by allowing each homogeneous coordinate block to be
scaled independently.

\begin{lemma}[Ordinary degree and the multicone]
\label{lem:cre-multicone-degree}
Let \(V\subseteq\mathbb A^{n_1}\times\cdots\times\mathbb A^{n_r}\)
be an irreducible affine variety of dimension \(d\) over an algebraically
closed field, and let \(X\) be its closure in
\(\mathbb P^{n_1}\times\cdots\times\mathbb P^{n_r}\).
For the hyperplane classes \(H_1,\ldots,H_r\), one has
\[
 \deg V\leq
 \sum_{\substack{u_1,\ldots,u_r\geq0\\u_1+\cdots+u_r=d}}
 \deg([X]H_1^{u_1}\cdots H_r^{u_r}).
\]
\end{lemma}

\begin{proof}
Let \(R\) be the multihomogeneous coordinate ring of \(X\), and
let \(Z\) be the ordinary projectivization of its affine multicone:
all coordinate variables are assigned ordinary degree one.  Write
\(d_{\mathbf u}=\deg([X]\prod_iH_i^{u_i})\).
The multigraded Hilbert polynomial has top-degree part
\[
 \sum_{|\mathbf u|=d}
 d_{\mathbf u}\prod_{i=1}^{r}\frac{m_i^{u_i}}{u_i!};
\]
the identification of these coefficients with mixed degrees holds over
an arbitrary field~\cite[Theorem~3.4, Definition~4.3, and
Remark~4.4]{CidRuiz2021}.
We spell out the passage to the ordinary grading to distinguish this
degree from the Segre degree.

The degree-\(N\) Hilbert function for the ordinary grading is
\(\sum_{|\mathbf m|=N}\dim R_{\mathbf m}\).
For each \(|\mathbf u|=d\), summing the corresponding monomial gives
\[
 \sum_{|\mathbf m|=N}\prod_i\frac{m_i^{u_i}}{u_i!}
 =\frac{N^{d+r-1}}{(d+r-1)!}+O(N^{d+r-2}).
\]
The multigraded Hilbert polynomial applies when every \(m_i\) is
sufficiently large.  The omitted boundary slabs have only
\(O(N^{r-2})\) indices.  Since \(R\) is a domain, multiplication
by a nonzero homogeneous coordinate in each block injects any graded
piece into one whose index lies in the polynomial range, after a fixed
translation.  Hence every boundary piece has dimension \(O(N^d)\),
and the slabs contribute only \(O(N^{d+r-2})\).  For \(r=1\)
the assertion is immediate.  Comparing leading coefficients yields
\begin{equation}
\label{eq:cre-multicone-identity}
 \deg Z=\sum_{|\mathbf u|=d}d_{\mathbf u}.
\end{equation}
This is the coarsening identity of~\cite[Theorem~2.8]{CidRuiz2021}
for the multihomogeneous coordinate ring of an irreducible variety.

Let \(\tau_i\) be the homogenizing coordinate in block \(i\).
Intersect \(Z\) with the \(r-1\) linear equations
\(\tau_i=\tau_1\), \(2\leq i\leq r\).  On their common nonzero chart,
the intersection is exactly \(V\) in its standard affine coordinates.
Indeed, before these equations are imposed, the open set where every
\(\tau_i\ne0\) is \(V\times(\mathbb G_m)^{r-1}\), and the
equations fix the \(r-1\) relative scaling coordinates to one.
The standard projective closure of \(V\) is therefore an
expected-dimensional irreducible component of this linear section.

Repeated hyperplane B\'ezout bounds the degree of such a component by
\(\deg Z\), even if other components occur where all \(\tau_i\)
vanish: along a chain of components leading to the expected-dimensional
component, each required hyperplane cut is proper, and retaining its
positive multiplicity cannot increase the total degree.  Combining this
linear-section comparison with \eqref{eq:cre-multicone-identity}
proves the lemma.
\end{proof}

We can now return to the selected component of the common-core incidence.

\begin{proposition}[Degree budget for every dominating component]
\label{prop:cre-refined-budget}
Every geometric irreducible component \(\mathcal C\) of
\(\mathcal Z_{\bfs_1,\ldots,\bfs_t}\) that dominates the core
space satisfies
\[
 \deg\mathcal C\leq L_{e,t}.
\]
Moreover, \(L_{e,t}\geq E^t\geq e\), and
\eqref{eq:cre-closed-budget} holds.
\end{proposition}

\begin{proof}
Use \eqref{eq:cre-elimination} in every completion block.  The images
of \(\mathcal C\) in the one-target incidences are dominating
components.  On a common base open set where the Hankel determinants are
nonzero, they are finite \'etale, and \(\mathcal C\) restricts to a
component of their fiber product.  Let \(C\) and \(C_i\) denote their multiprojective
closures.  Lemmas~\ref{lem:cre-multicone-degree}
and~\ref{lem:cre-mixed-product} give
\[
 \deg\mathcal C
 \leq\sum_{\substack{u_1,\ldots,u_t\geq0\\v=\sum_i u_i\leq a}}
 \deg\left([C]H_z^{a-v}\prod_iH_i^{u_i}\right)
 \leq\sum_{\substack{u_1,\ldots,u_t\geq0\\\sum_i u_i\leq a}}
 \prod_i c_{i,u_i}.
\]
Put \(b_0=b_a=E\) and \(b_u=D\) for \(0<u<a\).
Lemma~\ref{lem:cre-one-target-mixed} bounds the last expression by
\[
 \sum_{\substack{u_1,\ldots,u_t\geq0\\\sum_i u_i\leq a}}
 \prod_i b_{u_i}
 =\sum_{v=0}^{a}[w^v]P_e(w)^t
 =\Gamma_{e,t}.
\]
The original equations independently bound the degree of each
expected-dimensional component by \(D^t\), by ordinary projective
B\'ezout.  Taking the minimum proves the asserted bound.

For the closed form, first replace \(b_a\) by \(D\).
If exactly \(k\) target blocks have positive indices, they can be
chosen in \(\binom tk\) ways.  The number of positive \(k\)-tuples
whose sum is at most \(a\) is \(\binom ak\), and each has weight
\(E^{t-k}D^k\).  Restoring \(b_a=E\) changes only the terms with
one index equal to \(a\) and all others zero.  There are \(t\)
such terms, giving the correction
\(t(D-E)E^{t-1}\) in \eqref{eq:cre-closed-budget}.
Finally, the term of \(\Gamma_{e,t}\) with every index zero is
\(E^t\), and \(D\geq E\).  Thus \(L_{e,t}\geq E^t\geq e\).
\end{proof}

\subsection{Rational points and comparison of thresholds}

\begin{proof}[Proof of Theorem~\ref{thm:cre-common-core}]
Since \(L_{e,t}\geq e\), the field-size hypothesis implies
\(q\geq e\).  Lemma~\ref{lem:cre-infinity-component} supplies
an absolutely irreducible affine \(F\)-component \(\mathcal C\)
of dimension \(a\), and Proposition~\ref{prop:cre-refined-budget}
bounds its degree by \(L_{e,t}\).  Also \(L_{e,t}^2\geq e\), so
\[
 q>\max\{2eL_{e,t}^{\,2},\,2L_{e,t}^{\,4}\}.
\]
Cafure--Matera~\cite[Corollary~7.4]{CafureMatera2006} therefore gives
an \(F\)-rational point on \(\mathcal C\).  Its locator slots solve
\eqref{eq:cre-common-core}, as required.
\end{proof}

\begin{proof}[Proof of Theorem~\ref{thm:intro-componentwise}]
The hypothesis also implies the BCH rank condition, since
\[
 2^{\lceil m/2\rceil}\geq\sqrt q
 >L_{e,t}^{\,2}\geq(e!)^{2t}\geq2e-1.
\]
Every prescribed binary syndrome can therefore be identified with a
target in \(F^e\).  Apply Theorem~\ref{thm:cre-common-core}.
The distinct nonzero locators among its slots index at most
\((e-1)+et=(t+1)e-1\) parity-check columns.  Repeated locators
cancel in characteristic two, so the binary span of these columns
contains every target.  This proves the upper bound.
For the lower-bound comparison in the Introduction, note also that
\(L_{e,t}\geq(e!)^t\geq2^t\), whence the field-size hypothesis gives
\(m>4t+1\) and in particular \(t\leq m\).
\end{proof}

For comparison with the earlier degree calculation, define
\[
 p_e^{\mathrm{old}}(w)
 =E+D\sum_{u=1}^{a}\binom au\frac{w^u}{u!},
 \qquad
 \Delta^{\mathrm{mix}}_{e,t}
 =a!\sum_{v=0}^{a}\frac{[w^v](p_e^{\mathrm{old}}(w))^t}{(a-v)!},
\]
and put
\[
 \widehat B_{e,t}
 =\min\{D^t,\Delta^{\mathrm{mix}}_{e,t}\}.
\]
The sum \(\Delta^{\mathrm{mix}}_{e,t}\) is the exact mixed-degree
sum occurring before the simplification in the proof of
Lemma~\ref{lem:refined-degree-budget}.  Its one-target factors with
positive index are \(D\binom au\), and its terms carry the Segre
multinomial coefficients.  The present one-target factors are no larger,
and the multicone comparison removes those coefficients.  Termwise
comparison, followed by the earlier simplification, gives
\[
 L_{e,t}\leq\widehat B_{e,t}\leq B_{e,t}\leq D_e^t.
\]

The following table gives the least integer \(m\) satisfying
\(2^m>2B^4\) for each displayed degree budget \(B\).
These are sufficient extension degrees for the same support count,
not assertions of minimality.
\begin{center}
\begin{tabular}{ccrrr}
\hline
\(e\)&\(t\)&\(D_e^t\)&\(\widehat B_{e,t}\)&\(L_{e,t}\)\\
\hline
3&8&127&126&115\\
4&8&216&216&199\\
8&16&1342&1321&1193\\
\hline
\end{tabular}
\end{center}
For fixed \(e\geq3\), the closed form also gives
\[
 \Gamma_{e,t}\sim
 E^t\frac{(D/E)^a}{a!}t^a
 \qquad(t\longrightarrow\infty).
\]
Thus the sufficient extension degree still has leading term
\(4t\log_2(e!)\), while the constant in the degree budget is reduced.

The case \(e=2\) is included for uniformity.  In that case
\(\Gamma_{2,t}=(t+1)2^t\), but the specialized estimate of
Xiong--Yip~\cite[Theorem~IV.2]{XiongYip2026} already gives
\(R_t(C_{2,m})\leq2t+1\) for
\(m\geq2t+3+2\log_2(t-1)\), a stronger sufficient threshold.
The comparison above concerns the earlier general degree budgets.

\section{The uniform higher-order theorem}
\label{sec:uniform-order}

This section retains the completion-cover approach.  Unlike
Theorem~\ref{thm:intro-componentwise}, it controls connected simultaneous
covers through their sign fields.  Its independent purpose is to describe
those fields and the cost of regularizing a syndrome space.

\begin{theorem}[Uniform completion-cover bound]
\label{thm:intro-general-order}
Let \(e,t\geq2\).  If
\[
 2^m>\max\{2eB_{e,t}^{2},\,2B_{e,t}^{4}\},
\]
then
\[
 \sum_{j=0}^{t-1}\left\lceil\frac{2e-1}{2^j}\right\rceil
 \leq R_t(C_{e,m})\leq (t+1)e-1+c_{e,m,t}.
\]
For \(2\leq e\leq6\), the upper bound improves to
\[
 R_t(C_{e,m})\leq(t+1)e-1.
\]
If, in addition,
\[
 2^m\geq\frac{2^{t(et-1)}}{(et-1)!},
\]
then the lower bound may be replaced by the maximum of the displayed
Griesmer sum and \(et\).
\end{theorem}

We now prove this completion-cover bound.  By
\eqref{eq:Rt-subspaces}, the coding problem is to support an arbitrary binary
subspace \(W\subseteq F^e\) of dimension \(r\leq t\).  The proof follows one
fixed route.  We first split off the directions of \(W\) that are themselves
BCH columns.  On a complementary space, the lower-coordinate projection
\(\pi_e\) separates a pure kernel from targets with nonzero lower projection.
A matroid-intersection argument retains as many pure basis directions as the
power test can distinguish, and a small set of shared BCH columns regularizes
the remaining directions.

After this elementary preparation, every basis target satisfies the deferred
mixed common-core criterion, Proposition~\ref{prop:higher-order-common-core}.
That criterion, proved in Section~\ref{sec:higher-framework}, realizes all
basis targets using one common \((e-1)\)-tuple and one completing \(e\)-tuple
per target.  Lemma~\ref{lem:common-core-certificate} then returns to the
original support problem with \((r+1)e-1\) columns; finally we restore the
columns split off or used for regularization and maximize over \(W\).

\subsection{Regularizing an arbitrary syndrome space}

Call a syndrome \emph{admissible} if it is nonzero and not one-column.  The
next two lemmas carry out the basis selection, after which we state the
common-core certificate and the deferred criterion.

\begin{lemma}[Shared-column regularization]
\label{lem:regularize-syndrome-space}
Let \(e,r\geq2\), let \(W\subseteq F^e\) have a basis
\(\bfs_1^\T,\ldots,\bfs_r^\T\), and let
\(J\subseteq\{1,\ldots,r\}\) have size \(1\leq b\leq r-1\).  Suppose that
the targets indexed outside \(J\) are admissible and that \(q>2^{2r}\).
Put \(c=\lceil\log_2(b+1)\rceil\).  There are nonzero
\(z_1,\ldots,z_c\in F\) and distinct nonzero labels
\(\varepsilon_i\in\F_2^c\), \(i\in J\), such that, on putting
\[
 A\eqdef\Span_{\F_2}\{h_e(z_1),\ldots,h_e(z_c)\},
 \qquad
 V\eqdef\pi_e(A),
\]
one has \(A\cap W=0\), \(V\cap\pi_e(W)=0\), and the targets
\[
 \widetilde\bfs_i^\T\eqdef
 \begin{cases}
  \bfs_i^\T+\displaystyle\sum_{j=1}^{c}\varepsilon_{i,j}h_e(z_j),
      &i\in J,\\[3pt]
  \bfs_i^\T,&i\notin J,
 \end{cases}
\]
form an admissible basis.  The corrected targets have nonzero pairwise
distinct lower projections, none equal to the lower projection of an
uncorrected target.
\end{lemma}

\begin{proof}
Write
\[
 \lambda_e(z)\eqdef\pi_e(h_e(z))=(z,z^3,\ldots,z^{2e-3})^\T.
\]
The map \(\lambda_e:F\to F^{e-1}\) is injective.  Choose distinct nonzero
labels \(\varepsilon_i\in\F_2^c\), and let \(g(i)\) be the largest index of
a nonzero coordinate of \(\varepsilon_i\).  We construct the \(z_j\)
successively.  At stage \(j\), put
\[
 V_{j-1}\eqdef
 \Span_{\F_2}\{\lambda_e(z_1),\ldots,\lambda_e(z_{j-1})\}
\]
and exclude every \(z\) for which
\[
 \lambda_e(z)\in\pi_e(W)+V_{j-1}.
\]
There are at most \(2^{\dim\pi_e(W)+j-1}\) such values.  This choice makes
the \(\lambda_e(z_j)\) independent and maintains
\(V_j\cap\pi_e(W)=0\).

For each \(i\in J\) with \(g(i)=j\), set
\[
 \bfu_i^\T\eqdef
 \bfs_i^\T+\sum_{k<j}\varepsilon_{i,k}h_e(z_k).
\]
This syndrome is nonzero: otherwise it would lie in the intersection of
\(W\) with the span of the preceding correction columns, and projection to
the direct sum \(\pi_e(W)\oplus V_{j-1}\) would make both parts zero.
For any fixed nonzero syndrome
\(\bfu^\T=(u_1,u_2,\ldots,u_e)^\T\), at most two values of \(z\) make
\(\bfu^\T+h_e(z)\) zero or one-column.  Indeed, allowing \(w=0\) to
represent the zero adjusted target, the equality
\(\bfu^\T+h_e(z)=h_e(w)\) gives \(w=u_1+z\).  If \(u_1\neq0\), comparison
of cubic coordinates gives the nonzero quadratic
\[
 u_1z^2+u_1^2z+u_1^3+u_2=0.
\]
If \(u_1=0\), the first coordinate gives \(w=z\), and the full equality
would force \(\bfu^\T=0\).  Thus the additional exclusions at stage \(j\)
number at most twice the number of labels with highest bit \(j\).
Since \(c\leq r-1\), the total number excluded at any stage is at most
\[
 2^{2r-2}+2b<2^{2r}<q.
\]
Because \(j=g(i)\) is the highest nonzero bit of \(\varepsilon_i\), this
stage fixes the final adjusted target indexed by \(i\); hence every corrected
target is admissible at the end of the construction.

At the end, projection shows \(A\cap W=0\), so the adjusted targets remain
independent.  Each corrected lower projection is the sum of an element of
\(\pi_e(W)\) and a nonzero vector of \(V\); hence it lies outside
\(\pi_e(W)\).  Distinct labels, the independence of the
\(\lambda_e(z_j)\), and \(V\cap\pi_e(W)=0\) show that corrected lower
projections cannot coincide.  This also separates them from every
uncorrected lower projection and completes the proof.
\end{proof}

The number of basis elements that need correction is controlled by a
standard matroid-intersection argument.

\begin{lemma}[Power-basis defect]
\label{lem:power-basis-defect}
Let \(U\subseteq F\) be a binary subspace of dimension \(u\geq1\), and put
\(g=\gcd(k_e,q-1)\).  There are linearly independent
\(\gamma_1,\ldots,\gamma_\ell\in U\) whose \(k_e\)-th powers are also
linearly independent, where
\[
 \ell\geq
 \left\lceil\log_2\left(1+\frac{2^u-1}{g}\right)\right\rceil.
\]
\end{lemma}

\begin{proof}
On the ground set \(U^*=U\setminus\{0\}\), take the two vector matroids
defined respectively by the vectors \(\gamma\) and \(\gamma^{k_e}\).
For a subset \(A\subseteq U^*\), let \(a\) be its rank in the first matroid
and let \(b\) be the rank of \(U^*\setminus A\) in the second.  The power
map on \(F^*\) has fibers of size \(g\), so
\[
 2^u-1\leq(2^a-1)+g(2^b-1)
 \leq g(2^{a+b}-1).
\]
Edmonds' matroid-intersection theorem~\cite{Edmonds1979MatroidIntersection}
identifies the largest common independent-set size with the minimum of
\(a+b\) over these partitions, and the asserted bound follows.
\end{proof}

\subsection{The common-core criterion and the uniform theorem}

After regularization, we seek one \((e-1)\)-tuple used in every
representation and one completing \(e\)-tuple for each target.

\begin{lemma}[Common-core certificate]
\label{lem:common-core-certificate}
Let \(r\geq2\), and let \(W\subseteq F^e\) have a basis
\(\bfs_1^\T,\ldots,\bfs_r^\T\).  Suppose that there are tuples
\(\mathbf z=(z_1,\ldots,z_{e-1})\in F^{e-1}\) and
\(\mathbf x_i=(x_{i,1},\ldots,x_{i,e})\in F^e\), \(1\leq i\leq r\),
such that
\[
 \bfs_i^\T=
 \sum_{u=1}^{e-1}h_e(z_u)+
 \sum_{\ell=1}^{e}h_e(x_{i,\ell})
 \qquad(1\leq i\leq r).
\]
Then \(\mu_e(W)\leq(r+1)e-1\).
\end{lemma}

\begin{proof}
All basis syndromes are supported on
\[
 E(\mathbf z)\cup\bigcup_{i=1}^{r}E(\mathbf x_i).
\]
The reduced-support convention gives
\[
 \left|E(\mathbf z)\cup\bigcup_{i=1}^{r}E(\mathbf x_i)\right|
 \leq(e-1)+re=(r+1)e-1.
\]
\end{proof}

For \(r\geq2\), recall the degree budget \(B_{e,r}\) defined in the
introduction.  Both terms in its definition are nondecreasing in \(r\), so
\(B_{e,r}\leq B_{e,t}\) whenever \(r\leq t\).  We shall also use the
elementary bounds
\[
 B_{e,r}\geq D_e^2,
 \qquad
 B_{e,r}\geq(e!)^r.
\]
Indeed, \(D_e/e!\geq1\); for \(e=2\) the first inequality is immediate,
and for \(e\geq3\) it follows from
\((1+2(e-1)D_e/e!)^{e-1}\geq(D_e/e!)^2\).

\begin{proposition}[Mixed common-core criterion]
\label{prop:higher-order-common-core}
Let \(e,r\geq2\), and let \(W\subseteq F^e\) be an \(r\)-dimensional
binary syndrome space with an admissible basis consisting of
\[
 \bfs_{\gamma_1}^\T,\ldots,\bfs_{\gamma_u}^\T,
 \qquad
 \bft_1^\T,\ldots,\bft_v^\T,
 \qquad u+v=r,
\]
where the pure parameters
\(\gamma_1,\ldots,\gamma_u\) are linearly independent and either
\(2\leq e\leq6\) or
\(\gamma_1^{k_e},\ldots,\gamma_u^{k_e}\) are linearly independent.  The
lower projections
\[
 \pi_e(\bft_1^\T),\ldots,\pi_e(\bft_v^\T)
\]
are nonzero and pairwise distinct.  Empty parts of the basis are allowed.  If
\begin{equation}
\label{eq:higher-order-threshold}
 q>\max\{2eB_{e,r}^2,\,2B_{e,r}^4\},
\end{equation}
then
\[
 \mu_e(W)\leq(r+1)e-1.
\]
\end{proposition}

The proof is deferred to Section~\ref{sec:higher-framework}.

The field-size hypothesis below gives \(q>2B_{e,t}^4\geq2D_e^8\).
As in the proof of Theorem~\ref{r2:thm:effective-general-errors}, this
implies Cohen's condition~\eqref{r2:eq:ordinary-radius}.  We may therefore
use \(R_1(C_{e,m})=2e-1\) in the one-dimensional cases of the proof.

\begin{proof}[Proof of Theorem~\ref{thm:intro-general-order}]
Since \(D_e\geq2e-1\) and \(D_e\geq3\), the field-size hypothesis gives
\[
 2^{\lceil m/2\rceil}\geq\sqrt q>B_{e,t}^2\geq D_e\geq2e-1,
 \qquad
 2^m>2(e!)^{4t}>2^t.
\]
Thus \(m\geq t\) and Condition~\eqref{eq:rank-condition} holds; hence
Proposition~\ref{prop:lower-bound} gives the lower bound.
The field-size hypothesis also gives \(q-1\geq et-1\), so under the additional
counting hypothesis
Proposition~\ref{prop:sphere-lower-bound} gives the asserted \(et\) lower
bound.

For the upper bound, let \(W\subseteq F^e\) have dimension \(r\leq t\).
The case \(r=0\) is immediate, and the cited ordinary-radius result handles
\(r=1\).  Suppose that \(r\geq2\).  Choose a maximal linearly independent
family of one-column syndromes in \(W\), let \(P\) be its span, and choose a
complement \(W_0\), so
\[
 W=P\oplus W_0.
\]
By maximality, \(W_0\) contains no nonzero one-column syndrome, and
\[
 \mu_e(W)\leq\dim P+\mu_e(W_0).
\]
If \(W_0=0\), then \(\mu_e(W)\leq r\).  If \(\dim W_0=1\), the
ordinary-radius formula gives
\[
 \mu_e(W)\leq\dim P+2e-1\leq(r+1)e-1.
\]
Thus both cases satisfy the desired estimate.  Assume that
\(s\eqdef\dim W_0\geq2\).

Let \(U=\ker(\pi_e|_{W_0})\), identified with a binary subspace of \(F\)
through \(\gamma\mapsto\bfs_\gamma^\T\), and put \(u=\dim U\).  If
\(2\leq e\leq6\), choose any basis of \(U\) and complete it by lifting a
basis of \(\pi_e(W_0)\).  Proposition~\ref{prop:higher-order-common-core}
 applies because \(B_{e,s}\leq B_{e,t}\), and gives
\[
 \mu_e(W)\leq\dim P+(s+1)e-1\leq(r+1)e-1.
\]
This proves the sharper low-error assertion.  We may therefore assume
\(e\geq7\).

If
\(u>0\), Lemma~\ref{lem:power-basis-defect} allows a basis of \(U\) to be
chosen so that all but at most
\[
 b_u\eqdef u-
 \left\lceil\log_2\left(1+\frac{2^u-1}{g_{e,m}}\right)\right\rceil
\]
of its elements have linearly independent \(k_e\)-th powers; put \(b_0=0\).
The function
\(b_u\) is nondecreasing in \(u\): the quantity inside the logarithm for
\(u+1\) is at most twice the corresponding quantity for \(u\).  Hence
\[
 b_u\leq t-\ell_{e,m,t}.
\]
Complete this pure basis by lifting a basis of \(\pi_e(W_0)\).  The lifted
targets have nonzero pairwise distinct lower projections, and every chosen
target is admissible because \(W_0\) has no one-column syndrome.

If \(b_u>0\), apply Lemma~\ref{lem:regularize-syndrome-space} only to these
\(b_u\) exceptional pure targets.  Its field-size hypothesis holds because
\[
 q>2B_{e,t}^4\geq2(e!)^{4t}>2^{2t}\geq2^{2s}.
\]
The corrected exceptional targets have nonzero pairwise distinct lower
projections outside \(\pi_e(W_0)\).  Thus the uncorrected pure targets have
independent \(k_e\)-th powers, while every remaining target has a nonzero
lower projection distinct from all the others.  Proposition
 ~\ref{prop:higher-order-common-core}, with \(B_{e,s}\leq B_{e,t}\),
supports the adjusted basis on at most
\((s+1)e-1\) columns.  Adding back the shared correction columns costs at
most
\[
 \left\lceil\log_2(b_u+1)\right\rceil\leq c_{e,m,t}.
\]
The same conclusion holds with no correction when \(b_u=0\).  Therefore,
writing \(a=\dim P=r-s\),
\[
 \mu_e(W)\leq a+(s+1)e-1+c_{e,m,t}
 \leq(r+1)e-1+c_{e,m,t}.
\]
Equation~\eqref{eq:Rt-subspaces} proves the upper bound.
\end{proof}

\begin{remark}[The second- and third-order specializations]
At \(t=2\), the degree budget is \(D_e^2\), the field-size condition is
exactly the universal condition in the second-order theorem above, and
Theorem~\ref{thm:intro-general-order} gives
\[
 3e-1\leq R_2(C_{e,m})\leq3e.
\]
The dedicated second-order argument additionally supplies the sharper
root-free criterion expressed by \(G_e\).  At \(t=3\), the general theorem gives
\[
 R_3(C_{e,m})\leq4e-1+c_{e,m,3},
 \qquad c_{e,m,3}\leq2.
\]
For \(2\leq e\leq6\), the low-error refinement removes \(c_{e,m,3}\) altogether
and gives the numerical bounds in Theorem~\ref{thm:intro-main}.
\end{remark}

\subsection{Third-order thresholds}
\label{subsec:third-order-thresholds}

\begin{lemma}
\label{lem:numerical-thresholds}
Condition~\eqref{eq:higher-order-threshold} with \(r=3\) holds for
\(e=3,4,5,6\) whenever, respectively,
\[
 m\geq48,\qquad m\geq82,\qquad m\geq120,
 \qquad m\geq162.
\]
\end{lemma}

\begin{proof}
For \(e=3\), \(B_{3,3}=3375\), and
\[
 2^{47}<2(3375)^4<2^{48}.
\]
For \(e=4\), \(B_{4,3}=1{,}157{,}625\), and
\[
 2^{81}<2(1{,}157{,}625)^4<2^{82}.
\]
For \(e=5\), \(B_{5,3}=843{,}908{,}625\), and
\[
 2^{119}<2(843{,}908{,}625)^4<2^{120}.
\]
For \(e=6\), \(B_{6,3}=1{,}123{,}242{,}379{,}875\), and
\[
 2^{161}<2(1{,}123{,}242{,}379{,}875)^4<2^{162}.
\]
In all four cases the fourth-power term is the maximum.
\end{proof}

Together with Theorem~\ref{thm:intro-general-order}, these four calculations
give the third-order bounds stated in Theorem~\ref{thm:intro-main}. The next
section records these consequences and the sharper pure-space result before
Section~\ref{sec:higher-framework} proves the one deferred common-core input.

\section{Third-order bounds and pure spaces}
\label{sec:third-order}

We apply the uniform bound to three through six errors. For three errors,
a separate construction improves the upper bound on every pure
three-dimensional syndrome space to the lower bound ten.

\subsection{Three errors and the exact pure-space value}
\label{sec:e3}

Recall that
\[
 h_3(x)=(x,x^3,x^5)^\T,\qquad
  \pi_3(s_1,s_2,s_3)^\T=(s_1,s_2)^\T,
\qquad D_3=15.
\]
\begin{proof}[Proof of the \(e=3\) assertion in
Theorem~\ref{thm:intro-main}]
For \(e=3\) and \(r=t\),
Condition~\eqref{eq:higher-order-threshold} is precisely the field-size
hypothesis of Theorem~\ref{thm:intro-general-order}.
Its low-error clause therefore gives
\[
 R_t(C_{3,m})\leq3t+2.
\]
At \(t=3\), Proposition~\ref{prop:lower-bound} gives the lower bound 10,
and Lemma~\ref{lem:numerical-thresholds} shows that the point criterion holds
for every \(m\geq48\).  This proves the \(e=3\) assertion of
Theorem~\ref{thm:intro-main}.
\end{proof}

When \(m\) is even, a pure binary plane may be a scaled copy of \(\F_4\),
so its three nonzero elements have the same cube.  The calculation in
Lemma~\ref{lem:pure-sign-independence-small-e} shows why this leading-term
obstruction does not require an extra support column.

The pure case admits a sharper construction than the uniform common-core
bound.

\begin{proposition}[Exact pure spaces for three errors]
\label{prop:pure-e3-exact-ten}
Let \(m\geq18\), let \(\Gamma\subseteq F\) be a three-dimensional binary
subspace, and put
\[
 W_\Gamma\eqdef\{\bfs_\gamma^\T:\gamma\in\Gamma\}\subseteq F^3.
\]
Then
\[
 \mu_3(W_\Gamma)=10.
\]
\end{proposition}

\begin{proof}
 Suppose that a set \(I\) of \(n\) BCH columns supports \(W_\Gamma\), and
 define the selected-column syndrome map
 \[
  \sigma_I:\F_2^I\longrightarrow F^3,
  \qquad
  (\varepsilon_x)_{x\in I}\longmapsto
  \sum_{x\in I}\varepsilon_xh_3(x).
 \]
 Since \(W_\Gamma\subseteq\operatorname{im}\sigma_I\), choose a binary linear
 section \(\rho:W_\Gamma\to\F_2^I\) with
 \(\sigma_I\circ\rho=\operatorname{id}_{W_\Gamma}\).  Its image
 \(\rho(W_\Gamma)\) is a binary \([n,3]\) subcode all of whose words have zero
 first and third moments.  After zero-extension outside \(I\), it is a
 length-\(q-1\) subcode of \(C_{2,m}\).  Every nonzero word therefore has
weight at least \(5\), by the designed-distance bound for the primitive
double-error-correcting BCH code.  The binary Griesmer bound
\cite{Griesmer1960} gives
\[
 n\geq5+\left\lceil\frac52\right\rceil
       +\left\lceil\frac54\right\rceil=10.
\]
The reverse inequality is the ten-locator construction in
Proposition~\ref{prop:pure-e3-ten-locator}.
\end{proof}

The construction uses a different ten-column incidence pattern rather than
the uniform two-locator core.  Its proof is deferred to
Appendix~\ref{app:pure-e3-exact-ten} so that the main argument remains in
syndrome-support form.

\begin{remark}[Location of a possible value eleven]
\label{rem:e3-deep-obstruction}
Let \(m\geq5\), and let \(W\subseteq F^3\) be three-dimensional.  If \(W\)
contains a syndrome supported by one BCH column, then
\(\mu_3(W)\leq9\); if it contains one supported by two columns, then
\(\mu_3(W)\leq10\).  Indeed, split off that syndrome and support a binary
complement of dimension two using
\(R_2(C_{3,m})=8\)~\cite{EssayagZabokritskiy2026BCH3}.  Hence a space that
attains the possible value \(11\) in Theorem~\ref{thm:intro-main} must be
nonpure and every one of its nonzero syndromes must have ordinary coset
weight at least three, meaning that it requires at least three BCH columns.

This restriction is nonvacuous.  Suppose that \(m\geq9\) is odd and that
\(U\subseteq F\) is any three-dimensional binary subspace.  For
\[
 W_U\eqdef\{(0,u,0)^\T:u\in U\},
\]
the restriction of \(\pi_3\) to \(W_U\) has rank three, and every nonzero
 syndrome has ordinary coset weight exactly five.  Indeed, weight one is
 immediately impossible, while the published zero-first-coordinate argument excludes
 weights two and four
 \cite[proof of Lemma~1]{CharpinHellesethZinoviev2006}.  A weight-three
 representation of \((0,u,0)^\T\) by distinct locators \(x,y,z\) would have
 \(x+y+z=0\), \(xyz=u\), and \(e_2(x,y,z)u=0\), where
 \(e_2(x,y,z)=xy+xz+yz\), by the third and fifth
Newton identities.  Hence \(x,y,z\) would be the three roots of
\(t^3+u\).  But cubing permutes \(F^*\) when \(m\) is odd.  The known
ordinary radius \(R_1(C_{3,m})=5\), valid for \(m\geq5\)
\cite[Section~I]{CharpinHellesethZinoviev2006}, supplies the matching upper
bound.
\end{remark}

\subsection{Four through six errors}
\label{sec:e4}
\label{sec:e56}

For four errors the common-core construction uses three common locators
and four completing locators per target. Recall that
\[
 h_4(x)=(x,x^3,x^5,x^7)^\T,\qquad
 \pi_4(s_1,s_2,s_3,s_4)^\T=(s_1,s_2,s_3)^\T,
\qquad D_4=105.
\]
\begin{proof}[Proof of the \(e=4\) assertion in
Theorem~\ref{thm:intro-main}]
The low-error clause of Theorem~\ref{thm:intro-general-order} gives
\[
 R_t(C_{4,m})\leq4t+3
\]
under its field-size hypothesis.  At \(t=3\),
Proposition~\ref{prop:lower-bound} gives 13, and
Lemma~\ref{lem:numerical-thresholds} gives the range \(m\geq82\).  This proves
the \(e=4\) assertion of Theorem~\ref{thm:intro-main}.
\end{proof}

When \(4\mid m\), a pure binary three-space contained in
\(\gamma\F_{16}\) has fifth powers contained in \(\gamma^5\F_4\).  Thus the
fifth-power leading coefficients cannot be independent.  As in the cubic
case, Lemma~\ref{lem:pure-sign-independence-small-e} shows that the lower
polar terms remove this obstruction from the simultaneous-completion
criterion.

The exact sign calculations also remove the correction term for five and
six errors.  Thus Theorem~\ref{thm:intro-general-order} gives
\[
 R_t(C_{5,m})\leq5t+4,
 \qquad
 R_t(C_{6,m})\leq6t+5
\]
under the corresponding field-size hypotheses.  At \(t=3\),
Proposition~\ref{prop:lower-bound} and
Lemma~\ref{lem:numerical-thresholds} give
\[
 17\leq R_3(C_{5,m})\leq19\quad(m\geq120),
 \qquad
 20\leq R_3(C_{6,m})\leq23\quad(m\geq162).
\]
This proves the remaining assertions of Theorem~\ref{thm:intro-main}.

\section{The second-order common-core engine}
\label{r2:sec:second-order-engine}

This section records, in self-contained form, the two-target engine used
later in the higher-order construction.  Its output is stated first:
under an explicit point-count threshold, two suitable syndrome targets admit
representations with the same \(e-1\) locators and with \(e\) further
locators for each target.  The proof then descends through the one-target
incidence variety, the completion cover, the sign-field comparison, and the
finite boundary cases before returning to BCH support.

\paragraph{Recall.}
Throughout, \(F=\F_q\) has characteristic two, \(e\geq2\),
\[
 h_e(u)=(u,u^3,\ldots,u^{2e-1})^\T,\qquad h_e(0)=\bm 0^\T,
\]
and \(\pi_e:F^e\to F^{e-1}\) deletes the highest syndrome coordinate.
The quantity \(\mu_e(W)\) is the least size of a set of nonzero BCH locators
whose binary span contains \(W\).  A nonzero target is \emph{one-column} if
it equals \(h_e(u)\) for some \(u\in F^*\).  For
\(\mathbf x=(x_1,\ldots,x_r)\), recall also that \(E(\mathbf x)\) is its
reduced binary support; hence zeros and even repetitions only decrease the
resulting support.

The second-order common-core criterion was stated in
Proposition~\ref{r2:prop:common-core-plane}, and the elementary basis choice
needed at the final return was proved in
Lemma~\ref{r2:lem:projection-basis}.  Once the common locators are found,
Lemma~\ref{lem:common-core-certificate} supplies the support count.  We now
prove the geometric input behind the criterion.

\subsection{The one-target incidence variety}
\label{r2:subsec:incidence}

We begin with one of the two targets.  For a prescribed non-one-column
syndrome \(\bfs^\T\), the question is which choices of \(e-1\) common
locators admit \(e\) completing locators.  The incidence variety below
parametrizes exactly those core--completion pairs.  Its irreducibility will
give the connected completion family needed in the next stage.

Fix a nonzero syndrome
\[
 \bfs^\T=(\sigma_1,\sigma_3,\ldots,\sigma_{2e-1})^\T\in F^e
\]
which is not a one-column syndrome, equivalently
\(\bfs^\T\neq h_e(u)\) for every \(u\in F^*\).  We use
\(e-1\) ordered common locators
\[
        \mathbf{z}\eqdef(z_1,\ldots,z_{e-1})
\]
and \(e\) ordered completing locators
\[
        \mathbf{x}\eqdef(x_1,\ldots,x_e).
\]
For any field \(k\) and every \(r\geq0\), the notation
\(\mathbb A_k^r\) denotes affine \(r\)-space over \(k\); its \(k\)-rational
points form \(k^r\).
The corresponding incidence variety is
\[
 \mathcal V_{\bfs}\eqdef
 \left\{
 (\mathbf z,\mathbf x)\in\mathbb A_F^{\,2e-1}:
 \sum_{i=1}^{e-1}h_e(z_i)
 +\sum_{j=1}^{e}h_e(x_j)
 =\bfs^\T
 \right\}.
\]
It is defined by the \(e\) odd-moment equations of degrees
\[
        1,3,\ldots,2e-1.
\]

The geometric input used below comes from the following published
irreducibility result for long BCH incidence varieties.

\begin{proposition}[BCH incidence variety]
\label{r2:prop:incidence-top-component}
Let \(\bfs^\T\in F^e\) be nonzero and not a one-column syndrome.  Then
\(\mathcal V_{\bfs}\) is geometrically irreducible of dimension \(e-1\).
In particular, its unique geometric component is defined over \(F\).
\end{proposition}

\begin{proof}
Let \(\mathcal P_{\bfs}\subseteq\mathbb P_F^{\,2e-1}\) be the projective
closure defined by
\[
 \sum_{i=1}^{2e-1}u_i^{2j-1}
 =\sigma_{2j-1}u_0^{2j-1},
 \qquad 1\leq j\leq e.
\]
With \(n=2e-1\) and \(t=e\), the cited results
\cite[Propositions~6 and~7]{VladutSkorobogatov1989} show that this variety is geometrically
irreducible of dimension \(e-1\), except when
\(\sigma_{2j-1}=\sigma_1^{2j-1}\) for every \(2\leq j\leq e\).
The exceptional identities say precisely that the target is zero or
one-column, both excluded here.  To control the affine chart, let
\(\mathcal Y_e=\mathcal P_{\bfs}\cap\{u_0=0\}\).  The zero-target variety
is the projective cone over \(\mathcal Y_e\) with vertex
\([1:0:\cdots:0]\), and \(\mathcal Y_e\) is nonempty because
\([0:1:1:0:\cdots:0]\in\mathcal Y_e\).  Proposition~6 gives the cone
dimension \(e-1\), hence \(\dim\mathcal Y_e=e-2\).  Therefore
\(\mathcal P_{\bfs}\) is not contained at infinity.
Its affine chart \(u_0=1\), which is \(\mathcal V_{\bfs}\), is therefore a
nonempty dense geometrically irreducible open subset of dimension \(e-1\).
\end{proof}

\subsection{Completion polynomials and covers}
\label{r2:subsec:completion-covers}

To describe how completions vary with the common core, we regard the \(e\)
completing locators as the roots of one degree-\(e\) polynomial.  On the locus where the
associated Hankel matrix is nonsingular, those roots are distinct and their
orderings form the finite connected cover that will later be compared for the
two syndrome targets.

For a fixed common core \(\mathbf z\), define the residual odd moments
\[
 p_{2j-1}(\bfs,\mathbf z)
 \eqdef
 \sigma_{2j-1}
 +\sum_{i=1}^{e-1}z_i^{2j-1},
 \qquad 1\leq j\leq e.
\]
For the rest of this construction the target \(\bfs\) is fixed, and we
suppress the arguments by writing
\(p_r=p_r(\bfs,\mathbf z)\in F[\mathbf z]\).  Complete the definition by
putting
\[
        p_0\eqdef e\bmod2,
\qquad
        p_{2j}\eqdef p_j^2
        \quad(1\leq j\leq e-1),
\]
where the second rule is read recursively.  Thus every symbol \(p_r\) used
below for \(0\leq r\leq2e-1\) denotes a specified regular residual-moment
function.  Form the Hankel matrix
\[
        \mathscr H_e(\bfs,\mathbf z)
        \eqdef
        (p_{i+j})_{0\leq i,j\leq e-1}.
\]
General finite-field counts for rank-bounded Hankel matrices, including
counts with a fixed initial segment, appear in
\cite[Theorems~1.1 and~1.5]{DwivediGrinberg2022}.  Here the matrix lies on
the much smaller characteristic-two Frobenius-moment locus
\(p_{2j}=p_j^2\), and the needed conclusion is reconstruction rather than
an ambient rank count.

The next reconstruction statement is the key characteristic-two step.  The
Frobenius relations force the usual weighted Hankel decomposition to consist
of distinct roots, each with weight one; this explains the appearance of the
symmetric group.

\begin{lemma}[Hankel reconstruction]
\label{r2:lem:hankel-reconstruction}
At every geometric point of the open set
\[
        \det \mathscr H_e(\bfs,\mathbf z)\neq0,
\]
the residual moments determine a unique monic separable polynomial of
degree \(e\).  Its roots are precisely the \(e\) completing locators.
Ordering those roots gives \(e!\) points above the common core.
\end{lemma}

We call this unique polynomial the \emph{completion polynomial} of the target
over the chosen common core.

\begin{proof}
Fix a geometric point of the displayed open set, and let \(k\) be its
algebraically closed residue field.  All moments in this proof are evaluated
at that point.  Since \(\mathscr H_e\) is nonsingular, there are unique coefficients
\(c_0,\ldots,c_{e-1}\in k\) satisfying
\[
 p_{e+r}+\sum_{\ell=0}^{e-1}c_\ell p_{\ell+r}=0,
 \qquad 0\leq r\leq e-1.
\]
Put
\[
 f(t)\eqdef
 t^e+c_{e-1}t^{e-1}+\cdots+c_1t+c_0,
 \qquad
 A_f\eqdef k[t]/(f),
 \qquad
 \theta\eqdef t\bmod f.
\]
Define a \(k\)-linear functional \(\lambda\colon A_f\to k\) by
\[
 \lambda(\theta^j)\eqdef p_j,
 \qquad 0\leq j\leq e-1.
\]
We first claim that
\[
 \lambda(\theta^j)=p_j,
 \qquad 0\leq j\leq2e-1.
\]
This holds by definition for \(j<e\).  Suppose inductively that it holds
through degree \(e+r-1\), where \(0\leq r\leq e-1\).  Since
\(f(\theta)=0\), and since \(k\) has characteristic two,
\[
 \lambda(\theta^{e+r})
 =
 \sum_{\ell=0}^{e-1}c_\ell\lambda(\theta^{\ell+r})
 =
 \sum_{\ell=0}^{e-1}c_\ell p_{\ell+r}
 =
 p_{e+r}.
\]
This proves the claim.

Consider the bilinear form
\[
 B\colon A_f\times A_f\longrightarrow k,
 \qquad
 B(a,b)\eqdef\lambda(ab).
\]
In the basis \(1,\theta,\ldots,\theta^{e-1}\), its Gram matrix is \(\mathscr H_e\),
so \(B\) is nondegenerate.  The Frobenius relations
\(p_{2j}=p_j^2\) imply
\[
 \lambda(a^2)=\lambda(a)^2
 \qquad(a\in A_f).
\]
Indeed, if \(a=\sum_{j=0}^{e-1}a_j\theta^j\), then
\[
 \lambda(a^2)
 =
 \sum_{j=0}^{e-1}a_j^2p_{2j}
 =
 \left(\sum_{j=0}^{e-1}a_jp_j\right)^2.
\]
If \(a^2=0\), then for every \(b\in A_f\),
\[
 B(a,b)^2
 =
 \lambda(ab)^2
 =
 \lambda\bigl((ab)^2\bigr)
 =
 \lambda(a^2b^2)
 =
 0.
\]
Thus \(B(a,b)=0\) for every \(b\), and nondegeneracy gives \(a=0\).
If \(A_f\) had a nonzero nilpotent element, a suitable iterated square of it
would be nonzero and have square zero.  Hence \(A_f\) is reduced.

Since \(A_f\) is a finite-dimensional reduced \(k\)-algebra,
\(\dim_k A_f=e\), and \(k\) is algebraically closed, one has
\[
 A_f\simeq k^e.
\]
Equivalently,
\[
 f(t)=\prod_{i=1}^{e}(t-x_i)
\]
for distinct \(x_1,\ldots,x_e\in k\).  Let
\(\varepsilon_1,\ldots,\varepsilon_e\) be the corresponding primitive
idempotents.  Since
\[
 \lambda(\varepsilon_i)^2
 =
 \lambda(\varepsilon_i^2)
 =
 \lambda(\varepsilon_i),
\]
each \(\lambda(\varepsilon_i)\) is either zero or one.  It cannot be zero:
if \(\lambda(\varepsilon_i)=0\), then
\(\varepsilon_i b\in k\varepsilon_i\) for every \(b\in A_f\), and hence
\(B(\varepsilon_i,b)=0\) for every \(b\), contrary to nondegeneracy.
Consequently,
\[
 \lambda(\varepsilon_i)=1
 \qquad(1\leq i\leq e).
\]
It follows that
\[
 p_j
 =
 \lambda(\theta^j)
 =
 \sum_{i=1}^{e}x_i^j,
 \qquad 0\leq j\leq2e-1.
\]
Thus every ordering of \(x_1,\ldots,x_e\) is a completing tuple.

Conversely, let \(y_1,\ldots,y_e\) be any completing tuple and put
\[
 q_j\eqdef\sum_{i=1}^{e}y_i^j.
\]
The incidence equations give \(q_j=p_j\) for every relevant odd \(j\), and
\(q_0=e\bmod2=p_0\).  Induction using
\(q_{2j}=q_j^2\) and \(p_{2j}=p_j^2\) therefore gives
\[
 q_j=p_j,
 \qquad 0\leq j\leq2e-1.
\]
Write
\[
 g(t)\eqdef\prod_{i=1}^{e}(t-y_i)
 =
 t^e+d_{e-1}t^{e-1}+\cdots+d_0.
\]
For \(0\leq r\leq e-1\), summing \(y_i^r g(y_i)=0\) over \(i\) yields
\[
 p_{e+r}+\sum_{\ell=0}^{e-1}d_\ell p_{\ell+r}=0.
\]
Thus \(d_0,\ldots,d_{e-1}\) solve the same nonsingular Hankel system as
\(c_0,\ldots,c_{e-1}\), and hence \(g=f\).  Therefore every completing
tuple is an ordering of the distinct roots of \(f\), and the geometric
fiber consists precisely of their \(e!\) orderings.
\end{proof}

Let
\[
 \operatorname{pr}_{\mathbf z}\colon
 \mathcal V_{\bfs}\longrightarrow\mathbb A_F^{e-1},
 \qquad
 (\mathbf z,\mathbf x)\longmapsto\mathbf z,
\]
and put
\[
 U_{\bfs}\eqdef
 \left\{\mathbf z\in\mathbb A_F^{e-1}:
 \det \mathscr H_e(\bfs,\mathbf z)\neq0\right\}.
\]
Define
\[
 \mathcal X_{\bfs}\eqdef
 \operatorname{pr}_{\mathbf z}^{-1}(U_{\bfs}),
 \qquad
 \rho_{\bfs}\eqdef
 \operatorname{pr}_{\mathbf z}|_{\mathcal X_{\bfs}}
 \colon\mathcal X_{\bfs}\longrightarrow U_{\bfs}.
\]
We call \(\rho_{\bfs}\) the \emph{completion cover} of \(\bfs^\T\).

\begin{proposition}[The completion cover]
\label{r2:prop:completion-cover}
Over the algebraic closure of \(F\), the completion cover is a connected
finite \'{e}tale Galois cover of degree \(e!\), with Galois group \(S_e\).
\end{proposition}

\begin{proof}
Work over \(k=\overline F\), and put
\[
 \delta\eqdef\det \mathscr H_e(\bfs,\mathbf z),
 \qquad
 A\eqdef k[\mathbf z,\delta^{-1}].
\]
The scheme-theoretic inverse image defining \(\mathcal X_{\bfs}\) has
coordinate algebra
\[
 B\eqdef
 A[x_1,\ldots,x_e]/(g_1,\ldots,g_e),
 \qquad
 g_j\eqdef
 \sum_{i=1}^{e}x_i^{2j-1}-p_{2j-1},
 \quad 1\leq j\leq e.
\]
Let \(c_0,\ldots,c_{e-1}\in A\) be the unique solution of
\[
 p_{e+r}+\sum_{\ell=0}^{e-1}c_\ell p_{\ell+r}=0,
 \qquad 0\leq r\leq e-1,
\]
and put
\[
 f(t)\eqdef t^e+\sum_{\ell=0}^{e-1}c_\ell t^\ell.
\]
These coefficients are regular on \(U_{\bfs}\), since the Hankel system is
solved after inverting \(\delta\).

In \(B\), the defining odd-moment equations and the Frobenius relations give
\[
 p_r=\sum_{i=1}^{e}x_i^r,
 \qquad 0\leq r\leq2e-1.
\]
Consequently,
\[
 \sum_{i=1}^{e}x_i^j f(x_i)=0,
 \qquad 0\leq j\leq e-1.
\]
The coefficient matrix of this system is the Vandermonde matrix
\[
 V=(x_i^j)_{\substack{0\leq j\leq e-1\\1\leq i\leq e}},
 \qquad
 (\det V)^2=\det \mathscr H_e=\delta\in B^\times.
\]
Hence \(V\) is invertible and \(f(x_i)=0\) for every \(i\).  Thus every
\(x_i\) is integral over \(A\), so \(B\) is finite over \(A\).  Moreover,
all \(x_i-x_j\) are units, and successive use of the factor theorem gives
\[
 f(t)=\prod_{i=1}^{e}(t-x_i)
 \qquad\text{in }B[t].
\]

The Jacobian determinant of the defining equations with respect to the
completion coordinates is
\[
 \det\bigl(x_i^{2j-2}\bigr)_{1\leq j,i\leq e}
 =\prod_{1\leq i<k\leq e}(x_i+x_k)^2
 =\delta,
\]
which is invertible.  Hence \(\rho_{\bfs}\) is finite \'{e}tale.  By
Lemma~\ref{r2:lem:hankel-reconstruction}, every geometric fiber over
\(U_{\bfs}\) consists precisely of the \(e!\) orderings of the distinct
completing roots.  Therefore \(\rho_{\bfs}\) is surjective and has constant
degree \(e!\).
In particular, \(B\) is reduced, so this scheme-theoretic inverse image
agrees with \(\mathcal X_{\bfs}\) viewed as a variety.

The open set \(U_{\bfs}\) is nonempty: the lower-moment map
\[
 \mathbf z\longmapsto
 \left(\sum_i z_i,\sum_i z_i^3,\ldots,
       \sum_i z_i^{2e-3}\right)
\]
has nonzero Vandermonde-square Jacobian and is therefore dominant, whereas
the Hankel determinant is a nonzero polynomial in the residual lower
moments.  Indeed, after the formal specialization
\(p_r=\sum_{i=1}^e t_i^r\), it becomes
\(\prod_{i<j}(t_i+t_j)^2\), which is nonzero.  Dominance now shows that
\(U_{\bfs}\), and hence its surjective inverse image
\(\mathcal X_{\bfs}\), is nonempty.
By Proposition~\ref{r2:prop:incidence-top-component}, it is a nonempty open
subset of the geometrically irreducible variety \(\mathcal V_{\bfs}\), and
is therefore geometrically irreducible and connected.  Permuting the
completion coordinates acts simply transitively
on every geometric fiber.  Thus the cover is a connected \(S_e\)-torsor,
equivalently a Galois cover with group \(S_e\).
\end{proof}

The one-target construction is now complete: every eligible syndrome gives a
connected \(S_e\)-cover of the common-core space.  We return to the original
two-target problem by placing the covers for \(\bm a^\T\) and \(\bm b^\T\) over
the same common-core coordinates and asking whether they are independent.
For a fixed nonzero, non-one-column target \(\bfs^\T\), let
\[
        K\eqdef\overline F(z_1,\ldots,z_{e-1})
\]
be the common-core function field, and let \(L_{\bfs}/K\) be the splitting
field of its generic completion polynomial.  Proposition~\ref{r2:prop:completion-cover}
identifies its Galois group with \(S_e\).

For two nonzero, non-one-column targets \(\bm a^\T,\bm b^\T\), put
\[
        U_{\bm a,\bm b}\eqdef U_{\bm a}\cap U_{\bm b}
\]
and restrict both completion covers to this common base.  Their simultaneous
completion space over the nonsingular locus is the fiber product
\[
 \left(\mathcal X_{\bm a}\times_{U_{\bm a}}U_{\bm a,\bm b}\right)
 \times_{U_{\bm a,\bm b}}
 \left(\mathcal X_{\bm b}\times_{U_{\bm b}}U_{\bm a,\bm b}\right).
\]
It is geometrically integral exactly when the completion fields
\(L_{\bm a}\) and \(L_{\bm b}\) are linearly disjoint over \(K\).  Under that
condition, an effective rational-point estimate produces one common core for
which both targets have completing locators;
Lemma~\ref{lem:common-core-certificate} then gives the desired support bound.
The next three subsections establish this linear disjointness in the cases
needed below.  The universal Hankel divisor detects a valuation at which the
quadratic sign subfield
\[
        Q_{\bfs}\eqdef L_{\bfs}^{A_e}
\]
ramifies; comparing that ramification distinguishes the completion fields.

\subsection{The universal Hankel divisor}
\label{r2:subsec:universal-hankel}

The next task is to prove that the two completion fields are linearly
disjoint.  When the targets have different lower projections, we seek a
valuation at which the quadratic sign invariant of one completion field
ramifies and that of the other does not.  The determinant of the residual
Hankel matrix supplies such a valuation.  We first analyze its divisor in
universal lower-moment coordinates; the target is then inserted by translation
and the result is pulled back to the common-core space.  Recall that \(\pi_e\)
deletes the highest syndrome coordinate \(\sigma_{2e-1}\).

From this point through Proposition~\ref{r2:prop:different-projections}, assume
\(e\geq4\).  The cases \(e=2,3\) are treated separately in
Lemma~\ref{r2:lem:boundary-completion-covers} before they are used in the
universal theorem.

For this universal statement, the symbols
\(p_1,p_3,\ldots,p_{2e-3}\) are algebraically independent indeterminates in
the polynomial ring
\[
        \mathcal R_e
        \eqdef
        \F_2[p_1,p_3,\ldots,p_{2e-3}].
\]
They are not the target-dependent functions used above.  Put
\(p_0\eqdef e\bmod2\), define \(p_{2j}\eqdef p_j^2\) recursively, and set
\[
 \mathscr H_e^{\mathrm{univ}}\eqdef(p_{i+j})_{0\leq i,j\leq e-1}.
\]
For each fixed
target and common core, the specialization homomorphism
\[
 \mathcal R_e\longrightarrow F[\mathbf z],
 \qquad
 p_{2j-1}\longmapsto p_{2j-1}(\bfs,\mathbf z),
\]
recovers the residual moments and sends \(\mathscr H_e^{\mathrm{univ}}\) to
\(\mathscr H_e(\bfs,\mathbf z)\).  All identities in the universal ring are used
below through this explicit specialization.  The determinant of
\(\mathscr H_e^{\mathrm{univ}}\) is a square in
this ring.  We choose its canonical square root \(\Delta_e\) by the
Vandermonde normalization in the next statement.  For the recursion appearing
there, adopt the base conventions
\[
        \Delta_0=\Delta_1=1,
        \qquad
        \Delta_2=p_1,
        \qquad
        \Delta_3=p_1^3+p_3.
\]

\begin{lemma}[Universal Hankel divisor]
\label{r2:lem:universal-hankel}
Let \(e\geq4\).  There is a polynomial
\[
        \Delta_e
        \in\F_2[p_1,p_3,\ldots,p_{2e-3}]
\]
with the following properties.

\begin{enumerate}
\item One has
      \[
              \det\mathscr H_e^{\mathrm{univ}}=\Delta_e^2.
      \]
      If the moments come from \(e\) roots
      \(x_1,\ldots,x_e\), then
      \[
              \Delta_e
              =
              \prod_{1\leq i<j\leq e}(x_i+x_j).
      \]
\item The polynomial \(\Delta_e\) is absolutely irreducible and reduced.
      Its zero set is the smallest subset of
      \(\mathbb A_{\overline F}^{e-1}\) defined by polynomial equations over
      \(\overline F\) and containing
      \[
       \left\{
       \left(
       \sum_{i=1}^{e-2}r_i,
       \sum_{i=1}^{e-2}r_i^3,\ldots,
       \sum_{i=1}^{e-2}r_i^{2e-3}
       \right):
       r_i\in\overline F
       \right\}.
      \]
      Thus \(V(\Delta_e)\) is determined by the polynomial relations
      satisfied by the lower moments of \(e-2\) roots.
\item For a suitable
      \(\Phi_e\in\F_2[p_1,p_3,\ldots,p_{2e-5}]\),
      \[
              \Delta_e
              =
              \Delta_{e-2}p_{2e-3}+\Phi_e,
      \qquad
              \frac{\partial\Delta_e}
              {\partial p_{2e-3}}
              =
              \Delta_{e-2}.
      \]
\item The divisor has no nonzero additive translation stabilizer:
      \[
       \operatorname{Stab}_{+}(V(\Delta_e))
       \eqdef
       \left\{
       \mathbf c\in\overline F^{\,e-1}:
       V(\Delta_e)+\mathbf c=V(\Delta_e)
       \right\}
       =
       \{\bm 0\}.
      \]
\end{enumerate}
\end{lemma}

The proof is deferred to Appendix~\ref{r2:app:universal-hankel}.  Its first
step embeds the lower-moment ring into the polynomial ring in the roots by a
Vandermonde-square Jacobian.  The transitivity of \(S_e\) on unordered
root pairs then proves absolute irreducibility.  Expanding the determinant
over involutions gives the recursion in part~(3).  The same expansion,
together with one triangular monomial whose coefficient is one, eliminates
successively every coordinate of a possible translation stabilizer.  These
ingredients are all that is required for the applications below.

Recall that \(\pi_e\) is the lower-moment projection, and define the
common-core moment map
\[
 \mathbf m_e\colon\mathbb A_F^{e-1}\longrightarrow\mathbb A_F^{e-1},
 \qquad
 \mathbf z\longmapsto
 \left(
 \sum_i z_i,
 \sum_i z_i^3,\ldots,
 \sum_i z_i^{2e-3}
 \right)^\T.
\]
Over \(\overline F\), its Jacobian determinant is the nonzero polynomial
\[
 J_e(\mathbf z)
 \eqdef
 \prod_{1\leq i<j\leq e-1}(z_i+z_j)^2.
\]
We call
\[
 \mathcal D_{\mathrm{coll}}\eqdef V(J_e)
 \subseteq\mathbb A_{\overline F}^{e-1}
\]
the common-core collision divisor.  It is exactly the locus where two core
locators collide and the moment map ceases to be \'{e}tale.

For a target \(\bfs^\T\), define its translated divisor in lower-moment
space and its pullback to common-core space by
\[
 \mathbf p\eqdef(p_1,p_3,\ldots,p_{2e-3})^\T,
 \qquad
 D_{\bfs}\eqdef
 V\!\left(\Delta_e\bigl(\pi_e(\bfs^\T)+\mathbf p\bigr)\right),
 \qquad
 \widetilde D_{\bfs}\eqdef\mathbf m_e^{-1}(D_{\bfs}).
\]
Lemma~\ref{r2:lem:universal-hankel} shows that
\(D_{\bm a}=D_{\bm b}\) exactly when
\(\pi_e(\bm a^\T)=\pi_e(\bm b^\T)\).

We now justify the divisor pullback used below.  Put \(n_0\eqdef e-1\), and
in the lower-moment target space form
\[
 \mathscr H_{\mathrm{core}}
 \eqdef(p_{i+j})_{0\leq i,j<n_0},
 \qquad
 p_0\eqdef n_0\bmod2,
 \qquad
 p_{2j}\eqdef p_j^2,
\]
and let
\[
 \delta_{\mathrm{core}}\eqdef\det \mathscr H_{\mathrm{core}},
 \qquad
 U_{\mathrm{core}}\eqdef V(\delta_{\mathrm{core}})^c.
\]
Apply the reconstruction argument of Lemma~\ref{r2:lem:hankel-reconstruction}
and the finiteness and Jacobian arguments in the proof of
Proposition~\ref{r2:prop:completion-cover} to these universal moments, with
\(e\) replaced by \(n_0\).  They show that
\[
 \mathbf m_e^{-1}(U_{\mathrm{core}})
 \longrightarrow U_{\mathrm{core}}
\]
is a finite \'{e}tale surjective cover of degree \(n_0!\): every point of
\(U_{\mathrm{core}}\) reconstructs a unique monic separable polynomial of
degree \(n_0\), and the fiber consists of all orderings of its roots.  Under
the root specialization,
\[
 \delta_{\mathrm{core}}=J_e(\mathbf z),
\]
so this cover is precisely the noncollision locus of the common-core moment
map.  No connectedness assertion is needed here.

The polynomial \(\delta_{\mathrm{core}}\) is independent of the last
lower-moment coordinate \(p_{2e-3}\).  By the recursion in
Lemma~\ref{r2:lem:universal-hankel}, the defining polynomial of \(D_{\bfs}\) is
linear in that coordinate with nonzero coefficient equal to a translate of
\(\Delta_{e-2}\).  Hence the irreducible divisor \(D_{\bfs}\) is not
contained in \(V(\delta_{\mathrm{core}})\), and
\[
 D_{\bfs}^{\circ}\eqdef D_{\bfs}\cap U_{\mathrm{core}}
\]
is dense and nonempty.  Base change of the preceding finite \'{e}tale
surjective cover to \(D_{\bfs}^{\circ}\) has a component dominating
\(D_{\bfs}^{\circ}\).  We denote its closure in
\(\widetilde D_{\bfs}\) by \(\Gamma_{\bfs}\).  It dominates
\(D_{\bfs}\), is generically disjoint from
\(\mathcal D_{\mathrm{coll}}\), and, after normalization, gives an unramified
extension of the discrete valuation at the generic point of \(D_{\bfs}\).

\subsection{The sign field and geometric disjointness}
\label{r2:subsec:sign-disjointness}

The universal calculation has produced a target-dependent divisor in the
common-core space.  Recall that
\(K=\overline F(z_1,\ldots,z_{e-1})\), that \(L_{\bfs}/K\) is the splitting
field of the generic completion polynomial of \(\bfs^\T\), and that
\[
        Q_{\bfs}=L_{\bfs}^{A_e}
\]
is its quadratic sign subfield.  We now prove that \(Q_{\bfs}\) ramifies
along the component \(\Gamma_{\bfs}\) lying over the translated Hankel
divisor.  This ramification distinguishes the completion fields attached to
targets with different lower projections.

In characteristic two, write the quadratic extension as
\(Q_{\bfs}=K(y_{\bfs})\), where
\(y_{\bfs}^2+y_{\bfs}=a_{\bfs}\).
Replacing \(y_{\bfs}\) by \(y_{\bfs}+u\), with \(u\in K\), replaces
\(a_{\bfs}\) by \(a_{\bfs}+u^2+u\).  Thus its Artin--Schreier class is
\([a_{\bfs}]\in K/\wp(K)\), where
\[
 \wp(K)\eqdef\{u^2+u:u\in K\}.
\]
At a divisorial valuation of \(K\), one may replace a representative modulo
\(\wp(K)\) to cancel any removable even leading polar terms.  If the
remaining representative has a pole of odd order, we call it an
\emph{odd reduced pole}.  Such a pole certifies a nonzero class, because
\(v(u^2+u)=2v(u)\) whenever \(v(u)<0\).

Fix a nonzero, non-one-column target \(\bfs^\T\) and the component
\(\Gamma_{\bfs}\) defined above.  Let \(\kappa(\Gamma_{\bfs})\) be its
function field.  At its generic point, the lower residual moments lie over
the generic point of \(V(\Delta_e)\) and, after passing to an algebraic
closure of \(\kappa(\Gamma_{\bfs})\), can be written using \(e-2\) pairwise
distinct roots
\[
        r_1,\ldots,r_{e-2},
\qquad
        D_R\eqdef\prod_{i<j}(r_i+r_j),
\]
and define the highest-moment defect by
\[
 \eta_{\bfs}\eqdef
 p_{2e-1}(\bfs,\mathbf z)
 -
 \sum_{i=1}^{e-2}r_i^{2e-1}
 =
 \sigma_{2e-1}
 +\sum_{j=1}^{e-1}z_j^{2e-1}
 -\sum_{i=1}^{e-2}r_i^{2e-1}.
\]
In characteristic two the minus sign equals addition, but this form makes
the role of the defect transparent.

\begin{lemma}[Generic sign ramification]
\label{r2:lem:generic-sign-ramification}
Let \(e\geq4\), and let \(\bfs^\T\) be nonzero and not a one-column
syndrome.  Then \(\eta_{\bfs}\) is nonzero at the generic point of
\(\Gamma_{\bfs}\).  After an \'{e}tale extension of the corresponding
 discrete valuation ring, put
\[
 \varpi\eqdef
 \Delta_e\bigl(\pi_e(\bfs^\T)+\mathbf m_e(\mathbf z)\bigr).
\]
There is a branch of the normalization of the corresponding discrete
valuation ring in \(L_{\bfs}\) on which \(e-2\) roots have the distinct finite
residues \(r_1,\ldots,r_{e-2}\) and two roots escape.  The
Artin--Schreier class of \(Q_{\bfs}/K\) has reduced leading polar part
\[
 \eta_{\bfs}^{\,2e-3}
 D_R^{\,2e-1}
 \varpi^{-(2e-1)}.
\]
In particular, \(Q_{\bfs}/K\) is ramified at this valuation.
\end{lemma}

The proof is the explicit local calculation in
Appendix~\ref{r2:app:local-sign}.  In particular, the appendix constructs the
branch on which two roots escape, identifies its
local quadratic extension with the global sign field, and proves the odd
polar order stated in the lemma.  Only the resulting ramification statement
is used below.

\begin{proposition}[Different lower projections]
\label{r2:prop:different-projections}
Let \(e\geq4\), and let \(\bm a^\T,\bm b^\T\) be nonzero syndromes, neither
of which is one-column.  If
\[
        \pi_e(\bm a^\T)\neq\pi_e(\bm b^\T),
\]
then
\[
        L_{\bm a}\cap L_{\bm b}=K.
\]
Equivalently, the two completion covers have a geometrically connected
fiber product.
\end{proposition}

\begin{proof}
By Lemma~\ref{r2:lem:universal-hankel}, the two translated irreducible Hankel
divisors are distinct.  At the generic point of
\(\Gamma_{\bm a}\), the second pulled-back divisor is absent and, by the
choice made above, so is \(\mathcal D_{\mathrm{coll}}\).
Lemma~\ref{r2:lem:generic-sign-ramification}
shows that \(Q_{\bm a}\) is ramified there, whereas the second
Artin--Schreier equation is \'{e}tale; hence
\[
        Q_{\bm a}\neq Q_{\bm b}.
\]
Put
\[
        M\eqdef L_{\bm a}\cap L_{\bm b}.
\]
Since both completion fields are Galois over \(K\), so is \(M/K\).  Suppose
that \(M\neq K\).  Restriction from either completion field gives a
surjection
\[
        S_e\longrightarrow\Gal(M/K)
\]
with a proper normal kernel.  For \(e\geq5\), the normal subgroups of
\(S_e\) are \(1,A_e,S_e\), so the possible nontrivial quotients are
\(S_e\) and \(C_2\).  For \(e=4\), the normal subgroups are
\[
        1,\ V_4,\ A_4,\ S_4,
\]
and the nontrivial quotients are \(S_4,S_3,C_2\), respectively.  Every
proper kernel in these lists is contained in \(A_e\).  Galois
correspondence therefore gives
\[
        Q_{\bm a}\subseteq M,
        \qquad
        Q_{\bm b}\subseteq M.
\]
Each possible group \(\Gal(M/K)\) has a unique subgroup of index two; in
the exceptional quotient \(S_4/V_4\simeq S_3\), it is
\(A_4/V_4\simeq A_3\).  Hence \(M\) has a unique quadratic subfield, forcing
\(Q_{\bm a}=Q_{\bm b}\), contrary to the preceding ramification argument.
Thus \(M=K\).
\end{proof}

\subsection{The pure highest-coordinate obstruction}
\label{r2:subsec:pure-highest}

Proposition~\ref{r2:prop:different-projections} settles syndrome planes on which
\(\pi_e\) has positive rank.  It does not apply when \(\pi_e(W)=0\), because
then all three nonzero syndromes have the same lower projection.  Such a plane
consists of pure highest-coordinate targets
\[
        \bfs_\gamma^\T=(0,\ldots,0,\gamma)^\T,
\]
and recall that \(k_e=2e-3\).  In this remaining
case the highest-moment defect, rather than the translated divisor, must
distinguish the two sign fields.  We therefore prove that a basis with unequal
\(k_e\)-th powers gives independent completion fields.

\begin{lemma}[The pure-highest defect]
\label{r2:lem:pure-highest-defect}
Let \(e\geq4\).  For a pure highest-coordinate target \(\bfs_\gamma^\T\), the defect at the
generic Hankel valuation is
\[
        \eta_{\bfs_\gamma}=\gamma.
\]
\end{lemma}

The proof of Lemma~\ref{r2:lem:pure-highest-defect} is given in
Subsection~\ref{r2:subsec:pure-highest-defect}.  It combines the common roots
with the finite residual roots and shows that their reduced support is empty,
so their contribution to the highest moment is zero.
The chosen Hankel valuation detects the pure parameter through the exponent
\(2e-3\); only the resulting distinction of Artin--Schreier classes, not the
particular representative \(\eta_{\bfs_\gamma}=\gamma\), is intrinsic.

\begin{proposition}[Pure highest-coordinate planes]
\label{r2:prop:pure-highest-plane}
Let \(e\geq4\), and let
\(\bfs_{\gamma_1}^\T,\bfs_{\gamma_2}^\T\) be a basis of a pure
highest-coordinate syndrome plane.  If
\[
 \gamma_1^{k_e}\neq\gamma_2^{k_e},
\]
then
\[
 L_{\bfs_{\gamma_1}}\cap L_{\bfs_{\gamma_2}}=K.
\]
\end{proposition}

\begin{proof}
Because both lower projections vanish, the two translated lower Hankel
divisors coincide.  Choose the same dominating component \(\Gamma\) above
that divisor for both sign calculations.  Along this common component,
Lemma~\ref{r2:lem:generic-sign-ramification}
and Lemma~\ref{r2:lem:pure-highest-defect} give reduced leading terms whose
coefficient difference is
\[
 (\gamma_1^{k_e}+\gamma_2^{k_e})D_R^{2e-1}.
\]
It is nonzero, and the pole order \(2e-1\) is odd.  The difference of the
two Artin--Schreier classes is therefore not a coboundary, so the sign fields
are distinct.  The quotient argument in the proof of
Proposition~\ref{r2:prop:different-projections} now gives the asserted trivial
intersection.
\end{proof}

\subsection{From a connected fiber product to a common support}
\label{r2:subsec:return-to-support}

The field-comparison step is now complete, so we return to the support
problem.  Suppose that the completion fields of two linearly independent
non-one-column syndromes are linearly disjoint.  Their one-target covers then
combine into a geometrically integral fiber product.  A rational point on its
closure is exactly a common core together with two completing tuples, and
Lemma~\ref{lem:common-core-certificate} converts that point into
\[
 \mu_e\bigl(\operatorname{Span}_{\F_2}
 \{\bm a^\T,\bm b^\T\}\bigr)\leq3e-1.
\]

Let \(e\geq2\), and let \(\bm a^\T,\bm b^\T\) be linearly independent
non-one-column syndromes.  Define the simultaneous
common-core incidence variety by
\[
 \mathcal Z_{\bm a,\bm b}\eqdef
 \left\{(\mathbf z,\mathbf x,\mathbf y)\in\mathbb A_F^{3e-1}:
 \begin{aligned}
 \sum_{i=1}^{e-1}h_e(z_i)+\sum_{j=1}^{e}h_e(x_j)&=\bm a^\T,\\
 \sum_{i=1}^{e-1}h_e(z_i)+\sum_{j=1}^{e}h_e(y_j)&=\bm b^\T
 \end{aligned}
 \right\}.
\]
It has \(3e-1\) locator variables and \(2e\) equations.  Recall that
\(U_{\bm a,\bm b}=U_{\bm a}\cap U_{\bm b}\).
After base change to \(\overline F\), the restriction of
\(\mathcal Z_{\bm a,\bm b}\) over this common open set is precisely
\[
 \left(\mathcal X_{\bm a}\times_{U_{\bm a}}U_{\bm a,\bm b}\right)
 \times_{U_{\bm a,\bm b}}
 \left(\mathcal X_{\bm b}\times_{U_{\bm b}}U_{\bm a,\bm b}\right).
\]
If the two completion fields are linearly disjoint, this fiber product is
geometrically integral and finite \'{e}tale over \(U_{\bm a,\bm b}\), hence
has dimension \(e-1\).  Let \(\mathcal C_{\bm a,\bm b}\) be its closure in
\(\mathcal Z_{\bm a,\bm b}\times_F\overline F\).  This closure is an irreducible
component: any component containing it meets the displayed open set, where
the projection is finite over an \((e-1)\)-dimensional base, and therefore
has dimension at most \(e-1\).  Since
\(\dim\mathcal C_{\bm a,\bm b}=e-1\), strict containment is impossible.
Moreover, every component dominating the common-core base meets this open
set and induces a component of the integral fiber product.  Thus
\(\mathcal C_{\bm a,\bm b}\) is the unique component dominating the base.  It
is consequently Frobenius-stable, defined over \(F\), and absolutely
irreducible of dimension \(e-1\).

Recall that \(D_e=\prod_{j=1}^{e}(2j-1)\).  At order two the ordinary
B\'ezout term is the active term in the definition of \(B_{e,2}\), so
\[
 B_{e,2}=D_e^2.
\]
Indeed, put \(x=D_e/e!\geq1\).  For \(e=2\) the comparison is direct, while
for \(e\geq3\) one has
\((1+2(e-1)x)^{e-1}\geq x^2\); hence the second term in the minimum is at
least \(D_e^2\).  Thus \(B_{e,2}\) is exactly the product of the degrees of
the two sets of odd-moment equations.

\begin{proposition}[The common-core point criterion]
\label{r2:prop:common-core-point}
Let \(e\geq2\), and let \(\bm a^\T,\bm b^\T\) be linearly independent
non-one-column syndromes.  Suppose that
\[
        L_{\bm a}\cap L_{\bm b}=K.
\]
If
\[
        q>\max\{2eB_{e,2}^2,\,2B_{e,2}^4\},
\]
then
\[
        \mu_e\!\left(\Span_{\F_2}\{\bm a^\T,\bm b^\T\}\right)\leq3e-1.
\]
\end{proposition}

\begin{proof}
The component \(\mathcal C_{\bm a,\bm b}\) lies in the simultaneous variety,
which is cut out by two equations of each degree
\[
        1,3,\ldots,2e-1.
\]
Projective B\'ezout therefore bounds its degree \(\delta\) by
\[
        \delta\leq
        \left(\prod_{j=1}^{e}(2j-1)\right)^2
        =B_{e,2}.
\]
Cafure--Matera's effective rational-point criterion
\cite[Corollary~7.4]{CafureMatera2006} gives an \(F\)-rational point on an
absolutely irreducible affine variety of positive dimension \(e-1\) and
degree at most \(B_{e,2}\) under the displayed hypothesis.  Such a point gives
\(e-1\) common locators and two groups of \(e\) completing locators that
satisfy the equations in Lemma~\ref{lem:common-core-certificate} for
\(r=2\).  That lemma gives the asserted bound.
\end{proof}

The ramification argument above was proved for \(e\geq4\).  To apply the
common-core plane criterion for every \(e\geq2\), it remains to verify
directly that the same field-disjointness conclusions hold for \(e=2,3\).
\begin{lemma}[The boundary completion covers]
\label{r2:lem:boundary-completion-covers}
For \(e=2\) and \(e=3\), the conclusions of
Propositions~\ref{r2:prop:different-projections} and
~\ref{r2:prop:pure-highest-plane} remain valid.  In the pure-highest case for
\(e=3\), the latter conclusion requires a basis
\(\bfs_{\gamma_1}^{\T},\bfs_{\gamma_2}^{\T}\) with
\(\gamma_1^3\neq\gamma_2^3\).
\end{lemma}

\begin{proof}
Adopt the conventions \(\Delta_0=\Delta_1=1\) and interpret an empty
Vandermonde product as one.  For \(e=2\), direct calculation gives
\[
 \Delta_2=p_1.
\]
Its divisor is reduced and absolutely irreducible, and its additive
translation stabilizer is trivial.  The common-core moment map consists of
one locator and has no collision divisor.  For a non-one-column target
\(\bfs^{\T}=(\sigma_1,\sigma_3)^{\T}\), the highest-moment defect on the
translated divisor is
\[
 \eta_{\bfs}=\sigma_3+\sigma_1^3\neq0;
\]
equality would say precisely that \(\bfs^{\T}=h_2(\sigma_1)\).
With no finite residual roots, the compactified two-root chart is
\[
\Psi\colon\mathbb A_F^2\longrightarrow\mathbb A_F^2,
\qquad
\Psi(u,w)=(wu,w+w^3u^3),
\]
whose Jacobian determinant at \((0,\eta_{\bfs})\) is \(\eta_{\bfs}\).
Consequently, \(\varpi=\Delta_2\) is a uniformizer and the sign
Artin--Schreier class has reduced leading polar part
\[
 \eta_{\bfs}\varpi^{-3}.
\]
The splitting group is \(S_2\), and its sign field is the entire quadratic
splitting field.  Thus different lower projections give distinct branch
divisors and hence trivial intersection.  In a pure-highest plane,
\(\eta_{\bfs_\gamma}=\gamma\); the two elements of any binary basis are
distinct, so their sign fields are distinct as well.

For \(e=3\), one has
\[
 \Delta_3=p_1^3+p_3.
\]
For a translation by \((a_1,a_3)\),
\[
 \Delta_3(p_1+a_1,p_3+a_3)-\Delta_3(p_1,p_3)
 =a_1p_1^2+a_1^2p_1+a_1^3+a_3,
\]
so the translation stabilizer is again trivial.  At the generic Hankel
divisor there is one finite residual root, hence \(D_R=1\).  Here the
pulled-back uniformizer is
\[
 \varpi\eqdef
 \Delta_3\bigl(\pi_3(\bfs^\T)+\mathbf m_3(\mathbf z)\bigr).
\]
The nonvanishing argument for \(\eta_{\bfs}\) in
Appendix~\ref{r2:app:local-sign} applies unchanged.  The chart for the two escaping roots
and the identification of its quadratic extension with the global sign field
specialize as well: here \(\Delta_1=1\), and the boundary Jacobian is
\(\eta_{\bfs}\neq0\).  The local sign class therefore has reduced leading
polar part
\[
 \eta_{\bfs}^{\,3}\varpi^{-5}.
\]
Since the only normal subgroups of \(S_3\) are
\(1,A_3,S_3\), distinct sign quadratic fields force the two splitting
fields to be linearly disjoint.  Different lower projections distinguish
the translated Hankel divisors.  For pure-highest targets,
\(\eta_{\bfs_\gamma}=\gamma\), so
\(\gamma_1^3\neq\gamma_2^3\) distinguishes the sign fields.  This proves
all the asserted boundary conclusions.
\end{proof}

The preceding results, with the proofs deferred to
Appendices~\ref{r2:app:universal-hankel} and~\ref{r2:app:local-sign}, provide all
auxiliary inputs needed for Proposition~\ref{r2:prop:common-core-plane}.  It remains to
choose a basis of the original syndrome plane, place it in one of the two
field-disjointness cases, and apply the common-core point criterion.

\begin{proof}[Proof of Proposition~\ref{r2:prop:common-core-plane}]
Let \(W\subseteq F^e\) be a two-dimensional binary syndrome space.  If two
distinct one-column syndromes span \(W\), the conclusion is immediate.

Suppose first that \(\pi_e|_W\) has positive rank.  If this immediate case
does not occur,
Lemma~\ref{r2:lem:projection-basis} gives a basis \(\bm a^\T,\bm b^\T\) of
non-one-column syndromes with different lower projections.  For \(e\geq4\),
Proposition~\ref{r2:prop:different-projections} gives
\(L_{\bm a}\cap L_{\bm b}=K\); for \(e=2,3\), the same conclusion follows
from Lemma~\ref{r2:lem:boundary-completion-covers}.  The field-size hypothesis in
Proposition~\ref{r2:prop:common-core-plane} is exactly the hypothesis of
Proposition~\ref{r2:prop:common-core-point}, which now gives
\[
        \mu_e\!\left(\Span_{\F_2}\{\bm a^\T,\bm b^\T\}\right)\leq3e-1.
\]

Suppose instead that \(\pi_e(W)=0\) and choose the basis
\(\bfs_{\gamma_1}^\T,\bfs_{\gamma_2}^\T\) from the second hypothesis.
For \(e\geq4\), Proposition~\ref{r2:prop:pure-highest-plane} gives the same
linear-disjointness conclusion, while for \(e=2,3\) it again follows from
Lemma~\ref{r2:lem:boundary-completion-covers}.  Proposition~\ref{r2:prop:common-core-point}
therefore gives the same support bound.  This proves both cases of the
criterion.
\end{proof}

This completes the proof of the second-order common-core criterion.
The higher-order argument may now invoke the completion cover, the sign
valuations, and the boundary cases by their labels in this section.

\section{Simultaneous completion over a common core}
\label{sec:higher-framework}

We now prove Proposition~\ref{prop:higher-order-common-core}.  Recall that
\(h_e(x)=(x,x^3,\ldots,x^{2e-1})^\T\), that \(\pi_e\) deletes the highest
syndrome coordinate, and that a target is admissible when it is nonzero and
not one-column.  We write \(\mathbb A_F^n\) for affine \(n\)-space over
\(F\).

Fix linearly independent admissible targets
\(\bfs_1^\T,\ldots,\bfs_r^\T\).  Write
\(\bfs_i^\T=(s_{i,1},\ldots,s_{i,e})^\T\), where \(s_{i,j}\) is the
odd-moment coordinate of degree \(2j-1\) of target \(i\).  By
Lemma~\ref{lem:common-core-certificate}, it is enough to find an
\(F\)-rational point on the simultaneous incidence variety
\[
 \mathcal Z_{\bfs_1,\ldots,\bfs_r}
 \subseteq\mathbb A_F^{(r+1)e-1}
\]
with coordinates
\[
 z_1,\ldots,z_{e-1},
 \qquad x_{i,1},\ldots,x_{i,e}\quad(1\leq i\leq r),
\]
defined by the \(re\) equations
\[
 \sum_{u=1}^{e-1}z_u^{2j-1}
 +\sum_{\ell=1}^{e}x_{i,\ell}^{2j-1}
 =s_{i,j}
 \qquad(1\leq i\leq r,\ 1\leq j\leq e).
\]
Such a point is exactly one common tuple and \(r\) completing tuples of the
kind required in Lemma~\ref{lem:common-core-certificate}.

We obtain this point in four stages.  First, the one-target completion cover
developed in the preceding self-contained section assigns a quadratic sign
class to each basis target.  Second, the hypotheses of
Proposition~\ref{prop:higher-order-common-core} make these classes linearly
independent.  Third, this independence forces the simultaneous cover to have
the full product function field and hence a unique geometrically integral
component dominating the common-core coordinates.  Finally, the refined
degree bound and an effective finite-field point estimate produce an
\(F\)-rational point on that component.  At the end of the section we return
through Lemma~\ref{lem:common-core-certificate} to the BCH support bound.

\subsection{The one-target cover and its sign class}

Fix an algebraic closure \(\overline F\) of \(F\).  For an admissible target
\(\bfs^\T=(s_1,\ldots,s_e)^\T\), meaning a nonzero non-one-column
syndrome, recall the incidence variety \(\mathcal V_{\bfs}\), its
open subset \(\mathcal X_{\bfs}\), and the completion cover
\[
 \rho_{\bfs}:\mathcal X_{\bfs}\longrightarrow U_{\bfs}
\]
from Subsection~\ref{r2:subsec:completion-covers}.  Proposition
~\ref{r2:prop:completion-cover} shows that, after base change to
\(\overline F\), this is a connected finite \(S_e\)-cover of degree \(e!\).

Put
\[
 K\eqdef\overline F(z_1,\ldots,z_{e-1}).
\]
Denote the function-field extension of this generic \(S_e\)-cover by
\(L_{\bfs}/K\).  Its quadratic sign field is
\[
 Q_{\bfs}\eqdef L_{\bfs}^{A_e}.
\]
Here \(S_e\) and \(A_e\) are the symmetric and alternating groups,
respectively.  For \(e=2\), this convention gives \(A_2=1\) and
\(Q_{\bfs}=L_{\bfs}\).  As in
Subsection~\ref{r2:subsec:sign-disjointness}, write
\([a_{\bfs}]\in K/\wp(K)\) for its Artin--Schreier class, where
\(\wp(K)=\{u^2+u:u\in K\}\).  Recall that an odd reduced pole
certifies a nonzero class.

We distinguish mixed targets using the divisors determined by their distinct
lower projections.  For
general \(e\), pure targets are distinguished by the leading coefficients of
their \(k_e\)-th-power defects; for \(2\leq e\leq6\), the complete Berlekamp classes
give a stronger conclusion.  Here ``complete'' means the full
Artin--Schreier class, including its lower odd polar coefficients.
Independence of the resulting \(r\) quadratic
classes forces the joint function field to have full product degree
\((e!)^r\).  The corresponding dominating component of
\(\mathcal Z_{\bfs_1,\ldots,\bfs_r}\) is then geometrically integral, and an
effective point bound supplies the required \(F\)-rational point.

The following is the precise one-target input supplied by the preceding
second-order engine.  Its general calculation covers \(e\geq4\), and its
explicit boundary calculation covers \(e=2,3\).

\begin{proposition}[Sign-valuation input]
\label{prop:sign-input}
Let \(e\geq2\).
\begin{enumerate}
\item For every admissible target \(\bfs^\T\), there is a divisorial
valuation \(v_{\bfs}\) of \(K\) at which \([a_{\bfs}]\) has an odd reduced
pole.  If \(\bft^\T\) is another admissible target and
\(\pi_e(\bfs^\T)\neq\pi_e(\bft^\T)\), then \([a_{\bft}]\) is unramified
at \(v_{\bfs}\).
\item For pure targets
\(\bfs_\gamma^\T\), the relevant valuation is common
to all \(\gamma\), and the reduced leading polar coefficient of
\([a_{\bfs_\gamma}]\) is a fixed nonzero factor times
\(\gamma^{k_e}\).  Its pole order is \(2e-1\), which is odd.
\end{enumerate}
\end{proposition}

\begin{proof}
For \(e\geq4\), Lemma~\ref{r2:lem:generic-sign-ramification} and
Proposition~\ref{r2:prop:different-projections} give a valuation with an odd
reduced pole for the chosen target and establish unramifiedness for a target
with a different lower projection;
Lemma~\ref{r2:lem:pure-highest-defect} computes the pure-target defect.  The
boundary calculation in Lemma~\ref{r2:lem:boundary-completion-covers} treats
\(e=2,3\).  Writing \(\tau\) for a uniformizer at the corresponding
pure-target valuation, the reduced leading term is
\(\gamma^3\tau^{-5}\) for \(e=3\) and \(\gamma\tau^{-3}\) for \(e=2\).
For every \(e\geq4\), the generic calculation and the pure-target defect give
a common nonzero factor times
\(\gamma^{2e-3}\tau^{-(2e-1)}\).
\end{proof}

The leading coefficient in Proposition~\ref{prop:sign-input}(2) can vanish in
a sum even when the pure targets are independent.  For at most six errors,
the full polar parts still give a nonzero sign class for every nonempty sum,
provided that the pure targets are linearly independent.

\begin{lemma}[Pure sign independence for at most six errors]
\label{lem:pure-sign-independence-small-e}
Let \(2\leq e\leq6\), and let
\(\gamma_1,\ldots,\gamma_u\in F^*\) be linearly independent over \(\F_2\).
Then the sign classes
\[
 [a_{\bfs_{\gamma_1}}],\ldots,[a_{\bfs_{\gamma_u}}]
 \quad\text{in}\quad K/\wp(K)
\]
are linearly independent.
\end{lemma}

The proof is given in Appendix~\ref{app:pure-sign-independence}.  It is
kept separate because the common-core argument below uses only the resulting
independence, not the coefficient expansions.

\subsection{From independent signs to an integral component}

We use the standard sign-field compositum criterion for symmetric
extensions.  Bary-Soroker proves this criterion in
\cite[Lemmas~3.2--3.4]{BarySoroker2012Irreducible}; Carmon
\cite[Section~6]{Carmon2015ChowlaCharTwo} uses the same characteristic-two
form.  We record it here only to fix the notation needed below.

\begin{lemma}[Independent signs force the full compositum]
\label{lem:higher-field-compositum}
Let \(e,r\geq2\), and let \(L_1,\ldots,L_r\) be \(S_e\)-Galois
extensions of a field \(K\) of characteristic two.  Let
\(Q_i=L_i^{A_e}\), with Artin--Schreier classes
\([a_i]\in K/\wp(K)\).  If \([a_1],\ldots,[a_r]\) are linearly independent
over \(\F_2\), then
\[
 \Gal(L_1\cdots L_r/K)\simeq S_e^r,
 \qquad
 [L_1\cdots L_r:K]=(e!)^r.
\]
\end{lemma}

\begin{proof}
This is exactly the cited sign-field compositum criterion.  For \(e=2\) it
reduces to the usual Artin--Schreier compositum criterion.
\end{proof}

\begin{lemma}[Mixed sign independence]
\label{lem:mixed-power-signs}
Let \(e,r\geq2\).  Suppose that an admissible basis consists of pure targets
\(\bfs_{\gamma_1}^\T,\ldots,\bfs_{\gamma_u}^\T\), where the parameters are
linearly independent and either \(2\leq e\leq6\) or their \(k_e\)-th powers
are linearly independent.  Suppose also that the remaining targets
\(\bft_1^\T,\ldots,\bft_v^\T\) have nonzero pairwise distinct lower
projections.  Then all \(r=u+v\) sign classes are linearly independent in
\(K/\wp(K)\).
\end{lemma}

\begin{proof}
Consider a nonzero binary combination of the sign classes.  If it contains
the class of some \(\bft_j^\T\), then at the private valuation
\(v_{\bft_j}\) that class has an odd reduced pole and every other occurring
class is unramified, because all other lower projections are different from
\(\pi_e(\bft_j^\T)\).  If the combination contains only pure classes, its
nonvanishing follows from Lemma~\ref{lem:pure-sign-independence-small-e} when
\(2\leq e\leq6\).  Otherwise, its leading coefficient at the common pure
valuation is a fixed nonzero factor times the corresponding combination of
\(\gamma_1^{k_e},\ldots,\gamma_u^{k_e}\).  This is nonzero, and the pole
order \(2e-1\) is odd.  In either case the combination is not in \(\wp(K)\).
\end{proof}

\begin{lemma}[Refined degree budget]
\label{lem:refined-degree-budget}
Let \(\bfs_1^\T,\ldots,\bfs_r^\T\) be admissible targets, and suppose that
the fiber product of their generic completion covers is integral over the
common \'{e}tale open set.  The closure \(\mathcal C\) of this fiber product
in \(\mathcal Z_{\bfs_1,\ldots,\bfs_r}\) has affine degree at most
\(B_{e,r}\).
\end{lemma}

\begin{proof}
In the mixed-degree count, blocks receiving no completion hyperplanes
contribute the exact fiber degree \(e!\). Since a zero-dimensional slice
uses \(e-1\) hyperplanes in total, at most \(e-1\) blocks need the larger
mixed-degree bound.
Put \(a=e-1\), \(E=e!\), and \(D=D_e\).  Eliminate, in each completion
block, the equation of first moments.  The remaining one-target equations
have bidegrees at most
\[
 (3,3),(5,5),\ldots,(2e-1,2e-1)
\]
 in the common-core and completion variables.  For the \(i\)-th target,
 \(1\leq i\leq r\), let \(c_{i,u}\) denote the mixed degree in
 \(\mathbb P_z^a\times\mathbb P_{x_i}^a\) obtained by intersecting its
 one-target dominating component with \(a-u\) general core hyperplanes and
 \(u\) general completion hyperplanes.  The
 generic completion cover has degree \(E\), while refined
 multiprojective B\'ezout gives
 \[
 c_{i,0}=E,
 \qquad
 c_{i,u}\leq D\binom au\quad(1\leq i\leq r,\ 1\leq u\leq a)
 \]
\cite[Theorem~1.11]{DAndreaKrickSombra2013}.

For integers \(0\leq u_i\leq a\), put \(U=\sum_i u_i\leq a\).
Work in
\(M=\mathbb P_z^a\times\prod_{i=1}^r\mathbb P_{x_i}^a\), with
hyperplane classes \(H_z,H_1,\ldots,H_r\).  Let \(W_i\subset M\) be the
pullback of the \(i\)-th one-target dominating component's multiprojective
closure.  The mixed-degree convention of
\cite[Section~1.2]{DAndreaKrickSombra2013} gives
\[
 [W_i]=\sum_{v=0}^a c_{i,v}H_z^vH_i^{a-v}
 \qquad\text{in }A^*(M).
\]
Let \(S\) be a general product slice with class
\(H_z^{a-U}\prod_i H_i^{u_i}\).  The coefficient of
\(H_z^a\prod_iH_i^a\) in \([S]\prod_i[W_i]\) is exactly
\(\prod_i c_{i,u_i}\): the exponent of each \(H_i\) selects
\(v=u_i\).  Moving each \(W_i\) independently by projective
automorphisms of its core and completion factors makes its intersection
with the other moved cycles and \(S\) proper, with this total
intersection multiplicity.

It remains to compare with the original slice.  The \(a\) general
hyperplanes defining \(S\) avoid the boundary of the common \'{e}tale
locus in the multiprojective closure of \(\mathcal C\), since that
boundary has dimension less than \(a\).  On this locus the \(W_i\)
intersect scheme-theoretically in the smooth fiber product
\(\mathcal C\).  Consequently, at every point of
\(S\cap\mathcal C\), their \(ra\) local block equations together with
the \(a\) slice equations form a zero-dimensional complete intersection
in the smooth \((r+1)a\)-dimensional ambient space.  In a one-parameter
algebraic family from the identity to the general moves, its completed
local algebra is finite flat over the parameter ring: the equations
together with the parameter form a regular sequence.  Thus each original
local length persists in the moved intersection.  All other contributions
are nonnegative, so the mixed degree of \(\mathcal C\) is at most
\(\prod_i c_{i,u_i}\), even if the original intersection has excess
components at infinity.  The same local comparison, perturbing the
one-target bihomogenized equations to intersect properly, justifies the
preceding B\'ezout bounds on \(c_{i,u}\).  Hence the Segre degree of
\(\mathcal C\) is at
most
\[
 \sum_{\substack{0\leq u_i\leq a\ (1\leq i\leq r)\\ U\leq a}}
 \frac{a!}{(a-U)!\prod_{i=1}^r u_i!}
 E^{r-k}D^k\prod_{\substack{1\leq i\leq r\\u_i>0}}\binom a{u_i},
 \qquad
 k=\bigl|\{i\in\{1,\ldots,r\}:u_i>0\}\bigr|.
\]
The standard affine degree is no larger.  Explicitly, project the Segre
coordinates to
\[
 [1,z_1,\ldots,z_a,x_{1,1},\ldots,x_{r,a}].
\]
On the common affine chart
this projection is the identity on the affine variables, hence an isomorphism
onto the affine component, and linear projection cannot increase degree.

Write \(\lambda=D/E\geq1\).  For fixed \(U\), use
\(\lambda^k\leq\lambda^U\),
\(\binom a{u_i}\leq a^{u_i}/u_i!\), and
\[
 \sum_{u_1+\cdots+u_r=U}\frac1{\prod_i(u_i!)^2}
 \leq\frac{r^U}{U!}.
\]
The terms with total \(U\) are therefore bounded by
\(E^r\binom aU(a\lambda r)^U\), so the displayed degree estimate gives
\[
 \deg\mathcal C\leq E^r(1+a\lambda r)^a.
\]
Ordinary projective B\'ezout independently gives degree at most \(D^r\).
Taking the smaller of the two bounds proves
\(\deg\mathcal C\leq B_{e,r}\).
\end{proof}

The bound of Xiong, Yip, and
Zullo~\cite[Proposition~4.1]{XiongYipZullo2026} requires
\(q>2D_e^{4r}\).  The refined term in \(B_{e,r}\) gives a smaller
field-size threshold for the present completion-cover construction when
that term is less than \(D_e^r\).  Section~\ref{sec:cre-refinement} removes
the connectedness requirement for the degree estimate and gives the
degree budget \(L_{e,r}\leq B_{e,r}\).

\subsection{The rational point and the return to BCH support}

\begin{proposition}[Higher-order common-core point criterion]
\label{prop:higher-order-point}
Let \(e,r\geq2\), and let
\(\bfs_1^\T,\ldots,\bfs_r^\T\) be linearly independent admissible
syndromes whose sign classes are independent.  If
\eqref{eq:higher-order-threshold} holds, then
\[
 \mu_e\bigl(\Span_{\F_2}
 \{\bfs_1^\T,\ldots,\bfs_r^\T\}\bigr)\leq(r+1)e-1.
\]
\end{proposition}

\begin{proof}
Put
\[
 U\eqdef\bigcap_{i=1}^{r}U_{\bfs_i}.
\]
This is a nonempty open subset of
\(\mathbb A_{\overline F}^{e-1}\).  Over \(U\), form the fiber product of
the \(r\) restricted covers
\[
 \rho_{\bfs_i}^{-1}(U)\longrightarrow U\qquad(1\leq i\leq r).
\]
Lemma~\ref{lem:higher-field-compositum} makes this fiber product
geometrically integral and finite over \(U\), which has dimension \(e-1\).

Its closure in \(\mathcal Z_{\bfs_1,\ldots,\bfs_r}\) is irreducible.
Every component dominating the base meets \(U\), where its restriction is a
component of the integral fiber product; hence the displayed closure is the
unique dominating component.  The defining equations and targets are
\(F\)-rational, so uniqueness makes the component Frobenius-stable and hence
defined over \(F\).  Lemma~\ref{lem:refined-degree-budget} bounds its degree
by \(B_{e,r}\).
Under~\eqref{eq:higher-order-threshold}, Cafure and Matera's effective
rational-point criterion~\cite[Corollary~7.4]{CafureMatera2006} gives an
\(F\)-rational point on this positive-dimensional absolutely irreducible
component.  Lemma~\ref{lem:common-core-certificate} converts that point into
a BCH-column support of size at most \((r+1)e-1\).
\end{proof}

\begin{proof}[Proof of Proposition~\ref{prop:higher-order-common-core}]
Lemma~\ref{lem:mixed-power-signs} makes the sign classes of the given basis
independent.  Proposition~\ref{prop:higher-order-point} then gives the
asserted support bound.
\end{proof}

\section{Concluding remarks}
\label{sec:conclusion}

The bounds combine exact second-order families with effective estimates at
every higher order.  Their common ingredient is a shared core of BCH locator
columns.  Theorem~\ref{thm:intro-componentwise} applies to arbitrary targets
and retains the correction-free count \((t+1)e-1\) of
Xiong--Yip--Zullo with a smaller sufficient field threshold at large order.
The improvement estimates the selected component itself: factorial bounds
at both one-target endpoints, ordinary affine B\'ezout for intermediate
mixed degrees, and an unweighted multicone sum in place of the Segre degree.
It does not require independence of the sign fields.

In the separate completion-cover construction, passing to several
simultaneous syndromes uses linear
independence in the Artin--Schreier class group rather than pairwise sign
comparison.  Matroid intersection measures the defect
of the power map on the pure kernel, and binary labels repair the entire
defect with \(c_{e,m,t}\) shared columns.  This gives a uniform stable bracket
for every \(e,t\geq2\), with a correction term bounded independently of
\(t\) for fixed \(e\).  When \(\gcd(2e-3,2^m-1)=1\), the correction term
vanishes and the upper bound is the bare common-core count \((t+1)e-1\).
For \(2\leq e\leq6\), the full Berlekamp classes make the correction vanish even
when the power map has nontrivial fibers.
For three errors, the separate ten-column construction of
Proposition~\ref{prop:pure-e3-exact-ten} is exact on every three-dimensional pure binary space
already for \(m\geq18\); consequently, a global value 11 can only come from a
nonpure space containing syndromes of ordinary coset
weight at least three, that is, syndromes requiring at least three BCH
columns.
The componentwise degree estimate improves the stable field-size
condition at large order: for fixed \(e\), its sufficient exponent is
\(4t\log_2(e!)+O_e(\log t)\).
The sphere-covering inequality supplies the complementary asymptotic lower
bound \(et\), while the Griesmer bound remains stronger in the third-order
families treated explicitly here.

The smooth infinity component of Xiong, Yip, and
Zullo~\cite{XiongYipZullo2026} shows that the power-basis defect is unnecessary
for the universal upper bound; the componentwise argument here improves its
effective range.  It remains to determine how close the common-core count is to the
maximum of the sphere-covering and Griesmer lower bounds.  At third order,
the exact radii remain open in all four displayed families; for example, the
triple-error radius is 10 or 11 and the four-error radius is 13, 14, or 15 in
their stable ranges.  Improving the effective field-size thresholds is
another remaining problem.

\appendix

\section{The universal Hankel divisor}
\label{r2:app:universal-hankel}

This appendix proves the properties of the universal Hankel divisor used in
Section~\ref{r2:sec:second-order-engine}.  The generic triangular argument is
needed only for \(e\geq4\); the two smaller cases are included explicitly.
Lemma~\ref{r2:lem:delta-irreducible} and
Proposition~\ref{r2:prop:no-translation-stabilizer-proof} provide, respectively,
the irreducibility and translation assertions quoted there.

For the duration of the appendix, let \(k\) be an algebraically closed field
of characteristic two.  Adopt the conventions
\[
 \Delta_0=\Delta_1=1.
\]
Starting from the independent odd moments
\[
 p_1,p_3,\ldots,p_{2e-3},
\]
define the remaining moments recursively by \(p_{2j}\eqdef p_j^2\), and put
\(p_0\eqdef e\bmod2\).  Let
\[
 \mathscr H_e^{\mathrm{univ}}\eqdef(p_{i+j})_{0\leq i,j\leq e-1}.
\]
We give the moment variable \(p_j\) weight \(j\).  With this convention,
\(\Delta_e\) is weighted homogeneous of weight \(e(e-1)/2\).

\begin{lemma}
\label{r2:lem:delta-irreducible}
For every \(e\geq2\), there is a unique polynomial
\[
 \Delta_e\in\F_2[p_1,p_3,\ldots,p_{2e-3}]
\]
such that
\[
 \det\mathscr H_e^{\mathrm{univ}}=\Delta_e^2.
\]
It is absolutely irreducible and reduced.  Under the moment specialization
\(p_j=\sum_{i=1}^e x_i^j\), one has
\begin{equation}
\label{r2:eq:delta-vandermonde}
 \Delta_e=\prod_{1\leq i<j\leq e}(x_i+x_j).
\end{equation}
Moreover, \(V(\Delta_e)\) is the closure of the lower odd-moment image of
\(e-2\) roots.
\end{lemma}

\begin{proof}
In characteristic two, the non-involutive terms in the determinant expansion
cancel in inverse pairs.  Every surviving term is a square, so the first
assertion follows.  Alternatively, after the root specialization,
\(\mathscr H_e^{\mathrm{univ}}\)
is the product of the Vandermonde matrix and its transpose, which
gives~\eqref{r2:eq:delta-vandermonde}.

The homomorphism from the lower odd-moment ring to
\(k[x_1,\ldots,x_e]\) is injective.  Indeed, the Jacobian of the first
\(e-1\) odd power sums with respect to \(x_1,\ldots,x_{e-1}\) is
\[
 \det(x_i^{2j-2})_{1\leq i,j\leq e-1}
 =\prod_{i<j}(x_i+x_j)^2,
\]
which is nonzero.  If \(\Delta_e\) factored over \(k\), the image of each
factor in the root polynomial ring would, by unique factorization, be a
product of a subset of the distinct linear factors \(x_i+x_j\) in
\eqref{r2:eq:delta-vandermonde}.  That product is symmetric in the \(e\) roots.
Since the symmetric group is transitive on unordered pairs, the subset is
either empty or the full set.  Thus the factorization is trivial, and
\(\Delta_e\) is absolutely irreducible and reduced.

The lower odd-moment image of \(e-2\) roots is irreducible of dimension
\(e-2\), by the same Vandermonde-square Jacobian.  Padding these roots by
two zeros and applying~\eqref{r2:eq:delta-vandermonde} puts the image in
\(V(\Delta_e)\).  Both closures are irreducible hypersurfaces of dimension
\(e-2\) in the \((e-1)\)-dimensional lower-moment space, so they coincide.
\end{proof}

Taking the square root of the surviving determinant terms gives the useful
involution expansion
\begin{equation}
\label{r2:eq:delta-involutions}
 \Delta_e=
 \sum_{\substack{\tau\in S_{\{0,\ldots,e-1\}}\\ \tau^2=1}}
 \prod_{\tau(i)=i}p_i
 \prod_{\substack{\{i,j\}:\tau(i)=j\\i<j}}p_{i+j}.
\end{equation}
In particular, pairing the last two indices shows that
\begin{equation}
\label{r2:eq:delta-recursion}
 \Delta_e=\Delta_{e-2}p_{2e-3}+\Phi_e,
 \qquad
 \frac{\partial\Delta_e}{\partial p_{2e-3}}=\Delta_{e-2},
\end{equation}
where \(\Phi_e\) is independent of \(p_{2e-3}\).

The remaining ingredient is a coefficient that prevents a translation in
the next-to-last moment.

\begin{lemma}
\label{r2:lem:triangular-monomial}
Put
\[
 N_e\eqdef\frac{e(e-1)}2,
 \qquad
 \kappa_e\eqdef N_e-(2e-5).
\]
For every \(e\geq4\), the coefficient of
\(p_1^{\kappa_e}p_{2e-5}\) in \(\Delta_e\) is one.
\end{lemma}

\begin{proof}
Set \(p_1=1\), retain \(p_{2e-5}\), and set every other independent odd
moment to zero in~\eqref{r2:eq:delta-involutions}.  A positive fixed vertex
survives precisely when its label is a power of two.  An ordinary edge
survives precisely when the sum of its endpoints is a power of two.  To
obtain the coefficient of \(p_{2e-5}\), exactly one edge must instead have
endpoint sum \(2e-5\).  There are only two possible such edges:
\[
 \{e-4,e-1\},
 \qquad
 \{e-3,e-2\}.
\]

\medskip
\noindent\emph{Forced-pair reduction.}
We record the parity of the remaining involutions.  For a finite set \(S\)
of positive integers, repeatedly take its largest element \(M\) that is not
a power of two.  Its only possible smaller partner with power-of-two sum is
\[
 2^{\lceil\log_2 M\rceil}-M.
\]
If that partner is absent, the reduction fails; otherwise remove the pair
and continue.  At the end, suppose that \(k(S)\) powers of two remain.  When
vertex zero is also present, it can be paired with any one of these \(k(S)\)
vertices, while all the others are fixed; in addition, zero itself may be
fixed exactly when \(e\) is odd.  The parity of completions is therefore
\begin{equation}
\label{r2:eq:completion-parity}
 k(S)+(e\bmod2).
\end{equation}

For \(n\geq3\), define
\[
 U_n\eqdef\{1,\ldots,n\}\setminus\{n-2\},
 \qquad
 V_n\eqdef\{1,\ldots,n-3\}\cup\{n\}.
\]
Strong induction gives the following elementary reduction facts: both
\(U_n\) and \(V_n\) reduce successfully exactly when \(n\) is even; in that
  case \(k(U_n)\) is odd and \(k(V_n)\) is even.  For \(n=3\), both sets fail
  immediately: their largest vertex is \(3\), whose forced partner for sum
  \(4\) is the absent vertex \(1\).  Let \(P\leq n<2P\), where
  \(P\) is a power of two.  The boundary sets for \(n=P,P+1,P+2\) reduce
directly and give, respectively, success, failure, and success with the
  asserted parities.  Indeed, for \(P\geq8\), the forced reduction leaves
  \(\{2,P/2,P\}\) and \(\{1,2,P/2,P\}\) from \(U_P\) and \(V_P\),
  respectively, and leaves \(\{1,2,P/2\}\) and \(\{2,P/2\}\) from
  \(U_{P+2}\) and \(V_{P+2}\); the case \(P=4\) gives the same parities
  directly.  For \(n=P+1\), the largest vertex \(P+1\) requires the absent
  partner \(P-1\) in both sets, so both reductions fail.  If \(n\geq P+3\),
peeling the forced pairs whose sums
  are \(2P\) leaves
\[
   U_n\longrightarrow V_{2P+2-n}\cup\{P\},
 \qquad
   V_n\longrightarrow U_{2P+2-n}\cup\{P\}.
\]
  The smaller index \(2P+2-n\) has the same parity as \(n\).  The induction
  hypothesis therefore proves success exactly for even \(n\), while the extra
  singleton \(P\) interchanges the two asserted parities.

\medskip
\noindent\emph{Application to the coefficient.}
For \(e\geq5\), removing the first special edge leaves the positive set
\(U_{e-2}\), while removing the second leaves \(V_{e-1}\); vertex zero
remains in both cases.  If \(e\) is even, only \(U_{e-2}\) reduces, and its
odd value of \(k\) makes~\eqref{r2:eq:completion-parity} equal to one.  If
\(e\) is odd, only \(V_{e-1}\) reduces, and its even value of \(k\), together
with the fixed-zero contribution, again gives one.  The boundary case is
best read directly from the involution expansion:
\[
 \Delta_4=p_1p_5+p_3^2+p_1^3p_3+p_1^6.
\]
Thus the coefficient of \(p_1^3p_3\) is one.  This proves the assertion for
every \(e\geq4\).
\end{proof}

\begin{proposition}
\label{r2:prop:no-translation-stabilizer-proof}
For every \(e\geq2\), the additive translation stabilizer of the universal
Hankel divisor is trivial:
\[
 \{\mathbf a:V(\Delta_e)+\mathbf a=V(\Delta_e)\}=\{\bm 0\}.
\]
\end{proposition}

\begin{proof}
The base cases are direct:
\[
 \Delta_2=p_1,
 \qquad
 \Delta_3=p_1^3+p_3.
\]
For \(e\geq4\), Lemma~\ref{r2:lem:delta-irreducible} shows that a translation
preserving the divisor satisfies
\(\Delta_e(\mathbf p+\mathbf a)=c\Delta_e(\mathbf p)\) for some nonzero
constant \(c\).  The top weighted part is unchanged by translation, so
\(c=1\).  Comparing the coefficient of \(p_{2e-3}\) in
\eqref{r2:eq:delta-recursion} gives
\[
 \Delta_{e-2}(\mathbf p_{<2e-5}+\mathbf a_{<2e-5})
 =\Delta_{e-2}(\mathbf p_{<2e-5}).
\]
Thus the truncated translation stabilizes \(V(\Delta_{e-2})\); induction
from the displayed base cases forces all translation coordinates through
\(a_{2e-7}\) to vanish.

Only \(a_{2e-5}\) and \(a_{2e-3}\) remain.  Since \(\Delta_e\) is weighted
homogeneous of weight \(N_e\), Lemma~\ref{r2:lem:triangular-monomial} makes the
coefficient of \(p_1^{\kappa_e}\) in
\(\Delta_e(\mathbf p+\mathbf a)-\Delta_e(\mathbf p)\) equal to
\(a_{2e-5}\); no other translated monomial has the required residual
weight.  Thus \(a_{2e-5}=0\).  Finally,~\eqref{r2:eq:delta-recursion} reduces
the difference to \(a_{2e-3}\Delta_{e-2}\), so \(a_{2e-3}=0\) as well.
\end{proof}

Lemma~\ref{r2:lem:delta-irreducible} proves parts~(i) and~(ii) of
Lemma~\ref{r2:lem:universal-hankel},
\eqref{r2:eq:delta-recursion} proves part~(iii), and
Proposition~\ref{r2:prop:no-translation-stabilizer-proof} proves part~(iv).
This completes the proof of Lemma~\ref{r2:lem:universal-hankel}.

\section{The local sign calculation}
\label{r2:app:local-sign}

We prove Lemma~\ref{r2:lem:generic-sign-ramification} and the identity used for
pure highest-coordinate planes.  All local rings in this appendix are taken
after extending the constant field to an algebraic closure of characteristic
two.  The calculation below is written for \(e\geq4\); the direct boundary
calculations for \(e=2,3\) appear in
Lemma~\ref{r2:lem:boundary-completion-covers}.

\subsection{Nonvanishing of the highest-moment defect}

At the generic point of the lower Hankel divisor, the lower residual moments
are represented by \(e-2\) distinct finite roots
\(r_1,\ldots,r_{e-2}\).  Recall the residual Vandermonde product
\[
        D_R\eqdef\prod_{i<j}(r_i+r_j).
\]
We first justify the assertion that
\[
 \eta_{\bfs}
 =p_{2e-1}-\sum_{i=1}^{e-2}r_i^{2e-1}
\]
is not identically zero for a nonzero, non-one-column target.

Suppose otherwise.  Consider a component of the pulled-back lower Hankel
divisor which dominates the generic point of \(V(\Delta_e)\) and is
generically disjoint from the collision divisor of the common-core moment
map; this is the component used in the local calculation below.  The
common-core lower-moment map is generically finite there, since its Jacobian
is a nonzero Vandermonde square.  The map
\[
 (r_1,\ldots,r_{e-2})\longmapsto
 \left(\sum_i r_i,\sum_i r_i^3,\ldots,
       \sum_i r_i^{2e-3}\right)
\]
is likewise generically finite onto \(V(\Delta_e)\): the Jacobian minor
formed by its first \(e-2\) coordinates is a nonzero Vandermonde square.
Hence the corresponding lower-moment incidence has an
\((e-2)\)-dimensional component in the ordered variables
\[
 (z_1,\ldots,z_{e-1},r_1,\ldots,r_{e-2}).
\]
If \(\eta_{\bfs}\) vanished identically on this component, the final
odd-moment equation would hold there as well.  We would therefore obtain an
\((e-2)\)-dimensional family of representations of \(\bfs^{\T}\) by
\(2e-3\) ordered locator slots.  We show that no such family exists.  In any
multiplicity stratum, let \(r\) distinct
locator values occur with odd multiplicity and let \(s\) further values
occur with positive even multiplicity.  Then
\[
 r+2s\leq2e-3.
\]
The even-multiplicity values cancel from all power sums and contribute only
\(s\) free parameters.  The Jacobian of the \(e\) odd power sums in the
\(r\) odd-multiplicity values is a Vandermonde matrix in their squares and
has rank \(\min\{r,e\}\).  The stratum dimension is therefore at most
\[
 s+\max\{r-e,0\}\leq e-3,
\]
except when \(r=0\) or \(r=1\).  The first exception gives the zero target,
and the second gives a one-column target.  Both are excluded.  This proves
that \(\eta_{\bfs}\neq0\) at the generic point.

\subsection{A chart for the two escaping roots}

For two roots \(x,y\), put
\[
 s\eqdef x+y,
 \qquad
 p\eqdef xy,
 \qquad
 U_j(s,p)\eqdef x^{2j-1}+y^{2j-1}.
\]
The quadratic power-sum recurrence and weighted homogeneity give coefficients
\(c_{j,\ell}\in\F_2\) such that
\begin{equation}
\label{r2:eq:pair-power-expansion}
 U_j(s,p)=sp^{j-1}
 +\sum_{\ell=1}^{j-1}c_{j,\ell}s^{2\ell+1}p^{j-1-\ell}.
\end{equation}
The coefficient of
\(sp^{j-1}\) is one; equivalently, after writing \(x=y+s\), the expression
modulo \(s^2\) is \(sy^{2j-2}=sp^{j-1}\).

Compactify the pair parameters by
\[
 p=u^{-1},
 \qquad
 s=wu^{e-1}.
\]
Then~\eqref{r2:eq:pair-power-expansion} becomes the regular polynomial
\begin{equation}
\label{r2:eq:regularized-pair-moments}
 \widetilde U_j(u,w)
 =wu^{e-j}
 +\sum_{\ell=1}^{j-1}c_{j,\ell}w^{2\ell+1}
 u^{e-j+\ell(2e-1)},
 \qquad 1\leq j\leq e.
\end{equation}
Consequently, the moment map extends through \(u=0\) as the morphism
\[
 \Psi\colon\mathbb A^{e}_{\overline F}\longrightarrow
 \mathbb A^{e}_{\overline F}
\]
whose ordered source coordinates are
\((r_1,\ldots,r_{e-2},u,w)\), whose target coordinates are the odd moments
of degrees \(1,3,\ldots,2e-1\), and whose components are
\[
 \Psi(r_1,\ldots,r_{e-2},u,w)_j
 \eqdef
 \sum_{i=1}^{e-2}r_i^{2j-1}+\widetilde U_j(u,w).
\]
At the boundary point \(u=0,w=\eta_{\bfs}\), its Jacobian determinant is
\begin{equation}
\label{r2:eq:local-chart-jacobian}
 \det(d\Psi)
 =\eta_{\bfs}\prod_{i<j}(r_i+r_j)^2
 =\eta_{\bfs}D_R^2\neq0.
\end{equation}
Indeed, the root columns form a Vandermonde-square block, the \(u\)-column
has its sole nonzero entry \(\eta_{\bfs}\) in row \(e-1\), and the
\(w\)-column has its sole nonzero entry one in row \(e\).  Thus \(\Psi\) is
\'{e}tale there and supplies an actual branch of the normalization on which
the displayed finite roots remain distinct and two roots escape.

We also need transversality after pulling this chart back to the common-core
space.  Recall that \(D_{\bfs}^{\circ}=D_{\bfs}\cap U_{\mathrm{core}}\), and
that \(\Gamma_{\bfs}\) is the closure of a component of the finite \'{e}tale
base change over \(D_{\bfs}^{\circ}\).  We may therefore work at its generic
 point, where the common-core moment map is \'{e}tale and
 \[
 \varpi\eqdef
 \Delta_e\bigl(\pi_e(\bfs^\T)+\mathbf m_e(\mathbf z)\bigr)
\]
remains a uniformizer after the unramified base change.

The Vandermonde product factors exactly as
\begin{align}
\label{r2:eq:local-vandermonde-factorization}
 \varpi
 &=D_Rs\prod_{i=1}^{e-2}(p+sr_i+r_i^2)\notag\\
 &=D_Rwu\prod_{i=1}^{e-2}(1+r_i^2u+wr_iu^e).
\end{align}
The final product is a unit with residue one.  Hence \(u\) and \(\varpi\) are
uniformizers, and~\eqref{r2:eq:local-vandermonde-factorization} gives
\begin{equation}
\label{r2:eq:escaping-asymptotics}
 p=\frac{\eta_{\bfs}D_R}{\varpi}(1+O(\varpi)),
 \qquad
 s=\frac{\varpi^{e-1}}
 {\eta_{\bfs}^{\,e-2}D_R^{\,e-1}}(1+O(\varpi)).
\end{equation}
The Newton polygon of \(z^2+sz+p\) has one segment from \((0,-1)\) to
\((2,0)\).  Its two roots have negative valuation and are exactly the two
escaping roots.

\subsection{Identification of the sign quadratic}

After the substitution \(z=x/s\), the quadratic interchanging the escaping
pair becomes
\[
 z^2+z=\frac{p}{s^2}.
\]
Let \(L/K\) be the global \(S_e\)-completion field, let
\(\tau\in S_e\) be the transposition of the escaping pair, and put
\[
 E\eqdef L^{\langle\tau\rangle},
 \qquad
 Q\eqdef L^{A_e}.
\]
The local chart orders the finite roots but leaves the escaping pair
unordered, so it describes \(E\), which is unramified at the chosen
valuation; the quadratic extension \(L/E\) orders that pair.  More precisely, the \'{e}tale
boundary chart shows that the inertia action fixes the finite roots and can
only interchange the escaping pair.  Hence the inertia group is contained
in \(\langle\tau\rangle\), and its fixed field \(E\) is unramified at the
chosen base valuation.  The odd Artin--Schreier pole computed below shows
that the containment is in fact an equality.  Since
\[
 \langle\tau\rangle\cap A_e=1,
 \qquad
 \langle\langle\tau\rangle,A_e\rangle=S_e,
\]
Galois correspondence gives \(EQ=L\) and \(E\cap Q=K\).  Therefore the
displayed local quadratic is exactly the base change of the global sign
field.

Substitution of~\eqref{r2:eq:escaping-asymptotics} yields
\begin{equation}
\label{r2:eq:sign-leading-pole}
 \frac{p}{s^2}
 =\eta_{\bfs}^{\,2e-3}D_R^{\,2e-1}\varpi^{-(2e-1)}
 +O\!\left(\varpi^{-(2e-2)}\right).
\end{equation}
The leading pole order \(2e-1\) is odd and its coefficient is nonzero.  An
Artin--Schreier coboundary cannot remove an odd leading pole
\cite[Proposition~3.7.8]{Stichtenoth2009}.  This proves
Lemma~\ref{r2:lem:generic-sign-ramification}.

\subsection{The pure highest-coordinate defect}
\label{r2:subsec:pure-highest-defect}

Let \(\bfs_\gamma^{\T}=(0,\ldots,0,\gamma)^{\T}\).  On the common Hankel
divisor, combine the \(e-1\) common roots with the \(e-2\) finite residual
roots and cancel repeated nonzero values in pairs.  The resulting binary
support has size at most \(2e-3\), and all its odd moments through degree
\(2e-3\) vanish.  Frobenius powers then give the vanishing of every moment
of degree \(1,\ldots,2e-2\).  If the reduced support had size \(v>0\), its
first \(v\) moment equations would form an invertible Vandermonde system in
its distinct nonzero locators, a contradiction.  The support is empty, so
its highest odd moment also vanishes.  It follows that
\[
 \eta_{\bfs_\gamma}=\gamma.
\]
Thus the coefficient in~\eqref{r2:eq:sign-leading-pole} depends on a pure
highest target through \(\gamma^{2e-3}\), exactly as used in
Proposition~\ref{r2:prop:pure-highest-plane}.

\section{A ten-locator curve for pure three-error spaces}
\label{app:pure-e3-exact-ten}

For locators \(y_1,\ldots,y_{10}\in F\) and labels
\(\epsilon_j=(\epsilon_{j,1},\epsilon_{j,2},\epsilon_{j,3})\in\F_2^3\),
\(1\leq j\leq10\), define, for \(1\leq k\leq3\), the \(k\)-th selector-row
syndrome by
\(\sum_{j=1}^{10}\epsilon_{j,k}h_3(y_j)\).  Thus the following proposition
asserts the existence of locators and labels that realize the three displayed
pure targets simultaneously.

\begin{proposition}[Ten-locator realization]
\label{prop:pure-e3-ten-locator}
Let \(F=\F_{2^m}\) with \(m\geq18\), and let
\(\gamma_1,\gamma_2,\gamma_3\in F\) be linearly independent over
\(\F_2\).  There are ten locator slots over \(F\) whose three binary
selector rows have syndromes
\[
 (0,0,\gamma_1)^\T,\qquad
 (0,0,\gamma_2)^\T,\qquad
 (0,0,\gamma_3)^\T.
\]
Consequently, their span has BCH-column support at most \(10\).
\end{proposition}

\begin{proof}
We first display the support pattern.  The ten slots and their selector
labels are
\[
\begin{array}{c|c|c}
 \text{slots} & \text{multiplicity} & \text{selector label}\\
 \hline
 u_1,u_2 & 2 & (1,1,0)\\
 v_1,v_2 & 2 & (0,1,1)\\
 w       & 1 & (1,0,1)\\
 a_1,a_2 & 2 & (1,0,0)\\
 r       & 1 & (0,1,0)\\
 c_1,c_2 & 2 & (0,0,1).
\end{array}
\]
Thus the three weight-five selector rows are
\[
 \mathcal A=\{u_1,u_2,w,a_1,a_2\},\quad
 \mathcal B=\{u_1,u_2,v_1,v_2,r\},\quad
 \mathcal C=\{v_1,v_2,w,c_1,c_2\}.
\]
Here sums of selector sets mean symmetric differences.  Their sum is the
fourth weight-five set
\[
 \{a_1,a_2,r,c_1,c_2\}.
\]
The other three nonzero words have weights \(6,6,8\).

For a five-element selector set \(S\), write
\(p_j(S)=\sum_{\xi\in S}\xi^j\) and
\(e_5(S)=\prod_{\xi\in S}\xi\), omitting \(S\) when it is clear.  Newton's
identities in characteristic two give
\begin{equation}
\label{eq:pure-five-product}
 p_1=p_3=0\quad\Longrightarrow\quad p_5=e_5,
\end{equation}
so the fifth moment is their product.  Put
\[
 u=u_1+u_2,\quad p_u=u_1u_2,\qquad
 v=v_1+v_2,\quad p_v=v_1v_2.
\]
The first, third, and fifth moment equations for \(\mathcal B\) are
\begin{equation}
\label{eq:pure-B-equations}
 r=u+v,\qquad up_u+vp_v=uvr,\qquad p_up_vr=\gamma_2.
\end{equation}
Indeed, \(u_1^3+u_2^3=u^3+up_u\) and
\(u^3+v^3+(u+v)^3=uv(u+v)\).

The private pairs in \(\mathcal A\) and \(\mathcal C\) have products
\[
 \pi_A\eqdef a_1a_2
 =\frac{u(p_u+uw+w^2)}{u+w},\qquad
 \pi_C\eqdef c_1c_2
 =\frac{v(p_v+vw+w^2)}{v+w},
\]
and sums \(u+w\) and \(v+w\), respectively.  Hence
\begin{equation}
\label{eq:pure-products}
 \gamma_1=p_uw\pi_A,\qquad
 \gamma_2=p_up_vr,\qquad
 \gamma_3=p_vw\pi_C.
\end{equation}
Put \(\gamma_4\eqdef\gamma_1+\gamma_2+\gamma_3\).  The fourth
weight-five word and~\eqref{eq:pure-five-product} give
\(\gamma_4=\pi_A r\pi_C\), and therefore
\[
 \frac{\gamma_1\gamma_2\gamma_3}{\gamma_4}=(p_up_vw)^2.
\]

We now parametrize the resulting curve.  Define
\begin{equation}
\label{eq:pure-constant-definitions}
 a\eqdef\frac{\gamma_1}{\gamma_2},\qquad
 b\eqdef\frac{\gamma_3}{\gamma_2},\qquad
 s\eqdef\frac{\gamma_1\gamma_3}{\gamma_2\gamma_4},\qquad
 \lambda\eqdef\sqrt{s}\in F.
\end{equation}
Frobenius makes the square root unique.  Independence gives
\[
 a,b\notin\{0,1\},\quad a+1+b\neq0,\quad
 s\notin\{0,1,a\},\quad a+s\neq0,
\]
and
\begin{equation}
\label{eq:pure-b-relation}
 b=\frac{s(a+1)}{a+s}.
\end{equation}
For a free parameter \(x\), set
\begin{equation}
\label{eq:pure-D-p-q}
\begin{split}
 D&\eqdef(a+\lambda)x+\lambda(a+1),\\
 f&\eqdef\frac{(a+s)(x+1)(x+\lambda)}{D},\\
 g&\eqdef\frac{\lambda(\lambda+1)x(x+\lambda+1)}{D}.
\end{split}
\end{equation}
Adjoin \(r\) subject to
\begin{equation}
\label{eq:pure-Kummer}
 r^5=\frac{\gamma_2}{fg},
\end{equation}
and put
\begin{equation}
\label{eq:pure-normalized-core}
 u=xr,\quad v=(1+x)r,\quad w=\lambda r,\quad
 p_u=fr^2,\quad p_v=gr^2.
\end{equation}
Then~\eqref{eq:pure-B-equations} reduces to
\[
 xf+(1+x)g=x(1+x),\qquad fgr^5=\gamma_2.
\]
Moreover,
\[
 a=\frac{x\lambda(f+x\lambda+\lambda^2)}{g(x+\lambda)}.
\]
Solving these two normalized equations gives exactly
\eqref{eq:pure-D-p-q}.  Finally,
\eqref{eq:pure-b-relation} gives the identity
\[
 b=\frac{(1+x)\lambda
       (g+(1+x)\lambda+\lambda^2)}{f(1+x+\lambda)},
\]
which is the normalized \(\mathcal C\)-product equation.

It remains to split four quadratic pairs simultaneously.  A pair with
nonzero sum \(s_0\) and product \(p_0\) can be written as
\(s_0y,s_0(y+1)\); it splits precisely when
\(y^2+y=p_0/s_0^2\) has a solution.  Put \(K_0=F(x)\) and
\(\overline K=\overline F(x)\).  The four resulting Artin--Schreier
equations are defined over \(K_0\); after extension of
constants, their classes in
\(\overline K/\wp(\overline K)\), where
\(\wp(z)=z^2+z\), are
\begin{equation}
\label{eq:pure-AS-classes}
\begin{split}
 \alpha_u&=\frac{f}{x^2},\\
 \alpha_v&=\frac{g}{(x+1)^2},\\
 \alpha_A&=\frac{x(f+x\lambda+s)}{(x+\lambda)^3}
          =\frac{a}{\lambda}\frac{g}{(x+\lambda)^2},\\
 \alpha_C&=\frac{(1+x)(g+(1+x)\lambda+s)}{(x+\lambda+1)^3}
          =\frac{b}{\lambda}\frac{f}{(x+\lambda+1)^2}.
\end{split}
\end{equation}
If
\(\alpha=c_{-2}\tau^{-2}+c_{-1}\tau^{-1}+O(1)\), subtracting
\(\wp(\sqrt{c_{-2}}\tau^{-1})\) leaves simple-pole coefficient
\(\rho=c_{-1}+\sqrt{c_{-2}}\), so
\(\rho^2=c_{-1}^2+c_{-2}\).  Direct simplification gives
\begin{equation}
\label{eq:pure-residue-table}
\begin{array}{c|c|c}
 \text{class} & \text{private point} & \rho^2\\
 \hline
 \alpha_u & 0 &
 \dfrac{a(s+1)(s+a)(as+1)}{s(a+1)^4}\\[6pt]
 \alpha_v & 1 &
 \dfrac{s(a+1)(s+a)((a+1)s+a)}{a^4(s+1)}\\[6pt]
 \alpha_A & \lambda &
 \dfrac{a(a+1)(s+a)(s+a+1)}{s(s+1)}\\[6pt]
 \alpha_C & \lambda+1 &
 \dfrac{sa(a+1)(s+1)(s^2+a)}{(s+a)^4}.
\end{array}
\end{equation}
Thus the four private residues can vanish only under, respectively,
\begin{equation}
\label{eq:pure-exceptional-values}
 s=\frac1a,\qquad
 s=\frac a{a+1},\qquad
 s=a+1,\qquad
 s^2=a.
\end{equation}
Any two conditions in~\eqref{eq:pure-exceptional-values} imply
\(a^2+a+1=0\).  Their common value is then \(s=a+1=a^2\), and
\eqref{eq:pure-b-relation} gives \(b=a\), contrary to independence.
Hence at most one private residue vanishes.

The polynomial \(D\) is nonconstant: \(a+\lambda=0\) would imply
\(s=a^2\) and then \(b=a\).  Its root is
\[
 \xi_D\eqdef\frac{\lambda(a+1)}{a+\lambda}.
\]
Also
\[
 D(0)=\lambda(a+1),\quad
 D(1)=a(\lambda+1),\quad
 D(\lambda)=\lambda(\lambda+1),\quad
 D(\lambda+1)=s+a,
\]
so \(\xi_D\) is distinct from the four private points.  Every class in
\eqref{eq:pure-AS-classes} has a simple pole at \(\xi_D\).  In a binary
relation among the classes, the unique private odd pole forces each
nonexceptional coefficient to vanish.  If the one possible exceptional class
remains, its simple pole at \(\xi_D\) forces its coefficient to vanish as well.
The four classes are therefore linearly
independent in \(\overline K/\wp(\overline K)\).

For the effective point count, we use the standard Kummer,
Artin--Schreier, conductor--discriminant, and Hasse--Weil facts in
\cite[Chaps.~III and V]{Stichtenoth2009}.  Observe that
\begin{equation}
\label{eq:pure-Kummer-divisor}
 \frac{\gamma_2}{fg}
 =\text{(nonzero constant)}\,
 \frac{D^2}{x(x+1)(x+\lambda)(x+\lambda+1)}.
\end{equation}
This function has simple poles at \(0,1,\lambda,\lambda+1\) and double
zeros at \(\xi_D\) and \(\infty\).  Let \(L_0=K_0(r)\) and
\(\overline L=\overline K(r)\).  The displayed function is not a fifth
power even over \(\overline K\).  Thus
\(\overline L/\overline K\) has degree \(5\), is totally and tamely
ramified at these six points, and has genus \(8\):
\[
 2\operatorname{genus}(\overline L)-2=5(-2)+6(5-1)=14.
\]
The four Artin--Schreier classes remain independent over
\(\overline L\): otherwise the quadratic extension of \(\overline K\)
defined by a nonzero binary combination would embed in the odd-degree
extension \(\overline L/\overline K\).  Let \(M_0/L_0\) be the
compositum defined by the four equations \(y_i^2+y_i=\alpha_i\), and put
\(\overline M=M_0\overline F\).  The preceding independence shows that
\(\overline M/\overline L\) has group
\((\mathbb Z/2\mathbb Z)^4\); in particular, \(M_0\) is geometrically
integral.

After reduction over \(\overline K\), every nonzero character has only
simple poles among the four private points and \(\xi_D\).  Pullback to
\(\overline L\) changes each such pole to the reduced odd order \(5\),
hence conductor exponent \(6\).  At a fixed private point, at most the
eight characters containing the corresponding basis class have a pole;
at \(\xi_D\), at most all fifteen nontrivial characters do.  The total
number of character--pole incidences is therefore at most
\[
 4\cdot8+15=47.
\]
The conductor--discriminant formula for
\(\overline M/\overline L\) and Riemann--Hurwitz give
\[
 \deg\operatorname{Diff}(\overline M/\overline L)
 \leq6\cdot47=282,
\]
\[
 2\operatorname{genus}(\overline M)-2\leq16(2\cdot8-2)+282=506,
 \qquad \operatorname{genus}(\overline M)\leq254.
\]

Let \(\widetilde{\mathcal C}\) be the smooth projective \(F\)-curve with
function field \(M_0\).  It is geometrically integral
by the preceding paragraph.  The only base values excluded by the formulas
are
\[
 0,1,\lambda,\lambda+1,\xi_D,\infty.
\]
Indeed, the first four are the zeros of \(f\) or \(g\) at which one of the
four pair sums \(xr,(1+x)r,(x+\lambda)r,(1+x+\lambda)r\) vanishes;
\(\xi_D\) is the zero of \(D\), hence a pole of \(f\) and \(g\), and
\(\infty\) is their remaining pole.  Thus the list accounts for every zero
or pole used in the reconstruction formulas.
Each is totally ramified in the Kummer step, so there is one point of
\(L_0\) above it and at most \(16\) points of \(M_0\).  Thus at most \(96\)
\(F\)-points are excluded.  Hasse--Weil gives
\[
 \#\widetilde{\mathcal C}(F)\geq |F|+1-508\sqrt{|F|}.
\]
For \(|F|\geq2^{18}\), the right-hand side is at least \(2049>96\).
There is an \(F\)-rational point away from the excluded fibers.

At such a point, choose \(y_u,y_v,y_A,y_C\in F\) satisfying
\(y_i^2+y_i=\alpha_i\), and recover
\begin{align*}
 u_1&=uy_u,&u_2&=u(y_u+1),\\
 v_1&=vy_v,&v_2&=v(y_v+1),\\
 a_1&=(u+w)y_A,&a_2&=(u+w)(y_A+1),\\
 c_1&=(v+w)y_C,&c_2&=(v+w)(y_C+1).
\end{align*}
Together with \(w=\lambda r\) and the central locator \(r\), these are the
ten labeled slots above.  Their pair sums and products are
\(u,p_u,v,p_v,u+w,\pi_A,v+w,\pi_C\), respectively, so
\eqref{eq:pure-B-equations}--\eqref{eq:pure-products} verify all three
target syndromes.  Any zero or cross-slot coincidence only decreases the
reduced binary support.
\end{proof}

\section{Pure sign calculations for at most six errors}
\label{app:pure-sign-independence}

This appendix proves Lemma~\ref{lem:pure-sign-independence-small-e}.  The
proof computes the Berlekamp sign classes of pure completion polynomials.
These coefficient calculations are necessary for the low-error refinement,
but the main common-core argument uses only their independence conclusion.
Recall that
\[
 K=\overline F(z_1,\ldots,z_{e-1}),
 \qquad
 \wp(K)=\{u^2+u:u\in K\},
\]
and that \([a_{\bfs_\gamma}]\in K/\wp(K)\) denotes the quadratic sign
class associated with the pure target
\(\bfs_\gamma^\T=(0,\ldots,0,\gamma)^\T\).

\begin{proof}[Proof of Lemma~\ref{lem:pure-sign-independence-small-e}
for \(2\leq e\leq4\)]
For a separable monic polynomial \(f=\prod_i(t-r_i)\) in characteristic two,
write
\[
 \operatorname{Berl}(f)\eqdef
 \sum_{i<j}\frac{r_ir_j}{(r_i+r_j)^2}.
\]
Berlekamp's criterion~\cite{Berlekamp1976Discriminant} identifies its class
modulo \(\wp(K)\) with the quadratic sign subfield of the splitting field.
For \(e=2\), Proposition~\ref{prop:sign-input}(2) already proves the claim:
the leading coefficient of a nonempty sum is
\(\sum_i\gamma_i\neq0\).

Let \(e=3\).  Put
\[
 a=z_1+z_2,\qquad b=z_1z_2.
\]
On the dense open set \(ab\neq0\), Newton identities give the pure-target
completion polynomial
\[
 f_\gamma(t)=t^3+at^2+
 \left(b+\frac{\gamma}{ab}\right)t+\frac{\gamma}{b}.
\]
For \(f=t^3+\alpha_1t^2+\alpha_2t+\alpha_3\), direct expansion of the Berlekamp invariant gives
\[
 \operatorname{Berl}(f)=
 \frac{\alpha_3^2+\alpha_1\alpha_2\alpha_3+\alpha_1^3\alpha_3+\alpha_2^3}
 {(\alpha_3+\alpha_1\alpha_2)^2}.
\]
Substitution yields the following representative of the sign class:
\begin{equation}
\label{eq:cubic-pure-berlekamp}
 \beta_\gamma=
 \frac b{a^2}+\frac{\gamma}{ab^2}+\frac{a\gamma}{b^3}
 +\frac{\gamma}{a^3b}+\frac{\gamma^2}{a^4b^3}
 +\frac{\gamma^3}{a^5b^5}.
\end{equation}
Although derived on \(ab\neq0\), this is an identity in \(K\), so we may
evaluate it at boundary valuations.  For a nonempty set
\(I\subseteq\{1,\ldots,u\}\), put
\[
 \delta\eqdef\sum_{i\in I}\gamma_i\neq0,
 \qquad S_3\eqdef\sum_{i\in I}\gamma_i^3.
\]
At the generic divisor \(z_1=0\), the element \(b\) is a uniformizer and
\(a\) is transcendental in the residue field.  If \(S_3\neq0\), the sum of
the representatives in~\eqref{eq:cubic-pure-berlekamp} has an odd leading
pole of order five.  If \(S_3=0\), its leading pole has order three and
coefficient
\[
 \frac{\delta(a^5+\delta)}{a^4}\neq0
 \quad\text{in}\quad\overline F(a).
\]
An Artin--Schreier coboundary with a pole has even pole order.  Hence no
nonempty sum of the cubic sign classes lies in \(\wp(K)\).

Let \(e=4\), and let \(a,b,c\) be the elementary symmetric functions of
\(z_1,z_2,z_3\).  Put
\[
 H\eqdef ab+c=(z_1+z_2)(z_1+z_3)(z_2+z_3).
\]
On the dense open set \(cH\neq0\), Newton identities give
\begin{equation}
\label{eq:quartic-pure-completion}
 f_\gamma(t)=t^4+at^3+(b+u_\gamma)t^2
 +(c+au_\gamma)t+\frac{\gamma}{c},
 \qquad
 u_\gamma\eqdef\frac{a\gamma}{cH}.
\end{equation}
For \(f=t^4+\alpha_1t^3+\alpha_2t^2+\alpha_3t+\alpha_4\), direct symmetric expansion of the Berlekamp
invariant gives
\[
 \operatorname{Berl}(f)=
 \frac{n(\alpha_1,\alpha_2,\alpha_3,\alpha_4)}
 {(\alpha_1\alpha_2\alpha_3+\alpha_1^2\alpha_4+\alpha_3^2)^2},
\]
where
\[
 \begin{split}
 n={}&\alpha_1^2\alpha_2^3\alpha_4+\alpha_1^3\alpha_3^3
       +\alpha_1^3\alpha_2\alpha_3\alpha_4+\alpha_1^4\alpha_4^2\\
    &{}+\alpha_2^3\alpha_3^2+\alpha_1\alpha_2\alpha_3^3+\alpha_3^4.
 \end{split}
\]
After substituting~\eqref{eq:quartic-pure-completion}, the denominator is
\(c^2H^2\), and, writing \(\theta=u_\gamma\),
\[
 n=n_{\leq2}+(a^6+cH)\theta^3+ac\theta^4+a^2\theta^5,
 \qquad \deg_\theta n_{\leq2}\leq2.
\]
Consequently,
\begin{equation}
\label{eq:quartic-pure-berlekamp-leading}
 \beta_\gamma=
 \frac{a^7\gamma^5}{c^7H^7}
 +\frac{a^5\gamma^4}{c^5H^6}
 +\left(\frac{a^9}{c^5H^5}+\frac{a^3}{c^4H^4}\right)\gamma^3
 +O(H^{-4}),
\end{equation}
where the last term has pole order at most four along \(H=0\) with \(a,c\)
units.

For a nonempty set \(I\subseteq\{1,\ldots,u\}\), again put
\[
 \delta=\sum_{i\in I}\gamma_i\neq0,
 \qquad S_j=\sum_{i\in I}\gamma_i^j.
\]
If \(S_5\neq0\), equation~\eqref{eq:quartic-pure-berlekamp-leading} has an
odd leading pole of order seven.  Suppose that \(S_5=0\).  Then
\(S_4=\delta^4\neq0\).  At the component \(z_1+z_2=0\), use local
coordinates
\[
 z_1=x,\qquad z_2=x+\tau,\qquad z_3=v,
 \qquad w=x+v,
\]
and localize further at \(xvw\neq0\).  Thus
\[
 a=v+\tau,\qquad c=xv(x+\tau),\qquad H=\tau w(w+\tau).
\]
Expanding~\eqref{eq:quartic-pure-berlekamp-leading} gives
\[
 \sum_{i\in I}\beta_{\gamma_i}
 =\ell_6\tau^{-6}+\ell_5\tau^{-5}+O(\tau^{-4}),
\]
where
\[
 \ell_6=\frac{\delta^4}{x^{10}w^{12}}
 =\left(\frac{\delta^2}{x^5w^6}\right)^2,
 \qquad
 \ell_5=\frac{\delta^4+S_3xv^5w}{x^{11}vw^{11}}\neq0.
\]
Indeed, \(xv^5(x+v)\) is nonconstant in the algebraically independent
residues \(x,v\), whereas \(\delta^4\neq0\).  With
\(\rho=\delta^2/(x^5w^6)\), subtracting
\[
 \wp(\rho\tau^{-3})=\rho^2\tau^{-6}+\rho\tau^{-3}
\]
cancels the pole of order six without changing the coefficient of
\(\tau^{-5}\).  The reduced class therefore has an odd leading pole of order
five.  Thus no nonempty sum lies in \(\wp(K)\).  This proves the cases
\(2\leq e\leq4\).
\end{proof}

\begin{proof}[Proof of Lemma~\ref{lem:pure-sign-independence-small-e}
for \(e=5,6\)]
We first find an odd pole for every nonempty binary sum after specializing
the core coefficients. We then use this calculation to rule out an
Artin--Schreier coboundary over the ordered-core field.
Write the generic core polynomial and its reciprocal as
\[
 \Phi(y)=\prod_{i=1}^{e-1}(y+z_i)
 =y^{e-1}+b_1y^{e-2}+\cdots+b_{e-1},
 \qquad
 \widehat\Phi(t)=t^{e-1}\Phi(t^{-1}).
\]
There is a unique polynomial
\(\Psi(t)=\sum_{j=1}^{e}d_jt^j\) over
\(\overline F(b_1,\ldots,b_{e-1})\) satisfying
\begin{equation}
\label{eq:pure-completion-reciprocal}
 (\widehat\Phi\Psi)'=t^{2e-2}.
\end{equation}
Put
\[
 \Psi^*(y)=y^e\Psi(y^{-1}),
 \qquad
 f_\gamma(y)=y\Phi(y)+\gamma\Psi^*(y).
\]
The reciprocal polynomial of \(\Phi f_\gamma\) is
\(\widehat\Phi^2+\gamma\widehat\Phi\Psi\), whose derivative is
\(\gamma t^{2e-2}\).  Newton's identities therefore show that
\(f_\gamma\) is the completion polynomial of the pure target
\(\bfs_\gamma^\T\).

We use an exact coefficient form of the Berlekamp invariant.  For the
universal monic polynomial
\(f_0(y)=y^e+a_1y^{e-1}+\cdots+a_e\) over \(\mathbb Z[a_1,\ldots,a_e]\),
put
\[
 \Delta_+(f_0)=
 \frac{\operatorname {Res}_y(f_0(y),(-1)^ef_0(-y))}
      {2^e(-1)^ea_e},
 \qquad
 \Xi(f_0)=\frac{\Delta_+(f_0)-\operatorname {Disc}(f_0)}4.
\]
These are integral polynomials
\cite[Section~2.2]{Carmon2015ChowlaCharTwo}.  Reducing their coefficients
modulo two and substituting those of \(f_\gamma\) gives
\begin{equation}
\label{eq:universal-berlekamp-quotient}
 \beta_\gamma=
 \frac{\overline\Xi(f_\gamma)}{\overline\Delta_+(f_\gamma)}.
\end{equation}
It equals Berlekamp's symmetric root expression used above and hence
represents the sign class.

We apply~\eqref{eq:universal-berlekamp-quotient} at the generic divisor
\(b_{e-1}=0\).  To prove that the required Laurent coefficients in its
residue field are nonzero, specialize the remaining coefficients as follows.
For \(e=5\), take
\[
 \Phi(y)=y^4+xy^3+y^2+y+\tau,
 \qquad q_5=x^2\tau+x+1.
\]
Equation~\eqref{eq:pure-completion-reciprocal} gives
\[
 d_1=0,\quad d_2=\frac{x+1}{\tau q_5},\quad
 d_3=\frac{x(x+1)}{\tau q_5},\quad
 d_4=\frac{x+x\tau+1}{\tau q_5},\quad d_5=\frac1\tau.
\]
For \(e=6\), take
\[
 \Phi(y)=y^5+xy^4+y^3+y^2+y+\tau,
 \quad r_6=x^2+x+1,
 \quad q_6=r_6+(x+1)\tau+\tau^2,
\]
and obtain
\[
 \begin{split}
 d_1&=0,\qquad d_2=\frac{r_6+x\tau}{\tau q_6},\qquad
 d_3=\frac{xr_6+x^2\tau}{\tau q_6},\\
 d_4&=\frac{r_6+\tau}{\tau q_6},\qquad
 d_5=\frac{r_6+x^2\tau}{\tau q_6},\qquad d_6=\frac1\tau.
 \end{split}
\]

Write
\[
 \beta_\gamma=\sum_{j=0}^{2e-3}c_{e,j}(x,\tau)\gamma^j,
 \qquad c_{e,j}=0\quad\text{outside }0\leq j\leq2e-3.
\]
The coefficient \(c_{e,0}\) is regular at \(\tau=0\) and is therefore omitted
from the polar table.  For \(1\leq j\leq2e-3\), define \(\nu_{e,j}\) and
\(\lambda_{e,j}\) by
\[
 c_{e,j}(x,\tau)=\lambda_{e,j}(x)\tau^{-\nu_{e,j}}
 +O(\tau^{-\nu_{e,j}+1}).
\]
Exact expansion of~\eqref{eq:universal-berlekamp-quotient} gives the table
below.  The factors \(u_5=(x+1)^9\) and \(u_6=r_6^{11}\) clear residue-field
units.
\[
\begin{array}{c|c|c|c@{\qquad}c|c|c}
 e&j&\nu_{e,j}&u_e\lambda_{e,j}&j&\nu_{e,j}&u_e\lambda_{e,j}\\ \hline
5&7&9&u_5&6&7&x(x+1)^8\\
 &5&7&x^4u_5&4&5&(x^4+x^3+x^2+x+1)u_5\\
 &3&5&u_5&2&3&x(x+1)^8(x^2+x+1)\\
 &1&3&u_5&&&\\ \hline
6&9&11&u_6&8&9&(x+1)r_6^{10}\\
 &7&9&x^4u_6&6&7&(x^3+x+1)u_6\\
 &5&7&u_6&4&5&(x^4+x^3+x^2+x+1)r_6^{10}\\
 &3&5&u_6&2&3&x^2(x+1)r_6^{10}\\
 &1&3&u_6&&&
\end{array}
\]
An exact-arithmetic verification of the completion identities and every
entry of this table is archived in the accompanying Zenodo
record~\cite{BelinskyYohananov2026LaurentVerification}.

Let \(\varnothing\neq I\subseteq\{1,\ldots,u\}\), and put
\[
 \delta=\sum_{i\in I}\gamma_i\neq0,
 \qquad S_j=\sum_{i\in I}\gamma_i^j.
\]
For \(e=5\), if \(S_7\neq0\), the sum of the sign representatives has an
odd leading pole of order nine.  If \(S_7=0\), its coefficient at order
 seven, after division by \(u_5\), is
\[
 \frac{x}{x+1}S_6+x^4S_5.
\]
This rational function vanishes only when \(S_6=S_5=0\).  Then
\(S_3^2=S_6=0\), while \(S_4=\delta^4\neq0\), so the coefficient at order
five is \((x^4+x^3+x^2+x+1)\delta^4\neq0\).

For \(e=6\), the pole order is eleven when \(S_9\neq0\).  If \(S_9=0\),
 the coefficient at order nine, after division by \(u_6\), is
\[
 \frac{x+1}{r_6}S_8+x^4S_7.
\]
Here \(S_8=\delta^8\neq0\), and the two coefficient functions are not
proportional over the constant field.  Thus this rational function is nonzero.  Every
nonempty sum therefore retains an odd pole.

It remains to transport this calculation to the ordered-core field
\(K=\overline F(z_1,\ldots,z_{e-1})\).  Write \(k=\overline F\),
\(x=b_1\), and \(\tau=b_{e-1}\), and set
\[
 A=k(x,\tau)[b_2,\ldots,b_{e-2}]_
 {(b_2-1,\ldots,b_{e-2}-1)}.
\]
Its fraction field is the symmetric coefficient field and its residue
field is \(k(x,\tau)\).  Let \(\phi_e\) denote the displayed specialization
of \(\Phi\).  At \(\tau=0\), its discriminant is
\[
 \operatorname{Disc}(\phi_5)\big|_{\tau=0}=(x+1)^2,
 \qquad
 \operatorname{Disc}(\phi_6)\big|_{\tau=0}=(x^2+x+1)^2.
\]
Indeed, the residual polynomials are respectively
\(y(y^3+xy^2+y+1)\) and
\(y(y^4+xy^3+y^2+y+1)\); the nonzero factor has constant term one
and the displayed discriminant.  In particular,
\(\operatorname{Disc}(\Phi)\) is an \(A\)-unit.

Write \(\widehat\phi_e\) for the reciprocal of \(\phi_e\), obtained by
specializing \(\widehat\Phi\). The linear system~\eqref{eq:pure-completion-reciprocal}
is nonsingular after this
specialization.  To see this, a solution \(\Theta\) of the homogeneous
system would satisfy \((\widehat\phi_e\Theta)'=0\),
\(\deg\Theta\leq e\), and \(\Theta(0)=0\).  Since
\(\widehat\phi_e\) is separable with constant term one, it divides
\(\Theta\); the quotient has degree at most one and its derivative must
vanish.  Its constant term must also vanish, so \(\Theta=0\).
Consequently the coefficients of \(\Psi\), and hence of every
\(f_\gamma\), belong to \(A\).
For \(\gamma\neq0\), the specialized \(f_\gamma\) is separable:
the reciprocal derivative identity, with \(P=\phi_e f_\gamma\), gives
\(P+yP'=\gamma\).  At a root \(r\) of \(f_\gamma\), this reads
\(r\phi_e(r)f_\gamma'(r)=\gamma\neq0\).
Thus the Berlekamp denominator is also an \(A\)-unit, and
\(\beta=\sum_{i\in I}\beta_{\gamma_i}\) belongs to \(A\).

Let \(B\) be the integral closure of \(A\) in \(K\).
Because \(\operatorname{Disc}(\Phi)\) is invertible, the ordered-root
algebra is finite \'{e}tale over \(A\); its component with fraction
field \(K\) is \(B\).  In particular, \(B\) is normal and finite
\'{e}tale.  Choose a maximal ideal of \(B\) above that of \(A\).
Its residue field is the splitting field \(E\) of \(\phi_e\) over
\(k(x,\tau)\), since the specialized ordered roots are distinct.
If \(\beta=w^2+w\) for some \(w\in K\), the monic equation makes
\(w\) integral over \(B\); normality gives \(w\in B\).
Reduction would therefore make the displayed specialized sum an
Artin--Schreier coboundary in \(E\).

Finally, the two nonzero discriminants at \(\tau=0\) show that
\(E/k(x,\tau)\) is unramified at that valuation: the splitting algebra
is finite \'{e}tale over \(k(x)[\tau]_{(\tau)}\).
Every prolongation has ramification index one, so the nonzero odd leading
pole established above for the specialized sum retains its odd order in
\(E\).  This is impossible for an Artin--Schreier coboundary.
The argument applies to every nonempty \(I\), because binary independence
of the given parameters gives \(\delta=\sum_{i\in I}\gamma_i\neq0\).
Thus no nonempty binary sum lies in \(\wp(K)\), completing the proof.
\end{proof}

\enlargethispage{6\baselineskip}
\bibliographystyle{plain}
\bibliography{references}

\end{document}